\documentclass[journal]{IEEEtran}

\usepackage{balance}
\usepackage{subcaption}
\usepackage{amsthm}
\usepackage{times,amssymb,amsmath,amsfonts,float,nicefrac,color,bbm,mathrsfs,float}

\usepackage{enumerate,multirow,tikz,graphicx}

\usepackage{algpseudocode}
\usepackage{stackengine}
\usepackage[noadjust]{cite}

\usepackage{diagbox}
\usepackage[ruled,vlined,linesnumbered]{algorithm2e}

\newcommand{\RN}[1]{%
\textup{\expandafter{\romannumeral#1}}%
}

\usetikzlibrary{shapes.misc, positioning}
\usetikzlibrary{decorations.pathreplacing}

\tikzset{
    block/.style = {draw, thick, rectangle, minimum  width = 1em},
    sblock/.style = {draw, thick, rectangle, minimum height = 6em, minimum width = 6em},
    lgblock/.style = {draw, thick, rectangle, minimum height = 9em, minimum width = 6em},
    vblock/.style = {draw, thick, rectangle, minimum height = 3.7em, minimum width = 1.8em},
    ablock/.style = {draw, thick, rectangle, minimum height = 4em, minimum width = 4em},
    eblock/.style = {draw, thick, rectangle, minimum height = 5em, minimum width = 5em},
}

\tikzset{XOR/.style={draw,circle,append after command={
    [shorten >=\pgflinewidth, shorten <=\pgflinewidth,]
    (\tikzlastnode.north) edge (\tikzlastnode.south)
    (\tikzlastnode.east) edge (\tikzlastnode.west)
}}}

\newcommand\remove[1]{}

\newtheorem{proposition}{Proposition}

\newtheorem{lemma}{Lemma}

\newtheorem{cnstr}{Construction}

\def\mathbi#1{{\textbf{\textit #1}}}

\newcommand{\bP}{\mathbb{P}}

\newcommand{\cA}{\mathcal{A}}

\newcommand{\cC}{\mathcal{C}}

\newcommand{\cI}{\mathcal{I}}

\newcommand{\cM}{\mathcal{M}}

\newcommand{\cS}{\mathcal{S}}

\newcommand{\cY}{\mathcal{Y}}

\makeatletter
\DeclareRobustCommand*\triangledowndot{%
\ifmmode
{\mathpalette\triangleda@t\triangleda@@t}%
\else
\triangleda@t\@empty\triangleda@@t
\fi
}
\def\triangleda@t#1#2{#2{#1}}
\def\triangleda@@t#1{%
\setbox0=\hbox{$\m@th#1\bigtriangledown$}
\ooalign{\hbox to \wd0{\hss$\m@th#1 {\cdot}$\hss}\crcr\box0\crcr}
}
\makeatother

\makeatletter
\DeclareRobustCommand*\lozengedot{%
\ifmmode
{\mathpalette\triangledb@t\triangledb@@t}%
\else
\triangledb@t\@empty\triangledb@@t
\fi
}
\def\triangledb@t#1#2{#2{#1}}
\def\triangledb@@t#1{%
\setbox0=\hbox{$\m@th#1\lozenge$}
\ooalign{\hbox to \wd0{\hss$\m@th#1 \cdot$\hss}\crcr\box0\crcr}
}
\makeatother

\makeatletter
\DeclareRobustCommand*\vartriangledot{%
\ifmmode
{\mathpalette\triangledc@t\triangledc@@t}%
\else
\triangledc@t\@empty\triangledc@@t
\fi
}
\def\triangledc@t#1#2{#2{#1}}
\def\triangledc@@t#1{%
\setbox0=\hbox{$\m@th#1\triangle$}
\ooalign{\hbox to \wd0{\hss$\m@th#1 \cdot$\hss}\crcr\box0\crcr}
}
\makeatother

\newcommand{\oria}{\triangledown}
\newcommand{\orib}{\lozenge}
\newcommand{\oric}{\vartriangle}
\newcommand{\swpa}{\blacktriangledown}
\newcommand{\swpb}{\blacklozenge}
\newcommand{\swpc}{\blacktriangle}
\newcommand{\adda}{\triangledowndot}
\newcommand{\addb}{\lozengedot}
\newcommand{\addc}{\vartriangledot}

\DeclareMathOperator{\polar}{polar}

\DeclareMathOperator{\ABS}{ABS}
\DeclareMathOperator{\ABSP}{ABS+}
\DeclareMathOperator{\dB}{dB}

\DeclareMathOperator{\argmax}{argmax}
\DeclareMathOperator{\mode}{mode}

\newcommand{\mG}{\mathbf{G}}
\newcommand{\mP}{\mathbf{P}}
\newcommand{\mQ}{\mathbf{Q}}
\newcommand{\mI}{\mathbf{I}}
\newcommand{\mS}{\mathbf{S}}
\newcommand{\mA}{\mathbf{A}}
\newcommand{\mD}{\mathbf{D}}
\newcommand{\mE}{\mathbf{E}}
\newcommand{\mK}{\mathbf{K}}
\newcommand{\mM}{\mathbf{M}}

\newcommand{\tP}{\mathtt{P}}
\newcommand{\tB}{\mathtt{B}}
\newcommand{\tH}{\mathtt{H}}

\title{ABS+ Polar Codes: Exploiting More Linear Transforms on Adjacent Bits}

\author{
Guodong~Li,
Min~Ye,
and~
Sihuang~Hu,~\IEEEmembership{Member,~IEEE}

\thanks{
Research partially funded by
National Key R\&D Program of China under Grant No. 2021YFA1001000,
National Natural Science Foundation of China under Grant No. 12001322 and 12231014,
Shandong Provincial Natural Science Foundation under Grant No. ZR202010220025,
and a Taishan scholar program of Shandong Province.
An early version of this paper was presented in part
at the 2023 IEEE International Symposium on Information Theory
[DOI: 10.1109/ISIT54713.2023.10206934].
}
\thanks{
Guodong Li and Sihuang Hu are with Key Laboratory of Cryptologic Technology and Information Security, Ministry of Education, Shandong University, Qingdao, Shandong, 266237, China and School of Cyber Science and Technology, Shandong University, Qingdao, Shandong, 266237, China.
S. Hu is also with Quan Cheng Laboratory, Jinan 250103, China.
Email: guodongli@mail.sdu.edu.cn, yeemmi@gmail.com, husihuang@sdu.edu.cn
}
}

\begin{document}


\maketitle

\begin{abstract}
    ABS polar codes were recently proposed to speed up polarization by swapping certain pairs of adjacent bits after each layer of polar transform. In this paper, we observe that applying the Ar{\i}kan transform $(U_i, U_{i+1}) \mapsto (U_{i}+U_{i+1}, U_{i+1})$ on certain pairs of adjacent bits after each polar transform layer leads to even faster polarization.

    In light of this, we propose ABS+ polar codes which incorporate the Ar{\i}kan transform in addition to the swapping transform in ABS polar codes. In order to efficiently construct and decode ABS+ polar codes, we derive a new recursive relation between the joint distributions of adjacent bits through different layers of polar transforms. Simulation results over a wide range of parameters show that the CRC-aided SCL decoder of ABS+ polar codes improves upon that of ABS polar codes by $0.1\dB$--$0.25\dB$ while maintaining the same decoding time. Moreover, ABS+ polar codes improve upon standard polar codes by $0.2\dB$--$0.45\dB$ when they both use the CRC-aided SCL decoder with list size $32$. The implementations of all the algorithms in this paper are available at  \texttt{https://github.com/PlumJelly/ABS-Polar}

\end{abstract}


\section{Introduction} \label{sect:introduction}

\IEEEPARstart{P}{olar} code is the first code family that provably achieves the capacity for all binary-input memoryless symmetric (BMS) channels \cite{Arikan09}. In his original paper \cite{Arikan09}, Ar{\i}kan proposed the successive cancellation (SC) decoder and proved that polar codes achieve the capacity of BMS channels under the SC decoder. Later, successive cancellation list (SCL) decoder and CRC-aided SCL decoder \cite{Tal15},\cite{Niu12} were proposed to further reduce the decoding error probability for polar codes.

Although polar codes are able to attain the optimal code rate when the code length approaches infinity, the finite-length scaling of polar codes is far from optimal \cite{Hassani14, Guruswami15, Mondelli15, Mondelli16}.
An extensively-studied method to improve the finite-length performance of polar codes is to replace the $2\times 2$ Ar{\i}kan kernel with large kernels \cite{Buz2017ISIT, Buz2017, Ye2015, Fazeli21, Wang21, Guruswami22}. 
Recently, a window processing based algorithm was proposed to reduce the decoding complexity of polar codes with large kernels \cite{wpker}, which shows that polar codes with large kernel can be decoded with low decoding complexity.
In addition, convolutional polar codes \cite{convolutional} have also been shown to enhance the process of polarization.
In a previous paper \cite{Li2022TIT}, we proposed the Adjacent-Bits-Swapped (ABS) polar codes which polarize faster than standard polar codes and demonstrate better performance under the CRC-aided SCL decoder.
The ABS polar code construction is another way to improve the finite-length performance of polar codes, and its advantages over the large kernel method were discussed in Section~II-D of \cite{Li2022TIT}. In this paper, we propose a new family of codes called ABS+ polar codes, which further improve upon ABS polar codes in terms of the polarization speed and the decoding error probabilities.

Standard polar code construction consists of multiple consecutive layers of polar transforms.
In the ABS polar code construction, we swap certain pairs of adjacent bits after each layer of polar transform to speed up polarization.
Swapping two adjacent bits $U_i$ and $U_{i+1}$ can be written as the linear transform $(U_i, U_{i+1})\mapsto (U_{i+1}, U_i)$.
In this paper, we find that applying the $2\times 2$ Ar{\i}kan transform $(U_i, U_{i+1}) \mapsto (U_i+U_{i+1}, U_{i+1})$\footnote{In this paper, all additions between binary random variables are over the binary field unless otherwise specified.} to adjacent bits can also accelerate polarization.
In total, there are six invertible $2\times 2$ linear transforms over the binary field, including the swapping transform $(U_i, U_{i+1})\mapsto (U_{i+1}, U_i)$ and the Ar{\i}kan transform discussed above.
We show that these six invertible transforms are equivalent to (or the same as) the identity transform, the swapping transform, and the Ar{\i}kan transform for the purpose of accelerating polarization.
Therefore, for each pair of adjacent bits after each layer of polar transform, we only need to decide whether we apply the identity transform, the swapping transform, or the Ar{\i}kan transform.
In fact, the "+" sign in the name ABS+ polar codes comes from the addition in the Ar{\i}kan transform.
The "+" sign also has an additional meaning that ABS+ polar codes have smaller decoding error probabilities than ABS polar codes.

The encoding matrix $\mG_{n}^{\polar}$ in standard polar codes is obtained from the recursive relation $\mG_{n}^{\polar} = \mG_{n/2}^{\polar}\otimes \mG_2^{\polar}$,
where $\mG_2^{\polar} = \begin{bmatrix}
        1 & 0 \\
        1 & 1
    \end{bmatrix} $
and $\otimes$ is the Kronecker product.
Each Kronecker product is viewed as one layer of polar transform.
A standard polar code construction of code length $n = 2^m$ consists of $m$ layers of polar transforms.
In the ABS polar code construction, we add a permutation layer after each layer of polar transform, resulting in a different recursive relation
$\mG_n^{\ABS} = \mP_n^{\ABS}(\mG_{n/2}^{\ABS}\otimes\mG_2^{\polar})$.
The matrix $\mP_n^{\ABS}$ is an $n\times n$ permutation matrix which swaps certain pairs of  adjacent bits to accelerate polarization.
For ABS+ polar codes, we use a similar recursive relation
$\mG_n^{\ABSP} = \mQ_n^{\ABSP}(\mG_{n/2}^{\ABSP}\otimes \mG_2^{\polar})$
to construct the encoding matrix, where the $n\times n$ invertible matrix $\mQ_{n}^{\ABSP}$ performs the swapping transform or the Ar{\i}kan transform on certain pairs of adjacent bits.

In the ABS polar code construction, we require that the adjacent bits swapped by the permutation matrix $\mP_n^{\ABS}$ are fully separated.
This requirement plays a key role in the efficient decoding of ABS polar codes because it allows us to establish a recursive relation between the joint distribution of every pair of adjacent bits after each layer of polar transform.
Similarly, we require that the matrix $\mQ_n^{\ABSP}$ in the ABS+ polar code construction only performs the swapping transform and the Ar{\i}kan transform on adjacent bits that are fully separated.
As a consequence, a similar recursive relation between the joint distributions of adjacent bits can also be derived  for ABS+ polar codes, and the decoding of ABS+ polar codes has the same time complexity as the decoding of ABS polar codes.

We conduct extensive simulations over binary-input AWGN channels to compare the performance of the CRC-aided SCL decoder for ABS+ polar codes, ABS polar codes, and standard polar codes.
We run simulations for code length $256,512,1024,2048$ and code rates $0.3, 0.5$, and $0.7$.
The decoding time of the following three decoders is more or less the same: (1) standard polar codes with list size $32$, (2) ABS polar codes with list size $20$, (3) ABS+ polar codes with list size $20$. ABS+ polar codes with list size $20$ demonstrate $0.1\dB$--$0.25\dB$ (respectively, $0.15\dB$--$0.35\dB$) improvement over ABS polar codes with list size $20$ (respectively, standard polar codes with list size $32$). If we set the list size to be $32$ for both ABS+ and standard polar codes, then the decoding time of ABS+ polar codes is $60\%$ longer than that of standard polar codes, but ABS+ polar codes demonstrate $0.2\dB$--$0.45\dB$ improvement over standard polar codes.

The rest of this paper is organized as follows:
In Section \ref{sect:main_idea}, we analyze all the possible invertible transform to explain why we only apply the swapping transform or Ar{\i}kan transform on a pair of adjacent bits and describe our main idea of the ABS+ polar code construction.
In Section \ref{sect:construction}, we give the algorithm to construct ABS+ polar codes.
In Section \ref{sect:encoding}, we present the encoding algorithm for ABS+ polar codes and compare ABS+ polar codes with convolutional polar codes.
In Section \ref{sect:decoding}, we derive the new SC decoder for ABS+ polar codes, and this algorithm can be easily extended to obtain the SCL decoder.
Finally, in Section \ref{sect:simu}, we show our simulation results.

\section{Main idea of the ABS+ polar code construction} \label{sect:main_idea}
\subsection{The polarization framework}\label{sect:polar_framework}
Let $U_1, U_2, \dots, U_n$ be $n$ i.i.d. Bernoulli-$1/2$ random variables.
We view $(U_1, \dots, U_n)$ as the message vector, and we use an $n\times n$ invertible matrix $\mG_n$ to encode it into the codeword vector
$(X_1, \dots, X_n) = (U_1, \dots, U_n)\mG_n$.
Each $X_i$ is transmitted through a BMS channel $W$, and the channel output vector is denoted as $(Y_1, \dots, Y_n)$.
The SC decoder decodes $U_i$ from all the previous message bits $(U_1, U_2, \dots, U_{i-1})$ and all the channel output $(Y_1, Y_2, \dots, Y_n)$.
For $1\le i \le n$, the conditional entropy
\begin{equation}\label{eq:H_i}
    H_i(\mG_n, W) := H(U_i|U_1,\dots, U_{i-1}, Y_1,\dots,  Y_n)
\end{equation}
measures the reliability of $U_i$ under the SC decoder when we use $\mG_n$ to encode the message vector and transmit the corresponding codeword through the BMS channel $W$.
Since $\mG_n$ is invertible, the chain rule of conditional entropy implies that
\begin{equation}
    H_1(\mG_n, W) + 
    \cdots + H_n(\mG_n, W) = n (1-I(W)),
\end{equation}
where $I(W)$ is the channel capacity of $W$.
A family of matrices $\{\mG_n\}$ is said to be polarizing over a BMS channel $W$ if $H_i(\mG_n, W)$ is close to either 0 or 1 for almost all $1\le i\le n$ as $n\to \infty$.
It is well-known that if $\{\mG_n\}$ is polarizing, then we can construct a family of capacity-achieving codes from $\{\mG_n\}$.
In \cite{Li2022TIT}, we quantify the polarization level of an invertible encoding matrix $\mG_n$ over a BMS channel $W$ using the following function
\begin{equation}\label{eq:Gamma}
    \Gamma(\mG_n, W) := \frac1n \sum_{i = 1}^{n}H_i(\mG_n, W)(1-H_i(\mG_n, W)).
\end{equation}
A family of matrices $\{\mG_n\}$ is polarizing over $W$ if $\Gamma(\mG_n, W)$ approaches $0$ sufficiently fast as $n$ tends to infinity.
Moreover, simulation results indicate that faster convergence rate of $\Gamma(\mG_n, W)$ implies that the corresponding capacity-achieving codes have smaller gap to capacity and better finite-length performance.

The most prominent example of polarizing matrices is $\{\mG_n^{\polar}\}$ in the standard polar code construction, where the polarization level increases after each layer of polar transform, i.e., $\Gamma(\mG_{n/2}^{\polar}\otimes \mG_2^{\polar}, W)\le \Gamma(\mG_{n/2}^{\polar}, W)$.
However, each layer of polar transform also increases the code length by a factor of 2.
In the ABS polar code construction \cite{Li2022TIT}, we find that swapping certain pairs of adjacent bits after each layer of polar transform deepens polarization without increasing the code length.
More precisely, we have
$\Gamma(\mP_n^{\ABS}(\mG_{n/2}^{\ABS}\otimes \mG_2^{\polar}), W)\le \Gamma(\mG_{n/2}^{\ABS}\otimes \mG_2^{\polar}, W)$, where $\mP_n^{\ABS}$ is a carefully constructed permutation matrix consisting of multiple swaps of disjoint pairs of adjacent bits.
In this paper, we observe that other linear transforms\footnote{
    The ABS polar code construction was inspired by Reed-Muller codes.
    Since the encoding matrix of length-$n$ Reed-Muller codes is a row permutation of $\mG_{n}^{\polar}$,
    we also constructed $\mG_n^{\ABS}$ as a row permutation of $\mG_{n}^{\polar}$.
    This is why we restricted ourselves to the swapping transform in the ABS polar code construction.
}
on adjacent bits can also accelerate polarization without increasing code length.
There are in total six invertible $2\times 2$ linear transforms.
In the next subsection, we show that we only need to consider three out of these six transforms, that is, the identity transform, the swapping transform and the Ar{\i}kan transform.
Choosing one of these three linear transforms for each pair of adjacent bits after each layer of polar transform gives us ABS+ polar codes, which polarize even faster than ABS polar codes.

\subsection{Classification of invertible linear transforms on adjacent bits}\label{sect:classify}
As discussed in Section \ref{sect:introduction}, we need to construct an invertible matrix $\mQ_{n}^{\ABSP}$ when we build $\mG_n^{\ABSP}$ from $\mG_{n/2}^{\ABSP}$ using the recursive relation $\mG_{n}^{\ABSP} = \mQ_{n}^{\ABSP}(\mG_{n/2}^{\ABSP}\otimes\mG_2^{\polar})$.
The matrix $\mQ_n^{\ABSP}$ performs certain linear transforms on certain pairs of adjacent bits.
The first step of constructing $\mQ_n^{\ABSP}$ is to find the best linear transform for each fixed pair of adjacent bits.
For each $1\le i \le n-1$, we define $\cM_n^{(i)}$ as the set of $n\times n$ invertible matrices whose corresponding linear transforms map all coordinates except for $U_i$ and $U_{i+1}$ in a binary vector $(U_1, U_2,\dots, U_n)$ to themselves.
More precisely, let us write $(U'_1, U'_2, \dots, U'_n) = (U_1, U_2,\dots, U_n)\mM_n$ for an $n\times n$ invertible matrix $\mM_n$ and a binary vector $(U_1, U_2, \dots, U_n)$.
By definition, $\mM_n \in \cM_n^{(i)}$ if and only if $U_j' = U_j$ for all $j\in \{1, 2, \dots, n\}\setminus\{i, i+1\}$ and all $(U_1, U_2, \dots, U_n)\in \{0,1\}^n$.
For each $1\le i\le n-1$, the set $\cM_n^{(i)}$ consists of six matrices, denoted as $\mI_n, \mS_n^{(i)}, \mA_n^{(i)}, \mD_n^{(i)}, \mE_n^{(i)}, \mK_n^{(i)}$.
Their corresponding linear transforms are listed below:
\begin{enumerate}[(1)]
    \item  $\mI_n$ corresponds to the identity transform.

    \item  $\mS_n^{(i)}$ maps  $(U_i, U_{i+1})$ to $(U_{i+1}, U_{i})$ while fixing all the other $U_j'$s unchanged;

    \item  $\mA_n^{(i)}$ maps $(U_i, U_{i+1})$ to $(U_{i}+U_{i+1}, U_{i+1})$ while fixing all the other $U_j'$s unchanged;

    \item  $\mD_n^{(i)}$ maps $(U_i, U_{i+1})$ to $(U_{i}, U_{i}+U_{i+1})$ while fixing all the other $U_j'$s unchanged;

    \item  $\mE_n^{(i)}$ maps $(U_i, U_{i+1})$ to $(U_{i+1}, U_{i}+U_{i+1})$ while fixing all the other $U_j'$s unchanged;

    \item  $\mK_n^{(i)}$ maps $(U_i, U_{i+1})$ to $(U_{i}+U_{i+1}, U_{i})$ while fixing all the other $U_j'$s unchanged.
\end{enumerate}
Next we show that these six matrices in $\cM_n^{(i)}$ can be partitioned into 3 groups of size 2.
Two matrices in the same group are equivalent for the purpose of accelerating polarization.

\begin{lemma}\label{lemma:classify}
    Let $\mG_n$ be an $n\times n$ invertible matrix and let $W$ be a BMS channel.
    We have
    $\Gamma(\mD_n^{(i)}\mG_n, W)= \Gamma(\mG_n, W)$, $\Gamma(\mE_n^{(i)}\mG_n, W) = \Gamma(\mS_n^{(i)}\mG_n, W)$, and $\Gamma(\mK_n^{(i)}\mG_n, W) = \Gamma(\mA_n^{(i)}\mG_n, W)$.
\end{lemma}
\begin{proof}
    Recall the definition of $H_i(\mG_n, W)$ in \eqref{eq:H_i}.
    Since the linear transforms corresponding to all six matrices in $\cM_n^{(i)}$ map $U_j$ to itself for every $j\not= i, i+1$, we have
    \begin{equation*}
        \begin{aligned}
              & H_j(\mG_n, W)            = H_j(\mS_n^{(i)}\mG_n, W) = H_j(\mA_n^{(i)}\mG_n, W)  \\
            = & H_j(\mD_n^{(i)}\mG_n, W) = H_j(\mE_n^{(i)}\mG_n, W) = H_j(\mK_n^{(i)}\mG_n, W).
        \end{aligned}
    \end{equation*}
    Next, we will analyze $H_i$ and $H_{i+1}$ in three cases.
    \begin{enumerate}[(1)]
        \item $\mD_n^{(i)}$ maps  $(U_i, U_{i+1})$ to $(U_{i}, U_{i}+U_{i+1})$.
              Therefore,
              \begin{equation*}
                  \begin{aligned}
                        & H_i(\mD_n^{(i)}\mG_n, W)                               \\    = &  H(U_i| U_1, \dots, U_{i-1}, Y_1, \dots, Y_n) \\= &H_i(\mG^n, W), \\
                        & H_{i+1}(\mD_n^{(i)}\mG_n, W)                           \\  =  &H(U_i+U_{i+1}|U_1, \dots, U_{i-1}, U_i, Y_1, \dots, Y_n)      \\
                      = & H(U_{i+1} | U_1, \dots, U_{i-1}, U_i, Y_1, \dots, Y_n) \\
                      = & H_{i+1}(\mG^n, W).
                  \end{aligned}
              \end{equation*}

        \item $\mE_n^{(i)}$ maps $(U_i, U_{i+1})$ to $(U_{i+1}, U_{i}+U_{i+1})$, and $\mS_n^{(i)}$ maps  $(U_i, U_{i+1})$ to $(U_{i+1}, U_{i})$.
              Therefore,
              \begin{equation*}
                  \begin{aligned}
                        & H_i( \mE_n^{(i)}\mG_n, W)                                \\     =& H(U_{i+1} | U_1, \dots, U_{i-1}, Y_1, \dots, Y_n)\\=&H_i(\mS_n^{(i)}\mG_n, W),              \\
                        & H_{i+1}(\mE_n^{(i)}\mG_n, W)                             \\  = &H(U_i+U_{i+1}| U_1, \dots, U_{i-1}, U_{i+1}, Y_1, \dots, Y_n)                            \\
                      = & H(U_{i} | U_1, \dots, U_{i-1}, U_{i+1}, Y_1, \dots, Y_n) \\= &H_{i+1}(\mS_n^{(i)}\mG_n, W).
                  \end{aligned}
              \end{equation*}

        \item $\mK_n^{(i)}$ maps $(U_i, U_{i+1})$ to $(U_{i}+U_{i+1}, U_{i})$, and $\mK_n^{(i)}$ maps $(U_i, U_{i+1})$ to $(U_{i}+U_{i+1}, U_{i})$.
              Therefore,
              \begin{equation*}
                  \begin{aligned}
                        & H_i(\mK_n^{(i)}\mG_n, W)                                       \\   = &H(U_i+U_{i+1}| U_1, \dots, U_{i-1}, Y_1, \dots, Y_n)\\=&H_i(\mA_n^{(i)}\mG_n, W),                 \\
                        & H_{i+1}(\mK_n^{(i)}\mG_n, W)                                   \\ =  &H(U_i|U_1, \dots, U_{i-1},U_i+ U_{i+1}, Y_1, \dots, Y_n)                                       \\
                      = & H(U_{i+1} | U_1, \dots, U_{i-1},U_i+ U_{i+1}, Y_1, \dots, Y_n) \\ =& H_{i+1}(\mA_n^{(i)}\mG_n, W).
                  \end{aligned}
              \end{equation*}

    \end{enumerate}
    The lemma then follows immediately from definition \eqref{eq:Gamma}.
\end{proof}

Note that $\mS_n^{(i)}$ performs the swapping transform on $(U_i,\linebreak U_{i+1})$, and $\mA_n^{(i)}$ performs the Ar{\i}kan transform on $(U_i, U_{i+1})$.
Lemma \ref{lemma:classify} tells us that for the purpose of accelerating polarization, we only need to consider the identity transform, the swapping transform, and the Ar{\i}kan transform on $(U_i, U_{i+1})$ for each fixed value of $i$.
In fact, there will not be any extra benefit even if we take nonlinear transforms into account because the nonlinear transforms are equivalent to the linear ones in terms of the polarization speed. Indeed, every invertible transform (no matter linear or nonlinear) on two adjacent bits is a permutation on the set $\{00,01,10,11\}$. Therefore, the total number of invertible transforms on a pair of adjacent bits is $4!=24$. For every $a,b\in\{0,1\}$, we define an invertible transform $\sigma_i(a,b): (U_i,U_{i+1})\to (U_i+a, U_{i+1}+b)$. Then the $24$ invertible transforms on $(U_i,U_{i+1})$ are
\begin{equation} \label{eq:24trans}
    \begin{aligned}
        \bigcup_{a,b\in\{0,1\}} \Big\{ \mI_n\circ\sigma_i(a,b), \mS_n^{(i)}\circ\sigma_i(a,b), \mA_n^{(i)}\circ\sigma_i(a,b), \\
        \mD_n^{(i)}\circ\sigma_i(a,b), \mE_n^{(i)}\circ\sigma_i(a,b), \mK_n^{(i)}\circ\sigma_i(a,b) \Big\} .
    \end{aligned}
\end{equation}
It is clear that the transform $\sigma_i(a,b)$ does not change $H_i$ or $H_{i+1}$ for all $a,b\in\{0,1\}$. Therefore, every invertible transform in \eqref{eq:24trans} is equivalent to one of the $3$ transforms $\{\mI_n, \mS_n^{(i)}, \mA_n^{(i)}\}$ for the purpose of accelerating polarization. Thus we conclude that using nonlinear transforms on adjacent bits can not further improve the speed of polarization.

Similarly to the ABS polar code construction, we require that $\mQ_n^{\ABSP}$ only performs the swapping transform, and the Ar{\i}kan transform on fully separated pairs of adjacent bits.
This fully separated requirement is explained in the next subsection.

\subsection{The fully separated requirement on $\mQ_n^{\ABSP}$}\label{sect:fully_separated}
We first recall the fully separated requirement on the permutation matrix $\mP_n^{\ABS}$ in the ABS polar code construction.
Let $\cI^{(n), \ABS}\subseteq \{1,2,\dots, n-1\}$ be the set containing the indices of the first bit in each pair of adjacent bits that are swapped by $\mP_n^{\ABS}$.
Then $\mP_n^{\ABS}$ can be written as
$$
    \mP_n^{\ABS} = \prod_{i\in \cI^{(n), \ABS}} \mS_n^{(i)}.
$$
Suppose that $\cI^{(n), \ABS} = \{i_1, i_2, \dots, i_t\}$, where $t$ is the size of the set.
To enhance the polarization level, we will prove later in Lemma~\ref{lemma:even_reason} that the elements in $\cI^{(n), \ABS}$ must be even.
Furthermore, in order to track the evolution of adjacent-bits-channels, we can not apply transforms on two successive pairs of adjacent bits $(U_{2i}, U_{2i+1})$ and $(U_{2i+2}, U_{2i+3})$
as stated in \cite[Section~III-G]{Li2022TIT}.
The fully separated requirement on $\mP_n^{\ABS}$ stipulates that
\begin{equation}\label{eq:fully_separeted}
    i_2 \ge i_1 + 4, ~~ i_3 \ge i_2 + 4, ~~  i_4 \ge i_3 + 4, ~~  \dots, ~~  i_t \ge i_{t-1} + 4.
\end{equation}
This requirement allows us to track the joint distribution of every pair of adjacent bits through different layers of polar transforms.

A similar requirement is also imposed on $\mQ_{n}^{\ABSP}$ in the ABS+ polar code construction.
Let $\cI_S^{(n)}\subseteq \{1,2,\dots, n-1\}$ be the set containing the indices of the first bit in each pair of adjacent bits upon which $\mQ_n^{\ABSP}$ performs the swapping transform.
Let $\cI_A^{(n)}\subseteq \{1,2,\dots, n-1\}$ be the set containing the indices of the first bit in each pair of adjacent bits upon which $\mQ_n^{\ABSP}$ performs the Ar{\i}kan transform.
Then $\mQ_n^{\ABSP}$ can be written as
\begin{equation}\label{eq:Q_n}
    \begin{aligned}
        \mQ_n^{\ABSP} = \left(\prod_{i\in \cI^{(n)}_S}\mS_{n}^{(i)} \right)\cdot \left(\prod_{i\in \cI^{(n)}_A}\mA_{n}^{(i)}\right).
    \end{aligned}
\end{equation}
Define the set $\cI^{(n)} = \cI_S^{(n)} \cup \cI_A^{(n)}$, and we write the elements of $\cI^{(n)}$ as $\cI^{(n)} = \{i_1, i_2, \dots, i_t\}$, where $t = |\cI^{(n)}|$.
The fully separated requirement on $\mQ_n^{\ABSP}$ stipulates that (i) $\cI_S^{(n)} \cap \cI_A^{(n)} = \emptyset$;
(ii) the elements in the set $\cI^{(n)}$ satisfy \eqref{eq:fully_separeted}.
This requirement guarantees that ABS+ polar codes have the same decoding time complexity as ABS polar codes.
In Section \ref{sect:track_ABS+}, we will further prove that all the elements in the set $\cI^{(n)}$ are even numbers.

As a final remark, we need to choose $m$ matrices $\mQ_2^{\ABSP},\linebreak \mQ_4^{\ABSP}, \dots, \mQ_n^{\ABSP}$ one by one in the construction of ABS+ polar codes with code length $n = 2^m$.

\section{Code construction of ABS+ polar codes} \label{sect:construction}
The first step to construct ABS+ polar codes with code length $n=2^m$ is to choose $m$ matrices $\mQ_2^{\ABSP}, \mQ_4^{\ABSP},\linebreak \dots, \mQ_n^{\ABSP}$.
After that, we calculate the conditional entropies $\{H_i(\mG_n^{\ABSP}, W)\}_{i = 1}^{n}$ and use them to determine which bits are information bits.
Since the swapping transform and the Ar{\i}kan transform are applied to certain pairs of adjacent bits after each layer of polar transform, there is no recursive relation between bit-channels like the one in standard polar codes.
In ABS polar codes \cite{Li2022TIT}, we introduced the concept of adjacent-bits-channels to overcome this issue.
More precisely, the fully separated requirement \eqref{eq:fully_separeted} on $\mP_n^{\ABS}$ allows us to establish a  recursive relation between adjacent-bits-channels.
Since $\mQ_n^{\ABSP}$ also satisfies the requirement \eqref{eq:fully_separeted}, we are able to obtain a similar recursive relation between adjacent-bis-channels  for ABS+ polar codes.

The recursive relation for ABS polar codes is characterized by the Double-Bits (DB) polar transform and the Swapped-Double-Bits (SDB) polar transform: see Section~\ref{sect:track_ABS+}.
For ABS+ polar codes, we need one more transform, called the Added-Double-Bits (ADB) polar transform, to characterize the recursive relation because we have one more choice of linear transform for each pair of adjacent bits.

We organize this section as follows:
In Section \ref{sect:track_ABS+}, we introduce the ADB polar transform and establish the recursive relation for ABS+ polar codes.
In Section \ref{sect:construct_Q}, we describe how to choose the matrices $\mQ_2^{\ABSP}, \mQ_4^{\ABSP}, \dots, \mQ_n^{\ABSP}$.
Finally, in Section \ref{sect:cons_algo}, we summarize the algorithm of the ABS+ polar code construction.

\begin{figure*}
    \centering
    \begin{tikzpicture}
        \node [block, align=center] at (-2.2,1.6)  (y-1) { $U_1$ \\[0.5em]  $U_2$  \\[0.5em]  \vdots \\[0.5em]  $U_n$ };
        \node [sblock, align=center] at (0.1,1.6)  (y0) {$\mQ_n^{\ABSP}$};
        \node [block, align=center] at (2.3	,1.6)  (y1) { $\widehat{U}_1$ \\[0.5em]  $\widehat{U}_2$  \\[0.5em]  \vdots \\[0.5em]  $\widehat{U}_n$ };
        \node [sblock, align=center] at (4.6,1.6)  (y2) {$\mathbf{G}_{n/2}^{\ABSP} \otimes \mathbf{G}_2^{\polar}$};
        \node [block, align=center] at (7,1.6)  (y3) { $X_1$ \\[0.5em] $X_2$  \\[0.5em] \vdots  \\[0.5em]  $X_n$ };
        \node [block] at (8.5, 2.7) (w1) {$W$};
        \node [block] at (8.5, 2) (w2) {$W$};
        \node  at (8.5, 1.3)  {\vdots};
        \node [block] at (8.5, 0.6) (w3) {$W$};

        \node at (10, 2.7) (z1) {$Y_1$};
        \node at (10, 2) (z2) {$Y_2$};
        \node  at (10, 1.3) {\vdots};
        \node at (10, 0.6) (z3) {$Y_n$};

        \draw[->,thick] (y-1)--(y0);
        \draw[->,thick] (y0)--(y1);
        \draw[->,thick] (y1)--(y2);
        \draw[->,thick] (y2)--(y3);

        \draw[->,thick] (w1)--(z1);
        \draw[->,thick] (w2)--(z2);
        \draw[->,thick] (w3)--(z3);

        \draw[->,thick] (7.4, 2.7)--(w1);
        \draw[->,thick] (7.4, 2)--(w2);
        \draw[->,thick] (7.4, 0.6)--(w3);

        \node at (4.5, -0.7) [align=center] (p1) {$(\widehat{U}_1,\dots,\widehat{U}_n)=(U_1,\dots,U_n)\mQ_n^{\ABSP}, \quad
            (X_1,\dots,X_n)=(\widehat{U}_1,\dots,\widehat{U}_n) (\mathbf{G}_{n/2}^{\ABSP} \otimes \mathbf{G}_2^{\polar})$ \\[0.4em]
        $(X_1,\dots,X_n)=(U_1,\dots,U_n)\mathbf{Q}_n^{\ABSP}(\mathbf{G}_{n/2}^{\ABSP} \otimes \mathbf{G}_2^{\polar})=(U_1,\dots,U_n)\mathbf{G}_n^{\ABSP}$};

        \node at (0, -2.2) [align=center] (nbc1)
        {Two sets of bit-channels\\
        $\{W_i^{(n),\ABSP}:U_i\to U_1,\dots,U_{i-1},Y_1,\dots,Y_n\}_{i=1}^n$\\[0.1em]
        $\{\widehat{W}_i^{(n),\ABSP}:\widehat{U}_i\to \widehat{U}_1,\dots,\widehat{U}_{i-1},Y_1,\dots,Y_n\}_{i=1}^n$
        };
        \node at (8, -2.2) [align=center] (nbc2){
        Two sets of adjacent-bits-channels\\
        $\{V_i^{(n),\ABSP}:U_i,U_{i+1}\to U_1,\dots,U_{i-1},Y_1,\dots,Y_n\}_{i=1}^{n-1}$\\[0.1em]
        $\{\widehat{V}_i^{(n),\ABSP}:\widehat{U}_i,\widehat{U}_{i+1}\to \widehat{U}_1,\dots,\widehat{U}_{i-1},Y_1,\dots,Y_n\}_{i=1}^{n-1}$
        };
    \end{tikzpicture}
    \caption{
    Definitions of bit-channels $\{W_i^{(n),\ABSP}\}_{i=1}^n$, $\{\widehat W_i^{(n),\ABSP}\}_{i=1}^n$ and adjacent-bits-channels $\{V_i^{(n),\ABSP}\}_{i=1}^{n-1}$,
    $\{\widehat V_i^{(n),\ABSP}\}_{i=1}^{n-1}$, where $U_1, U_2, \dots,U_n$ are i.i.d. Bernoulli-$1/2$ random variables.
    }
    \label{fig:structure_of_ABS+}
\end{figure*}

\subsection{Recursive relation between adjacent-bits-channels} \label{sect:track_ABS+}
In this subsection, we describe how to calculate the  conditional entropies $\{H_i(\mG_n^{\ABSP}, W)\}_{i=1}^n$ when $\mQ_2^{\ABSP}, \mQ_4^{\ABSP}, \linebreak\dots, \mQ_n^{\ABSP}$ are known.
We first define the bit-channels $\{W_i^{(n),\ABSP}\}_{i=1}^n$ for ABS+ polar codes in Fig. \ref{fig:structure_of_ABS+}.
Since $H_i(\mG_n^{\ABSP}, W) = 1 - I(W_i^{(n),\ABSP})$, we only need to calculate the transition probabilities of  $\{W_i^{(n),\ABSP}\}_{i=1}^n$.
In standard polar codes, the transition probabilities of bit-channels are calculated using a recursive relation which is not available for ABS or ABS+ polar codes.
Following the method in \cite{Li2022TIT}, we define the adjacent-bits-channels $\{V_i^{(n), \ABSP}\}^{n-1}_{i = 1}$ for ABS+ polar codes in Fig. \ref{fig:structure_of_ABS+}, and we will derive a recursive relation between $\{V_i^{(n), \ABSP}\}_{i = 1}^{n-1}$ and $\{V_i^{(n/2),\ABSP}\}_{i = 1}^{n/2-1}$.
Once we obtain the transition probabilities of $\{V_i^{(n), \ABSP}\}^{n-1}_{i = 1}$ from this recursive relation, the transition probabilities of $\{W_i^{(n), \ABSP}\}_{i = 1}^{n}$ can be calculated as follows:
\begin{equation}\label{eq:v_to_w}
    \begin{aligned}
         & W_{i}^{(n)  , \ABSP}(y_1,\dots,y_n, u_1,\dots, u_{i-1} | u_i    ) \\ = &\cfrac12 \sum_{u_{i+1}\in \{0,1\}} V_{i}^{(n), \ABSP}(y_1,\dots, y_n, u_1,\dots, u_{i-1} | u_i, u_{i+1}),       \\
         & W_{i+1}^{(n), \ABSP}(y_1,\dots,y_n, u_1,\dots, u_{i}   | u_{i+1}) \\&=\cfrac12 V_{i}^{(n), \ABSP}(y_1,\dots, y_n, u_1,\dots, u_{i-1} | u_i, u_{i+1})
    \end{aligned}
\end{equation}
for $1\le i\le n-1$.
In Fig.~\ref{fig:structure_of_ABS+}, we also introduce bit-channels $\{\widehat W^{(n), \ABSP}\}_{i=1}^n$ and adjacent-bits-channels $\{\widehat V_i^{(n), \ABSP}\}_{i=1}^n$, which represent the channels prior to applying a layer of invertible transforms $\mQ_{n}^{\ABSP}$ to specific pairs of adjacent bits.

As mentioned at the end  of Section \ref{sect:fully_separated}, the set $\cI^{(n), \ABS}$ in ABS polar codes and  $\cI^{(n)}$ in ABS+ polar codes both satisfy the requirement \eqref{eq:fully_separeted}, which guarantees the existence of a recursive relation between $\{V_i^{(n), \ABSP}\}_{i = 1}^{n-1}$ and $\{V_i^{(n/2), \ABSP}\}_{i = 1}^{n/2-1}$.
The necessity of the requirement \eqref{eq:fully_separeted} was explained in \cite[Section~III-G]{Li2022TIT}.
Here we point out  another similarity between the two sets  $\cI^{(n), \ABS}$ and $\cI^{(n)}$.
In \cite[Section III-C]{Li2022TIT}, we showed that all the elements in  $\cI^{(n), \ABS}$ are even numbers.
Next we prove that all the elements in $\cI^{(n)}$ are also even numbers.
Following the notation in Fig. \ref{fig:structure_of_ABS+}, we only need to show that applying the swapping transform or the Ar{\i}kan transform to $(\widehat U_{2j-1}, \widehat U_{2j})$ does not increase the polarization level for any $1\le j\le n/2$.
This is proved in Lemma \ref{lemma:even_reason} below.

Recall the definitions of $\mS_{n}^{(i)}$ and $\mA_n^{(i)}$ in Section \ref{sect:classify}.

\begin{figure}
    \centering
    \begin{tikzpicture}[scale=0.90]
        \draw
        node at (-0.5,10.5) [] (u1)  {$\widehat U_{2j-1}$}
        node at (-0.5,9) [] (u2)  {$\widehat U_{2j}$}
        node at (1.0,10.5) [XOR,scale=1.2] (x1) {}
        node at (2.5,10.5) [] (xx1)  {$\tilde X_{2j-1}$}
        node at (2.5,9) [] (xx2)  {$\tilde X_{2j}$}
        node at (4.8,10.5) [block] (v1)  {$W_{j}^{(n/2), \ABSP}$}
        node at (4.8,9) [block] (v2)  {$W_{j}^{(n/2),\ABSP}$}
        node at (7.5,10.5) [] (y1)  {$\tilde Y_{2j-1}$}
        node at (7.5,9) [] (y2)  {$\tilde Y_{2j}$};
        \draw[fill] (1.0, 9) circle (.6ex);
        \draw[very thick,->](u1) -- node {}(x1);
        \draw[very thick,->](u2) -| node {}(x1);
        \draw[very thick,->](x1) -- (xx1);
        \draw[very thick,->](u2) -- (xx2);
        \draw[very thick,->](xx1) -- (v1);
        \draw[very thick,->](xx2) -- (v2);
        \draw[very thick,->](v1) -- node {}(y1);
        \draw[very thick,->](v2) -- node {}(y2);
    \end{tikzpicture}
    \caption{The channel mapping from $\tilde X_{2j-1}$ to $\tilde Y_{2j-1}$ and the channel mapping from  $\tilde X_{2j}$ to $\tilde Y_{2j}$ are both  $W_{j}^{(n/2), \ABSP}$.
        Moreover, $(\tilde X_{2j-1},\tilde Y_{2j-1})$ and $(\tilde X_{2j},\tilde Y_{2j})$ are independent.}
    \label{fig:even_reason}
\end{figure}
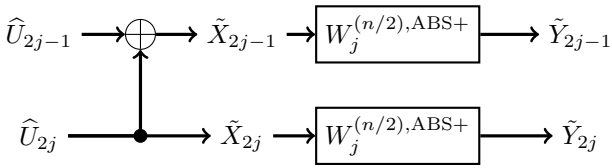

\begin{lemma}\label{lemma:even_reason}
    Let $W$ be a BMS channel.
    We use the shorthand notation $\widehat \mG_n^{\ABSP} = \mG_{n/2}^{\ABSP}\otimes \mG_2^{\polar}$.
    For $1\le j\le n/2$, we have
    \begin{equation}\label{eq:even_reason_1}
        \begin{aligned}
            \Gamma(\mS^{(2j-1)}_{n} \widehat \mG_n^{\ABSP}, W) & = \Gamma(\mA^{(2j-1)}_{n}\widehat \mG_n^{\ABSP}, W) \\&\ge \Gamma(\widehat \mG_n^{\ABSP}, W).
        \end{aligned}
    \end{equation}
\end{lemma}

\begin{proof}
    Recall the definition of $H_i(\mG_n, W)$ in \eqref{eq:H_i}.
    Clearly, for every $i\not= 2j-1, 2j$, we have
    \begin{equation}\label{eq:even_reason_2}
        \begin{aligned}
            H_i(\mS^{(2j-1)}_{n}\widehat \mG_n^{\ABSP}, W) & = H_i(\mA^{(2j-1)}_{n}\widehat \mG_n^{\ABSP}, W) \\&= H_i(\widehat \mG_n^{\ABSP}, W).
        \end{aligned}
    \end{equation}
    Next we analyze $H_{2j-1}$ and $H_{2j}$.
    Following the notation in Fig. \ref{fig:structure_of_ABS+},
    \begin{equation}\label{eq:even_reason_3}
        \begin{aligned}
             & H_{2j-1}(\widehat \mG_n^{\ABSP}, W) \\= &H(\widehat U_{2j-1} | \widehat U_1, \widehat U_2, \dots, \widehat U_{2j-2}, Y_1, Y_2, \dots, Y_n), \\
             & H_{2j  }(\widehat \mG_n^{\ABSP}, W) \\= &H(\widehat U_{2j  } | \widehat U_1, \widehat U_2, \dots, \widehat U_{2j-1}, Y_1, Y_2, \dots, Y_n).
        \end{aligned}
    \end{equation}
    For $1\le i\le n/2$, define $\tilde X_{2i-1} = \widehat U_{2i-1} + \widehat U_{2i}$ and $\tilde X_{2i} = \widehat U_{2i}$.
    We further define
    \begin{equation*}
        \begin{aligned}
             & \tilde Y_{2j-1} = (\tilde X_1, \tilde{X}_3, \dots, \tilde{X}_{2j-3}, Y_1, Y_3, \dots, Y_{n-1}), \\
             & \tilde Y_{2j} = (\tilde X_2, \tilde{X}_4, \dots, \tilde{X}_{2j-2}, Y_2, Y_4, \dots, Y_{n}).
        \end{aligned}
    \end{equation*}
    Then \eqref{eq:even_reason_3} can be written as
    \begin{equation*}
        \begin{aligned}
             & H_{2j-1}(\widehat \mG_n^{\ABSP}, W) = H(\widehat U_{2j-1} | \tilde Y_{2j-1}, \tilde Y_{2j}),                    \\
             & H_{2j  }(\widehat \mG_n^{\ABSP}, W) = H(\widehat U_{2j  } |\widehat U_{2j-1},  \tilde Y_{2j-1}, \tilde Y_{2j}).
        \end{aligned}
    \end{equation*}
    It is easy to see that the channel mapping from $\tilde X_{2j-1}$ to $\tilde Y_{2j-1}$ and the channel mapping from  $\tilde X_{2j}$ to $\tilde Y_{2j}$ are both  $W_{j}^{(n/2), \ABSP}$.
    Moreover, the two pairs of random variables $(\tilde X_{2j-1},\tilde Y_{2j-1})$ and $(\tilde X_{2j},\tilde Y_{2j})$ are independent.
    This is illustrated in Fig. \ref{fig:even_reason}.
    Therefore,
    \begin{equation}\label{eq:even_reason_4}
        \begin{aligned}
                & H_{2j-1}(\widehat \mG_n^{\ABSP}, W) + H_{2j}(\widehat \mG_n^{\ABSP}, W)                                                                   \\ = &H(\widehat U_{2j-1} | \tilde Y_{2j-1}, \tilde Y_{2j}) + H(\widehat U_{2j  } |\widehat U_{2j-1},  \tilde Y_{2j-1}, \tilde Y_{2j}) \\
            =   & H(\widehat U_{2j-1}, \widehat U_{2j} | \tilde Y_{2j-1}, \tilde Y_{2j})=H(\tilde X_{2j-1}, \tilde X_{2j} | \tilde Y_{2j-1}, \tilde Y_{2j}) \\=&2 (1 - I(W_{j}^{(n/2), \ABSP})),                                  \\
                & H_{2j-1}(\widehat \mG_n^{\ABSP}, W)                                                                                                       \\=&1-I((W^{(n/2),\ABSP}_j)^-)\ge 1 - I(W_{j}^{(n/2), \ABSP})\\
            \ge & 1-I((W^{(n/2),\ABSP}_j)^+)= H_{2j}(\widehat \mG_n^{\ABSP}, W).
        \end{aligned}
    \end{equation}
    By definition,
    \begin{equation*}
        \begin{aligned}
              & H_{2j-1}(\mS_n^{(2j-1)}\widehat \mG_n^{\ABSP}, W) = H(\widehat U_{2j  } | \tilde Y_{2j-1}, \tilde Y_{2j})                  \\
            = & H(\tilde X_{2j  } | \tilde Y_{2j-1}, \tilde Y_{2j}) = H(\tilde X_{2j  } | \tilde Y_{2j}) = 1 - I(W_{j}^{(n/2), \ABSP}),    \\
              & H_{2j  }(\mS_n^{(2j-1)}\widehat \mG_n^{\ABSP}, W) = H(\widehat U_{2j-1} |\widehat U_{2j},  \tilde Y_{2j-1}, \tilde Y_{2j}) \\=&H(\tilde X_{2j-1} + \tilde X_{2j} |\tilde X_{2j},  \tilde Y_{2j-1}, \tilde Y_{2j})=H(\tilde X_{2j-1}|\tilde X_{2j},  \tilde Y_{2j-1}, \tilde Y_{2j}) \\
            = & H(\tilde X_{2j-1}| \tilde Y_{2j-1}) =  1- I(W_{j}^{(n/2), \ABSP}).
        \end{aligned}
    \end{equation*}
    Similarly,
    \begin{equation}\label{eq:even_reason_5}
        \begin{aligned}
              & H_{2j-1}(\mA_n^{(2j-1)}\widehat \mG_n^{\ABSP}, W)                   \\= &H(\widehat U_{2j-1} + \widehat{U}_{2j} | \tilde Y_{2j-1}, \tilde Y_{2j}) = H(\tilde X_{2j-1} | \tilde Y_{2j-1}, \tilde Y_{2j}) \\
            = & H(\tilde X_{2j-1} | \tilde Y_{2j-1}) = 1 - I(W_{j}^{(n/2), \ABSP}), \\
              & H_{2j  }(\mA_n^{(2j-1)}\widehat \mG_n^{\ABSP}, W)                   \\= &H(\widehat U_{2j  } |\widehat U_{2j-1} + \widehat U_{2j},  \tilde Y_{2j-1}, \tilde Y_{2j})\\ =&H(\tilde X_{2j} |\tilde X_{2j-1},  \tilde Y_{2j-1}, \tilde Y_{2j}) \\
            = & H(\tilde X_{2j}| \tilde Y_{2j}) =  1- I(W_{j}^{(n/2), \ABSP}).
        \end{aligned}
    \end{equation}
    Then we have
    \begin{equation*}
        \begin{aligned}
              & H_{2j-1}(\mS_n^{(2j-1)}\widehat \mG_n^{\ABSP}, W) = H_{2j}(\mS_n^{(2j-1)}\widehat \mG_n^{\ABSP}, W) \\
            = & H_{2j-1}(\mA_n^{(2j-1)}\widehat \mG_n^{\ABSP}, W) = H_{2j}(\mA_n^{(2j-1)}\widehat \mG_n^{\ABSP}, W) \\= &1 - I(W_{j}^{(n/2), \ABSP}).
        \end{aligned}
    \end{equation*}
    Combining this with \eqref{eq:even_reason_2} and the definition of $\Gamma(\mG_n, W)$ in \eqref{eq:Gamma}, we prove the first equality in \eqref{eq:even_reason_1}.

    According to \eqref{eq:even_reason_4}, we have
    \begin{equation}\label{eq:even_reason_6}
        \begin{aligned}
                & H_{2j-1}(\widehat \mG_n^{\ABSP}, W)(1-H_{2j-1}(\widehat \mG_n^{\ABSP}, W))                             \\ &\hspace*{.6in} + H_{2j}(\widehat \mG_n^{\ABSP}, W)(1-H_{2j}(\widehat \mG_n^{\ABSP}, W))     \\
            =   & H_{2j-1}(\widehat \mG_n^{\ABSP}, W) + H_{2j}(\widehat \mG_n^{\ABSP}, W)                                \\ &\hspace*{.4in} - (H_{2j-1}(\widehat \mG_n^{\ABSP}, W))^2-(H_{2j}(\widehat \mG_n^{\ABSP}, W))^2 \\
            =   & 2(1 - I(W_{j}^{(n/2), \ABSP}))- (H_{2j-1}(\widehat \mG_n^{\ABSP}, W))^2                                \\ &\hspace*{1.4in} -(H_{2j}(\widehat \mG_n^{\ABSP}, W))^2                                           \\
            \le & 2(1 - I(W_{j}^{(n/2), \ABSP})) - \frac12 (H_{2j-1}(\widehat \mG_n^{\ABSP}, W)                          \\ &\hspace*{1.4in} + H_{2j}(\widehat \mG_n^{\ABSP}, W))^2                                    \\
            =   & 2(1 - I(W_{j}^{(n/2), \ABSP})) - 2(1 - I(W_{j}^{(n/2), \ABSP}))^2                                      \\
            =   & 2(1 - I(W_{j}^{(n/2), \ABSP}))I(W_{j}^{(n/2), \ABSP})                                                  \\
            =   & H_{2j-1}(\mA_n^{(2j-1)}\widehat \mG_n^{\ABSP}, W)(1-H_{2j-1}(\mA_n^{(2j-1)}\widehat \mG_n^{\ABSP}, W)) \\
                & + H_{2j}(\mA_n^{(2j-1)}\widehat \mG_n^{\ABSP}, W)(1-H_{2j}(\mA_n^{(2j-1)}\widehat \mG_n^{\ABSP}, W)),
        \end{aligned}
    \end{equation}
    where the inequality follows from the Cauchy–Schwarz inequality, and the last equality follows from \eqref{eq:even_reason_5}.
    Combining \eqref{eq:even_reason_6} with \eqref{eq:even_reason_2}, we prove the inequality in \eqref{eq:even_reason_1}.
    This completes the proof of this lemma.
\end{proof}

Since the elements $i_1, i_2, \dots, i_t$ in $\cI^{(n)}$ are even numbers, we rewrite them as  $\cI^{(n)} = \{2j_1, 2j_2, \dots, 2j_t\}$, and the condition \eqref{eq:fully_separeted} becomes
\begin{equation}\label{eq:even_fully_separated}
    j_2 \ge j_1 + 2, ~ j_3\ge j_2 + 2, ~ j_4 \ge j_3 + 2, \dots, j_t \ge j_{t-1} + 2.
\end{equation}

\begin{figure}
    \centering
    \begin{subfigure}{.85\linewidth}
        \centering
        \includegraphics[scale=0.85]{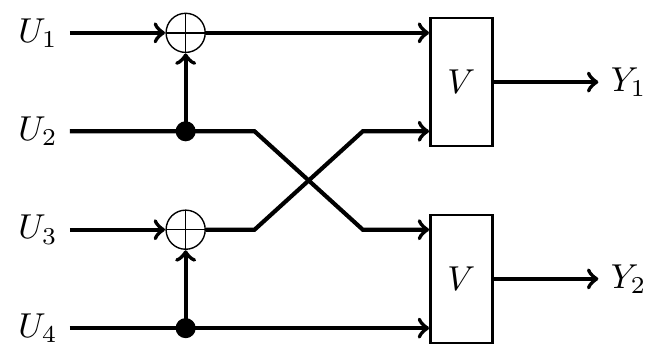}
        \caption{ DB polar transform induces three adjacent-bits-channels: (1) $V^{\oria}:U_1,U_2\to Y_1,Y_2$; (2) $V^{\orib}:U_2,U_3\to U_1,Y_1,Y_2$; (3) $V^{\oric}:U_3,U_4\to U_1,U_2,Y_1,Y_2$.}
    \end{subfigure}

    \begin{subfigure}{.85\linewidth}
        \centering
        \includegraphics[scale = 0.85]{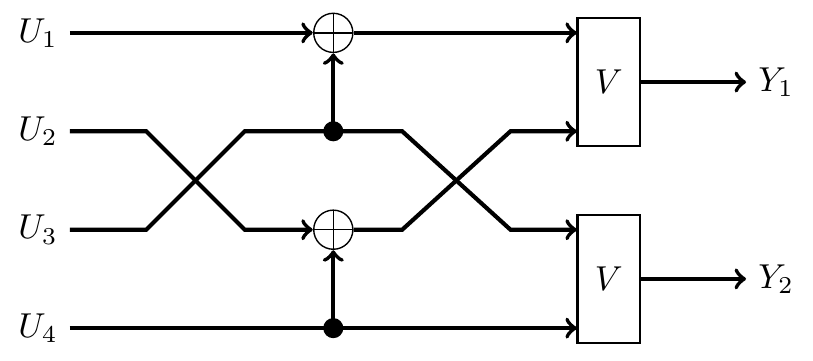}
        \caption{SDB polar transform induces three adjacent-bits-channels: (1) $V^{\swpa}:U_1,U_2\to Y_1,Y_2$; (2) $V^{\swpb}:U_2,U_3\to U_1,Y_1,Y_2$; (3) $V^{\swpc}:U_3,U_4\to U_1,U_2,Y_1,Y_2$.}
    \end{subfigure}
    \caption{The Double-Bits (DB) polar transform and the Swapped-Double-Bits (SDB) polar transform. $U_1,U_2,U_3,U_4$ are i.i.d. Bernoulli-$1/2$ random variables. }
    \label{fig:DB&SDB}
\end{figure}

The DB polar transform and the SDB polar transform (see Fig. \ref{fig:DB&SDB} for their definitions) were introduced in \cite[Section~III]{Li2022TIT} to describe the recursive relation for ABS polar codes.
In this paper, we introduce a new transform called the  Added-Double-Bits (ADB) polar transform to characterize the recursive relation for ABS+ polar codes.
The DB polar transform corresponds to applying the identity transform on a pair of adjacent bits, the SDB polar transform corresponds to applying the swapping transform on a pair of adjacent bits, and the ADB polar transform corresponds to applying the Ar{\i}kan transform on a pair of adjacent bits.

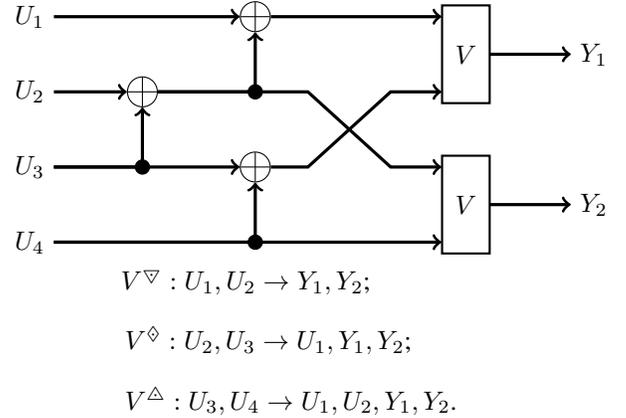
\begin{figure}
    \centering
    \begin{subfigure}{0.85\linewidth}
        \centering
        \begin{tikzpicture}
            \draw
            node at (-1.5,10.5) [] (u1)  {$U_1$}
            node at (-1.5, 9.5) [] (u2)  {$U_2$}
            node at (-1.5, 8.5) [] (u3)  {$U_3$}
            node at (-1.5, 7.5) [] (u4)  {$U_4$};

            \draw
            node at (0.0, 9.5) [XOR, scale=1.2] (x2) {};
            \draw[fill] (0.0, 8.5) circle (.6ex);

            \draw
            node at (1.5,10.5) [XOR,scale=1.2] (x1) {};
            \draw[fill] (1.5, 9.5) circle (.6ex);

            \draw
            node at (1.5,8.5) [XOR,scale=1.2] (x3) {};
            \draw[fill] (1.5, 7.5) circle (.6ex);

            \draw
            node at (4.1,10.5) (v1) {}
            node at (4.1, 9.5) (v2) {}
            node at (4.1, 8.5) (v3) {}
            node at (4.1, 7.5) (v4) {}
            node at (4.5,10)          (t1)  {}
            node at (4.3,10) [vblock] (76)  {$V$}
            node at (6,  10) []       (y1)  {$Y_1$}
            node at (4.5, 8)          (t3)  {}
            node at (4.3, 8) [vblock] (83)  {$V$}
            node at (6,   8) []       (y3)  {$Y_2$};

            \draw[very thick,->](u1) -- node {}(x1);
            \draw[very thick,->](x1) -- (v1);

            \draw[very thick,->](u2) -- node {}(x2);
            \draw[very thick,->](x2) -- (2.2, 9.5) -- (3.3, 8.5) -- (v3);

            \draw[very thick,->](u3) -- node {}(x3);
            \draw[very thick,->](x3) -- (2.2,8.5) -- (3.3,9.5) -- (v2);

            \draw[very thick,->](u4) -- (v4);

            \draw[very thick,->](x2) -| node {}(x1);
            \draw[very thick,->](u4) -| node {}(x3);
            \draw[very thick,->](u3) -| node {}(x2);

            \draw[very thick,->](t1) -- node {}(y1);
            \draw[very thick,->](t3) -- (y3);
        \end{tikzpicture}
    \end{subfigure}
    \begin{subfigure}{0.85\linewidth}
        \centering
        \begin{tikzpicture}
            \draw node at (-0.6, 0.8)  [align=left] () { $V^{\adda}:U_1,U_2\to Y_1,Y_2$;};
            \draw node at (-0.3,   0)  [align=left] () { $V^{\addb}:U_2,U_3\to U_1,Y_1,Y_2$;};
            \draw node at (   0,-0.8)  [align=left] () { $V^{\addc}:U_3,U_4\to U_1,U_2,Y_1,Y_2$.};
        \end{tikzpicture}
    \end{subfigure}
    \caption{The Added-Double-Bits (ADB) polar transform.
    Two copies of adjacent-bits-channel $V$ are transformed into three adjacent-bits-channels $V^{\adda},V^{\addb},V^{\addc}$, where $U_1,U_2,U_3,U_4$ are i.i.d. Bernoulli-$1/2$ random variables.
    The name ``Added-Double-Bits" refers to the addition between $U_2$ and $U_3$.
    }
    \label{fig:ADB}
\end{figure}

The details of the ADB polar transform are illustrated in Fig. \ref{fig:ADB}. Given an adjacent-bits-channel $V:\{0,1\}^2\rightarrow \cY$, the transition probabilities of $V^{\adda}: \{0,1\}^2\rightarrow \cY^2$, $V^{\addb}:\{0,1\}^2 \rightarrow \{0,1\} \times \cY^2$, and $V^{\addc}:\{0,1\}^2 \rightarrow \{0,1\}^2\times \cY^2$ in Fig. \ref{fig:ADB} are given by
\begin{equation} \label{eq:transition_ADB}
    \begin{aligned}
         & V^{\adda}(y_1,y_2|u_1,u_2)=\frac{1}{4}\sum_{u_3,u_4\in\{0,1\}} V(y_1|u_1+ u_2+ u_3, \\ &\hspace*{1.5in} u_3+ u_4) V(y_2|u_2+ u_3,u_4)\\&\hspace*{1.0in} \text{~for~} u_1,u_2\in\{0,1\} \text{~and~} y_1,y_2\in\cY ,                                      \\
         & V^{\addb}(u_1,y_1,y_2|u_2,u_3)=\frac{1}{4}\sum_{u_4\in\{0,1\}}V(y_1|u_1+ u_2+ u_3,  \\ &\hspace*{1.5in} u_3+ u_4) V(y_2|u_2+ u_3,u_4)  \\& \hspace*{0.8in} \text{~for~} u_1,u_2,u_3\in\{0,1\} \text{~and~} y_1,y_2\in\cY ,                                  \\
         & V^{\addc}(u_1,u_2,y_1,y_2|u_3,u_4)=\frac{1}{4} V(y_1|u_1+ u_2+ u_3,                 \\ &\hspace*{1.2in} u_3+ u_4) V(y_2|u_2+ u_3,u_4)\\&\hspace*{0.6in} \text{~for~} u_1,u_2,u_3,u_4\in\{0,1\} \text{~and~} y_1,y_2\in\cY .
    \end{aligned}
\end{equation}
Now we are ready to state the recursive relation between $\{V_i^{(n), \ABSP}\}_{i = 1}^{n-1}$ and $\{V_i^{(n/2), \ABSP}\}_{i = 1}^{n/2-1}$.

\begin{lemma}\label{lemma:recursive_relation}
    Let $n\ge 4$. We write $\mQ_n^{\ABSP}$ in the form of \eqref{eq:Q_n} and require that $\cI^{(n)} = \{2j_1, 2j_2, \dots, 2j_t\}$ satisfies \eqref{eq:even_fully_separated}.
    For $1\le j \le n/2 - 1$, we have the following results:

    \noindent
    {\em \bf Case i)} If $2j \in \cI_S^{(n)}$, then
    \begin{align*}
        V_{2j-1}^{(n),\ABSP} = (V_j^{(n/2),\ABSP})^\swpa, \\
        V_{2j}^{(n),\ABSP} = (V_j^{(n/2),\ABSP})^\swpb,   \\
        V_{2j+1}^{(n),\ABSP} = (V_j^{(n/2),\ABSP})^\swpc,
    \end{align*}

    \noindent
    {\em \bf Case ii)} If $2j \in \cI_A^{(n)}$, then
    \begin{equation*}
        \begin{aligned}
            V_{2j-1}^{(n),\ABSP} = (V_j^{(n/2),\ABSP})^\adda, \\
            V_{2j}^{(n),\ABSP} = (V_j^{(n/2),\ABSP})^\addb,   \\
            V_{2j+1}^{(n),\ABSP} = (V_j^{(n/2),\ABSP})^\addc,
        \end{aligned}
    \end{equation*}

    \noindent
    {\em \bf Case iii)} If $2(j-1)\in \cI_{S}^{(n)}, 2(j+1)\in \cI^{(n)}$. then
    $$V_{2j}^{(n),\ABSP} = (V_j^{(n/2),\ABSP})^\orib . $$

    \noindent
    {\em \bf Case iv)} If $2(j-1)\in \cI_{S}^{(n)}, 2(j+1)\notin \cI^{(n)}$, then
    $$ V_{2j}^{(n),\ABSP} = (V_j^{(n/2),\ABSP})^\orib, V_{2j+1}^{(n),\ABSP} = (V_j^{(n/2),\ABSP})^\oric . $$

    \noindent
    {\em \bf Case v)} If $2(j-1)\notin\cI_S^{(n)}, 2(j+1)\in \cI^{(n)}$. then
    $$ V_{2j-1}^{(n),\ABSP} = (V_j^{(n/2),\ABSP})^\oria, V_{2j}^{(n),\ABSP} = (V_j^{(n/2),\ABSP})^\orib . $$

    \noindent
    {\em \bf Case vi)} If $2(j-1)\notin \cI_S^{(n)}, 2j, 2(j+1)\notin \cI^{(n)}$, then
    \begin{align*}
        V_{2j-1}^{(n),\ABSP} = (V_j^{(n/2),\ABSP})^\oria, \\
        V_{2j}^{(n),\ABSP} = (V_j^{(n/2),\ABSP})^\orib,   \\
        V_{2j+1}^{(n),\ABSP} = (V_j^{(n/2),\ABSP})^\oric .
    \end{align*}
\end{lemma}
We omit the proof of Lemma \ref{lemma:recursive_relation} because it is similar to the proof of \cite[Lemma 1, Lemma 2]{Li2022TIT}.

For a given BMS channel $W$, the starting point of the recursive relation in Lemma \ref{lemma:recursive_relation} is $V^{(2), \ABSP}_1$, whose transition probabilities can be calculated as follows
\begin{equation}\label{eq:V_init}
    \begin{aligned}
        V_1^{(2),\ABSP}((y_1, y_2) | u_1, u_2) = W(y_1 | u_1 + u_2) W(y_2 | u_2) \\
        \hspace*{.4in}\text{for~} u_1,u_2\in\{0,1\} \text{~and~} y_1,y_2\in\cY.
    \end{aligned}
\end{equation}

After calculating the transition probabilities of the adjacent-bits-channels ${\{V_i^{(n), \ABSP}\}_{i = 1}^{n-1}}$, we can use \eqref{eq:v_to_w} to obtain the transition probabilities of the bit-channels ${\{W_i^{(n), \ABSP}\}_{i = 1}^{n}}$.
This allows us to calculate $\{H_i(\mG_n^{\ABSP}, W)\}_{i = 1}^n$ and determine which bits are information bits.

\subsection{Constructing the matrices $\mQ_2^{\ABSP}, \mQ_4^{\ABSP}, \dots, \mQ_n^{\ABSP}$}\label{sect:construct_Q}

We construct $\mQ_2^{\ABSP}, \mQ_4^{\ABSP}, \dots, \mQ_n^{\ABSP}$ one by one, starting from $\mQ_2^{\ABSP}$.
Therefore, the matrices $\mQ_2^{\ABSP}, \mQ_4^{\ABSP}, \dots, \mQ_{n/2}^{\ABSP}$ and $\mG_{n/2}^{\ABSP}$ are already known when we construct $\mQ_n^{\ABSP}$.
Lemma \ref{lemma:recursive_relation}  allows us to calculate the transition probabilities of the adjacent-bits-channels $\{V_{i}^{(n/2), \ABSP}\}_{i = 1}^{n/2-1}$ from $\mQ_2^{\ABSP}, \mQ_4^{\ABSP}, \dots, \mQ_{n/2}^{\ABSP}$, so we also know the transition probabilities of $\{V_{i}^{(n/2), \ABSP}\}_{i = 1}^{n/2-1}$ when constructing $\mQ_n^{\ABSP}$.
Constructing the matrix $\mQ_n^{\ABSP}$ is  equivalent to constructing the two sets $\cI_S^{(n)}$ and $\cI_A^{(n)}$ in \eqref{eq:Q_n}, whose elements are all even numbers.

According to the recursive relation $\mG_n^{\ABSP} = \mQ_n^{\ABSP}(\mG_{n/2}^{\ABSP}\otimes \mG_2^{\polar})$, our objective is to choose $\cI_S^{(n)}$ and $\cI_A^{(n)}$ whose corresponding matrix $\mQ_n^{\ABSP}$ minimizes $\Gamma(\mG_n^{\ABSP}, W)$ for a given $\mG_{n/2}^{\ABSP}$.
This is equivalent to maximizing $\Gamma(\mG_{n/2}^{\ABSP}\otimes \mG_2^{\polar}, W)-\Gamma(\mG_n^{\ABSP}, W)$.
To that end, let us introduce some notation.
Suppose that $V:\{0,1\}^2\to \cY$ is an adjacent-bits-channel.
Let two i.i.d. Bernoulli-$1/2$ random variables $U_1$ and $U_2$ be the inputs of $V$, and let $Y$ be the corresponding channel output.
Define
\begin{equation*}
    \gamma(V) = H(U_1| Y)(1 - H(U_1| Y)) + H(U_2|U_1, Y)(1- H(U_2|U_1, Y)).
\end{equation*}
Recall the definitions of $\{V_{i}^{(n), \ABSP}\}_{i = 1}^{n-1}$ and $\{\widehat V_{i}^{(n), \ABSP}\}_{i = 1}^{n-1}$ in Fig. \ref{fig:structure_of_ABS+}.
It is easy to see that
\begin{equation*}
    \begin{aligned}
         & \Gamma(\mG_n^{\ABSP}, W)                                                                                 \\= &\cfrac1n \Big[
        H_1(\mG_n^{\ABSP}, W)(1-H_1(\mG_n^{\ABSP}, W))                                                              \\
         & + \sum_{j = 1}^{n/2-1}\gamma(V_{2j}^{(n), \ABSP})                                                        \\& + H_{n}(\mG_n^{\ABSP}, W)(1-H_{n}(\mG_n^{\ABSP}, W))
        \Big],                                                                                                      \\
         & \Gamma(\mG_{n/2}^{\ABSP}\otimes \mG_2^{\polar}, W)                                                       \\= &\cfrac1n \Big[
        H_1(\mG_{n/2}^{\ABSP}\otimes \mG_2^{\polar}, W)(1-H_1(\mG_{n/2}^{\ABSP}\otimes \mG_2^{\polar}, W))          \\
         & + \sum_{j = 1}^{n/2-1}\gamma(\widehat V_{2j}^{(n), \ABSP})                                               \\
         & + H_{n}(\mG_{n/2}^{\ABSP}\otimes \mG_2^{\polar}, W)(1-H_{n}(\mG_{n/2}^{\ABSP}\otimes \mG_2^{\polar}, W))
        \Big].
    \end{aligned}
\end{equation*}
By the definition of $H_i(\mG^{\ABSP}_n ,W)$ and since $1, n \notin \cI^{(n)}$, we know that
$H_1(\mG_n^{\ABSP}, W) =H_1(\mG_{n/2}^{\ABSP}\otimes \mG_2^{\polar}, W)$ and $H_{n}(\mG_n^{\ABSP}, W) =H_{n}(\mG_{n/2}^{\ABSP}\otimes \mG_2^{\polar}, W)$.
Therefore,
\begin{equation}\label{eq:score_Qn_1}
    \begin{aligned}
         & \Gamma(\mG_{n/2}^{\ABSP}\otimes \mG_2^{\polar}, W)-\Gamma(\mG_n^{\ABSP}, W) \\= &\cfrac1n \sum_{j=1}^{n/2-1} \gamma(\widehat V_{2j}^{(n), \ABSP})-\gamma(V_{2j}^{(n), \ABSP})
    \end{aligned}
\end{equation}
Lemma \ref{lemma:recursive_relation} implies that
\begin{align*}
     & \widehat V_{2j}^{(n), \ABSP} = (V_{j}^{(n/2), \ABSP})^{\orib}, \\
     & V_{2j}^{(n), \ABSP} =
    \begin{cases}
        (V_{j}^{(n/2), \ABSP})^\swpb & \text{if~} 2j\in \cI_S^{(n)}, \\
        (V_{j}^{(n/2), \ABSP})^\addb & \text{if~} 2j\in \cI_A^{(n)}, \\
        (V_{j}^{(n/2), \ABSP})^\orib & \text{otherwise.}
    \end{cases}
\end{align*}
Taking this into \eqref{eq:score_Qn_1}, we obtain that
\begin{equation}\label{eq:score_Qn_2}
    \begin{aligned}
          & \Gamma(\mG_{n/2}^{\ABSP}\otimes \mG_2^{\polar}, W)-\Gamma(\mG_n^{\ABSP}, W)                                            \\
        = & \cfrac 1n \Big[ \sum_{2j\in \cI_S^{(n)}} (\gamma((V_{j}^{(n/2), \ABSP})^{\orib})-\gamma((V_{j}^{(n/2), \ABSP})^\swpb)) \\
        + & \hspace*{.15in} \sum_{2j\in \cI_A^{(n)}} (\gamma((V_{j}^{(n/2), \ABSP})^{\orib})-\gamma((V_{j}^{(n/2), \ABSP})^\addb))
        \Big].
    \end{aligned}
\end{equation}
Since our objective is to maximize the right-hand side of \eqref{eq:score_Qn_2}, for every $2j \in \cI^{(n)} = \cI_S^{(n)}\cup \cI^{(n)}_A$, we have
\begin{equation} \label{eq:divide_cI}
    2j \in
    \begin{cases}
        \cI^{(n)}_S & \text{if~}  \gamma\left((V_{i}^{(n/2),\ABSP})^\swpb\right) \le  \gamma\left((V_{i}^{(n/2),\ABSP})^\addb\right), \\
        \cI^{(n)}_A & \text{if~}  \gamma\left((V_{i}^{(n/2),\ABSP})^\addb\right)  <   \gamma\left((V_{i}^{(n/2),\ABSP})^\swpb\right). \\
    \end{cases}
\end{equation}
Therefore, to construct  $\cI_S^{(n)}$ and $\cI_A^{(n)}$, we only need to find their union $\cI^{(n)}$.
Next we define the function
\begin{equation}\label{eq:score_function}
    \begin{aligned}
         & \mathtt{score}(j)                                                                                            \\=&\max\Big\{ \gamma\left((V_{j}^{(n/2),\ABSP})^\orib\right)-\gamma\left((V_{j}^{(n/2),\ABSP})^\swpb\right), \\
         & \hspace*{.35in}\gamma\left((V_{j}^{(n/2),\ABSP})^\orib\right)-\gamma\left((V_{j}^{(n/2),\ABSP})^\addb\right)
        \Big\}
    \end{aligned}
\end{equation}
for $1\le j\le n/2-1$.
Taking \eqref{eq:divide_cI}, \eqref{eq:score_function} into \eqref{eq:score_Qn_2}, we obtain that
\begin{equation}\label{eq:score_Qn_3}
    \begin{aligned}
         & \Gamma(\mG_{n/2}^{\ABSP}\otimes \mG_2^{\polar}, W)-\Gamma(\mG_n^{\ABSP}, W) \\= &\cfrac 1n \sum_{2j\in \cI^{(n)}} \mathtt{score}(j).
    \end{aligned}
\end{equation}
Therefore, we need to find $\cI^{(n)}$ to maximize the right-hand side of \eqref{eq:score_Qn_3} under the constraint \eqref{eq:even_fully_separated},
i.e., we need to solve the following optimization problem:
\begin{equation}\label{eq:optimization}
    \begin{aligned}
        \cI^{(n)} = & \argmax_{\cS\subseteq \{2, 4, \dots, n-2\}}\sum_{2j\in \cS}\mathtt{score}(j) \\
                    & \text{s.t.~} |j_1-j_2|\ge 2 \text{~for all distinct~} 2j_1, 2j_2\in \cS.
    \end{aligned}
\end{equation}
This problem can be solved using a dynamic programming method with time complexity $O(n)$.
More precisely, for $k\in \{2, 4, \dots, n-2\}$, we define
\begin{equation}
    \begin{aligned}
        \cI^{(n)}_k = & \argmax_{\cS\subseteq \{2, 4, \dots, k\}}\sum_{2j\in \cS}\mathtt{score}(j) \\
                      & \text{s.t.~} |j_1-j_2|\ge 2 \text{~for all distinct~} 2j_1, 2j_2\in \cS.
    \end{aligned}
\end{equation}
Note that $\cI^{(n)} = \cI^{(n)}_{n-2}$.
The sets $\cI_2^{(n)}, \cI^{(n)}_4, \dots, \cI^{(n)}_{n-2}$ can be calculated from the following recursive relation
\begin{equation*}
    \cI_{k+2}^{(n)} =
    \begin{cases}
        \cI_{k-2}^{(n)} \cup \{k+2\} & \text{if~}\sum_{2j\in \cI_{k-2}^{(n)} \cup \{k+2\}} \mathtt{score}(j) \\
                                     & >\sum_{2j\in\cI_k^{(n)}} \mathtt{score}(j)                            \\
        \cI_k^{(n)}                  & \text{otherwise}.
    \end{cases}
\end{equation*}
The starting point of this recursive relation is
\begin{equation}
    \cI^{(n)}_0 = \emptyset, \qquad
    \cI^{(n)}_2 = \begin{cases}
        \{2\}     & \text{if~} \mathtt{score}(1)  >  0 \\
        \emptyset & \text{otherwise}.
    \end{cases}
\end{equation}
In this way, we solve the optimization problem \eqref{eq:optimization}.
Finally, we use \eqref{eq:divide_cI} to obtain $\cI^{(n)}_S$ and $\cI^{(n)}_A$.

\subsection{Summary of the ABS+ polar code construction}\label{sect:cons_algo}
In the previous subsections, we describe two main ingredients of the ABS+ polar code construction.
The first ingredient is the method to recursively calculate the transition probabilities of $\{V_i^{(n), \ABSP}\}_{i = 1}^{n-1}$ when $\mQ_2^{\ABSP}, \mQ_4^{\ABSP}, \dots, \mQ_n^{\ABSP}$ are known.
The second ingredient is the algorithm to construct the matrix $\mQ_n^{\ABSP}$ when the transition probabilities of $\{V_i^{(n/2),\ABSP}\}_{i=1}^{n/2-1}$ are available.
Moreover, we also need to quantize the output alphabets using Algorithm 1 in \cite{Li2022TIT} to ensure that the output alphabet size of $\{V_{i}^{(n), \ABSP}\}_{i = 1}^{n-1}$ does not increase exponentially in $n$.
Below we put everything together and summarize the ABS+ polar code construction in Algorithm \ref{algo:ABS+_Construction}.
\begin{algorithm}
    \DontPrintSemicolon
    \caption{\texttt{ABS+Construct}$(n,k,W)$}
    \label{algo:ABS+_Construction}
    \KwIn{code length $n=2^m\ge 4$, code dimension $k$, and the BMS channel $W$}
    \KwOut{the  matrices $\mQ_2^{\ABSP}, \mQ_4^{\ABSP}, \dots, \mQ_n^{\ABSP}$, and the index set $\cA$ of the information bits}

    Quantize the output alphabet of $W$ using the method in \cite{Tal13}
    \Comment{This step is needed when the output alphabet size of $W$ is very large \cite[Section III]{Tal13} or when $W$ has a continuous output alphabet \cite[Section VI]{Tal13}.}

    Set $\mathbf{Q}_2^{\ABSP}$ to be the $2\times 2$ identity matrix

    Calculate the transition probability of $V_1^{(2),\ABSP}$ from $W$ using \eqref{eq:V_init}

    Quantize the output alphabet of $V_1^{(2),\ABSP}$ using \cite[Algorithm 1]{Li2022TIT}

    \For{$n_c=4,8,16,\dots,n$}
    {

    Construct $\mQ_{n_c}^{\ABSP}$ from $\{V_i^{(n_c/2),\ABSP}\}_{i=1}^{n_c/2-1}$using the method in Section \ref{sect:construct_Q}

    Calculate the transition probabilities of $\{V_i^{(n_c),\ABSP}\}_{i=1}^{n_c-1}$ from $\mQ_{n_c}^{\ABSP}$ and $\{V_i^{(n_c/2),\ABSP}\}_{i=1}^{n_c/2-1}$ using Lemma~\ref{lemma:recursive_relation}

    Quantize the output alphabets of $\{V_i^{(n_c),\ABSP}\}_{i=1}^{n_c-1}$ using \cite[Algorithm 1]{Li2022TIT}
    }

    Calculate the transition probabilities of $\{W_i^{(n),\ABSP}\}_{i=1}^n$ from the transition probabilities of $\{V_i^{(n),\ABSP}\}_{i=1}^{n-1}$.

    Sort the capacity of the bit-channels $\{W_i^{(n),\ABSP}\}_{i=1}^n$ to obtain $I(W_{i_1}^{(n),\ABSP})\ge I(W_{i_2}^{(n),\ABSP})\ge \dots\ge I(W_{i_n}^{(n),\ABSP})$, where $\{i_1,i_2,\dots,i_n\}$ is a permutation of $\{1,2,\dots,n\}$

    $\cA\gets\{i_1,i_2,\dots,i_k\}$

    \Return $\mQ_2^{\ABSP}, \mQ_4^{\ABSP}, \dots, \mQ_n^{\ABSP},\cA$

\end{algorithm}

\section{The encoding algorithm for ABS+ polar codes}\label{sect:encoding}
In this section, we describe the encoding algorithm of ABS+ polar codes and give an example of an ABS+ polar code with code length $n = 16$.
We will also use this example to illustrate how our new SC decoder works in Section \ref{sect:decoding}.

Let $\cC$ be an $(n,k)$ ABS+ polar code defined by the matrices $\mQ_2^{\ABSP}, \mQ_4^{\ABSP}, \dots, \mQ_n^{\ABSP}$ (or  equivalently, defined by the sets $\cI_S^{(2)}, \cI_A^{(2)}, \cI_S^{(4)}, \cI_A^{(4)},\dots, \cI_S^{(n)}, \cI_A^{(n)}$).
Let $\cA = \{i_1, i_2, \dots, i_k\}$ be the index set of the information bits in $\cC$.
We present the encoding algorithm of the code $\cC$ in Algorithm \ref{algo:ABS+_Encoding}.

\begin{algorithm}
    \DontPrintSemicolon
    \caption{\texttt{ABS+Encode}$((m_1,m_2,\dots,m_k))$}
    \label{algo:ABS+_Encoding}
    \KwIn{the message vector $(m_1,m_2,\dots,m_k)\in\{0,1\}^k$}
    \KwOut{the codeword $(c_1,c_2,\dots,c_n)\in\{0,1\}^n$, where $n=2^m$ is the code length}

    Initialize $(c_1,c_2,\dots,c_n)$ as the all-zero vector

    $(c_{i_1},c_{i_2},\dots,c_{i_k})\gets (m_1,m_2,\dots,m_k)$
    \Comment{$i_1,i_2,\dots,i_k$ are the indices of the information bits.}

    \For{$i=0,1,2,3,\dots,m-1$}
    {
    $t\gets 2^i$

    $n_c\gets 2^{m-i}$

    \For{$h=1,2,3,\dots,t$}
    {
    \For{$j=1,2,3,\dots,n_c/2-1$}
    {
    \uIf{\em $2j\in \cI_S^{(n_c)}$}
    {
    $(c_{h + (2j-1)t}, ~c_{h + 2jt})\gets(c_{h + 2jt}, ~c_{h + (2j-1)t})$

    \Comment{Swapping transform}
    }
    \ElseIf{\em $2j\in \cI_A^{(n_c)}$}
    {
    $(c_{h + (2j-1)t}, ~c_{h + 2jt})\gets(c_{h + (2j-1)t} + c_{h + 2jt} , ~c_{h + 2jt})$

    \Comment{Ar{\i}kan transform}
    }

    }

    \For{$j=0,1,2,3,\dots,n_c/2-1$}
    {
    $(c_{h+2jt},~c_{h+(2j+1)t})\gets(c_{h+2jt} + c_{h+(2j+1)t},~c_{h+(2j+1)t})$

    }
    }
    }

    \Return  $(c_1,c_2,\dots,c_n)$
\end{algorithm}

\begin{proposition}
    The time complexity of Algorithm \ref{algo:ABS+_Encoding} is $O(n\log(n))$.
\end{proposition}
\begin{figure*}
    \centering
    \includegraphics[scale = 0.65]{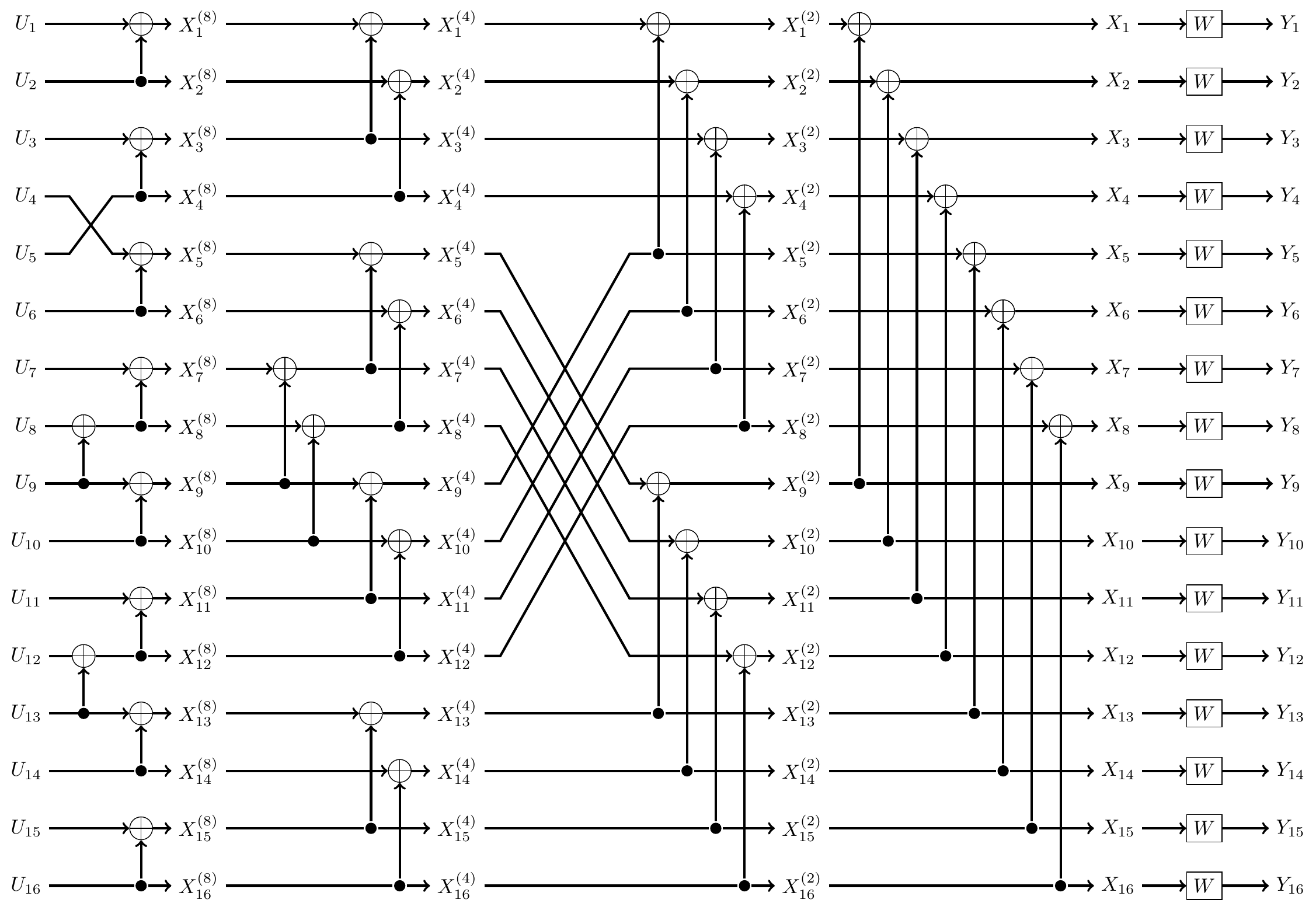}
    \caption{
    The $(16, 8)$ ABS+ polar code defined by the sets in \eqref{eq:example_parameter}.
    We swap $U_4$ and $U_5$ because $\cI_S^{(16)} = \{4\}$.
    We apply the Ar{\i}kan transform on the pairs $(U_8, U_9)$ and $(U_{12}, U_{13})$ because $\cI_A^{(16)} = \{8, 12\}$.
    We apply the Ar{\i}kan transform on the pairs $(X_7^{(8)},X_9^{(8)})$ and $(X_8^{(8)},X_{10}^{(8)})$ because $\cI_A^{(8)} = \{4\}$.
    We apply the swapping  transform on the pairs $(X_{5}^{(4)},X_{9}^{(4)})$, $(X_{6}^{(4)},X_{10}^{(4)})$, $(X_{7}^{(4)},X_{11}^{(4)})$ and $(X_{8}^{(4)},X_{12}^{(4)})$ because $\cI_S^{(4)} = \{2\}$.
    }
    \label{fig:example_enc}
\end{figure*}

Apart from Lines 7-11, the rest of Algorithm \ref{algo:ABS+_Encoding} is the same as the encoding algorithm of standard polar codes, whose time complexity is $O(n\log (n))$.
According to \eqref{eq:Q_n}, the operations in Lines 7-11 are equivalent to
\begin{align*}
          & (c_h, c_{h+t}, c_{h+2t}, \dots, c_{h+(n_c-1)t})                    \\
    \gets & (c_h, c_{h+t},  c_{h+2t}, \dots, c_{h+(n_c-1)t})\mQ_{n_c}^{\ABSP}.
\end{align*}
The fully separated requirement \eqref{eq:even_fully_separated} in our code construction guarantees that
each element in the vector $(c_h, c_{h+t}, c_{h+2t}, \dots, c_{h+(n_c-1)t})$ is involved in at most one swapping transform or one Ar{\i}kan transform.
Therefore, the number of operations in Lines 7-11 is no more than $n_c = 2^{m-i}$.
For each fixed value of $i$, Lines 7-11 are executed $t = 2^i$ times and induce at most $n_c \cdot t = n$ operations.
Since the value of $i$ ranges from $0$ to $\log(n)-1$ in Line 3, the total number of operations in Lines 7-11 is upper bounded by $n\log(n)$.
Thus we conclude that the time complexity of Algorithm \ref{algo:ABS+_Encoding} is $O(n\log(n))$.

Next we give a concrete example of an $(n = 16, k = 8)$ ABS+ polar code defined by the following sets:
\begin{equation}\label{eq:example_parameter}
    \begin{aligned}
         & \begin{array}{llll}
               \cI_S^{(2)} = \emptyset, & \cI_S^{(4)} = \{2\},      & \cI_S^{(8)} = ~\emptyset, & \cI_S^{(16)} = \{4\},     \\
               \cI_A^{(2)} = \emptyset, & \cI_A^{(4)} = ~\emptyset, & \cI_A^{(8)} = \{4\},      & \cI_A^{(16)} = \{8, 12\}, \\
           \end{array} \\
         & \cA = \{9, 10, 11, 12, 13, 14, 15, 16\}.
    \end{aligned}
\end{equation}
The encoding circuit of this specific ABS+ polar code is given in Fig. \ref{fig:example_enc}.

We can see from Fig. \ref{fig:example_enc} that the structure of ABS+ polar codes bears some resemblance to that of convolutional polar codes. Both approaches deepen the polarization level by applying invertible transforms to adjacent bits.
However, they differ in two aspects. Firstly, ABS+ polar codes consider all invertible transforms that can enhance the polarization level, whereas convolutional polar codes only employ the Ar{\i}kan transform on adjacent bits.
Additionally, ABS+ polar codes strictly limit the number of adjacent bit pairs participating in the invertible transforms, while convolutional polar codes apply the Ar{\i}kan transform to all adjacent bits regardless of whether it enhances the polarization level of the current pair of adjacent bits or not.
In terms of implementation, convolutional polar codes require tracking the joint distribution of successive three bits at each layer in order to apply the Ar{\i}kan transformation to all adjacent bits.
For ABS+ polar codes, if we also track the joint distribution of each successive three bits, we can ignore the fully separated requirement \eqref{eq:fully_separeted} and apply invertible transforms to more adjacent bits. This can further enhance the polarization level but also increase the algorithm's complexity.
If we refer to the ABS+ polar codes with ignored fully separated requirement \eqref{eq:fully_separeted} as extended ABS+ polar codes, then convolutional polar codes can be considered as a special case of extended ABS+ polar codes, as latter exhibit higher flexibility in determining when and which transform to apply.

\section{The SC decoding algorithm for ABS+ polar codes}\label{sect:decoding}

Although the SCL decoder is more widely used in practice, we will only describe the SC decoder in this paper for the sake of simplicity.
The method of extending the SC decoder to obtain the SCL decoder is well-known in the polar coding literature: see \cite{Tal15, Li2022TIT} for example.

We will first present a version of the  SC decoder with space complexity $O(n\log(n))$.
This version is relatively easy to understand.
Then in Section~\ref{sect:efficient_dec}, we present a space-efficient version with space complexity $O(n)$.

Recall that $(U_1, U_2, \dots, U_n)$ is the message vector, and $(X_1, X_2, \dots, X_n)$ is the codeword vector.
Following the example in Fig. \ref{fig:example_enc}, we define some intermediate vectors $\{(X_1^{(n_c)}, X_2^{(n_c)}, \dots, X_n^{(n_c)})\}_{n_c = 2, 4, \dots, n}$.
Let $(X_1^{(n)}, X_2^{(n)}, \dots, X_n^{(n)}) = (U_1, U_2, \dots, U_n)$.
The intermediate vectors are defined recursively from  $n_c = n/2$ to $n_c = 2$ using the following relation
\begin{equation}
    \begin{aligned}
         & (X_{1}^{~(n_c)}, X_{2}^{~(n_c)}, \dots, X_{n}^{~(n_c)}) \\= &(X_{1}^{(2n_c)}, X_{2}^{(2n_c)}, \dots, X_{n}^{(2n_c)}) \\&\cdot\big((\mQ_{2n_c}^{\ABSP} (\mI_{n_c}\otimes \mG_2^{\polar}))\otimes \mI_{n/(2n_c)}\big),
    \end{aligned}
\end{equation}

It is easy to see that the $n/n_c$ random vectors
\begin{align*}
    \Big\{( & X_{\beta}^{(n_c)}, X_{\beta + n/n_c}^{(n_c)}, X_{\beta + 2n/n_c}^{(n_c)} \dots, X_{\beta+(n_c-1)n/n_c}^{(n_c)}, \\&Y_\beta, Y_{\beta+n/n_c}, Y_{\beta+2n/n_c}, \dots, Y_{\beta + ({n_c-1})n/n_c})\Big\}_{\beta= 1}^{n/n_c}
\end{align*}
are independent and identically distributed.
For $1\le i\le n_c-1$ and $1\le \beta \le n/n_c$, we define two random vectors
\begin{equation}\label{eq:V_output}
    \begin{aligned}
         & X_{i,\beta}^{(n_c)} = (X_{\beta}^{(n_c)}, X_{\beta + n/n_c}^{(n_c)}, \dots, X_{\beta+(i-2)n/n_c}^{(n_c)}) \\
         & Y_{\beta}^{(n_c)} = (Y_\beta, Y_{\beta+n/n_c}, \dots, Y_{\beta + ({n_c-1})n/n_c}).
    \end{aligned}
\end{equation}
According to the definition in Fig. \ref{fig:structure_of_ABS+}, the channel mapping from $(X^{(n_c)}_{\beta + (i-1)n/n_c}, X^{(n_c)}_{\beta + in/n_c})$ to $(X_{i, \beta}^{(n_c)}, Y_{\beta}^{(n_c)})$ is the adjacent-bits-channel $V_{i}^{(n_c), \ABSP}$ for all $1\le \beta \le n/n_c$.
Below we omit ``ABS+" in the superscript and simply write $V_{i}^{(n_c), \ABSP}$ as $V_{i}^{(n_c)}$.

Let $(y_1, \dots, y_n)$ be a realization of  the channel output random vector, i.e., $(y_1, \dots, y_n)$ is the input to the SC decoder.
For $n_c = 2, 4, \dots, n$, let $(\hat x_1^{(n_c)}, \dots, \hat x_n^{(n_c)})$ be the decoding result of  $(X_1^{(n_c)}, \dots, X_n^{(n_c)})$ given by the SC decoder.
Similarly, let $(\hat u_1, \dots, \hat u_n)$ be the decoding result of  $(U_1, \dots, U_n)$.
For ABS+ polar codes, the SC decoder determines the value of $U_i$ from the conditional probabilities
\begin{equation}\label{eq:condi_proba}
    \begin{aligned}
          & \bP(U_1 = \hat u_1, \dots, U_{i-1} = \hat u_{i-1}, Y_1 = y_1, \dots, Y_n = y_n \\&\hspace*{1.4in}| U_i = u_i, U_{i+1} = u_{i+1}) \\
        = & V_{i}^{(n)} (\hat u_1, \dots, \hat u_{i-1}, y_1, \dots, y_n | u_i, u_{i+1}),   \\&\hspace*{1.8in} u_i, u_{i+1} \in \{0,1\}.
    \end{aligned}
\end{equation}
For $1\le i\le n_c-1$ and $1\le \beta \le n/n_c$, we write
\begin{equation}\label{eq:dec_V_output}
    \begin{aligned}
         & \hat{ \mathbi{x}}_{i,\beta}^{(n_c)} = (\hat x_{\beta}^{(n_c)}, \hat x_{\beta + n/n_c}^{(n_c)}, \dots, \hat x_{\beta+(i-2)n/n_c}^{(n_c)}) \\
         & {\mathbi{y}}_{\beta}^{(n_c)} = (y_\beta, y_{\beta+n/n_c}, \dots, y_{\beta + ({n_c-1})n/n_c}).
    \end{aligned}
\end{equation}
The SC decoder calculates the conditional probabilities in \eqref{eq:condi_proba} recursively from
\begin{equation}\label{eq:condi_proba_middle}
    \begin{aligned}
          & \bP(X_{i,\beta}^{(n_c)} = \hat{\mathbi{x}}_{i,\beta}^{(n_c)} ,Y_{\beta}^{(n_c)} = \mathbi{y}_{\beta}^{(n_c)}| X_{\beta + (i-1)n/n_c}^{(n_c)}= a, \\&\hspace*{1.8in} X_{\beta + in/n_c}^{(n_c)}= b) \\
        = & V_i^{(n_c)}(\hat{ \mathbi{x}}_{i,\beta}^{(n_c)},{\mathbi{y}}_{\beta}^{(n_c)}| a,b), \qquad a, b \in \{0,1\}.
    \end{aligned}
\end{equation}

For each $n_c\in\{2, 4, 8,\dots, n\}$, we use a data structure $\tP_{n_c}$ to store the probabilities in \eqref{eq:condi_proba_middle}.
More specifically, $\tP_{n_c}$ is a four-dimensional array with indices $i\in \{1, 2,\dots, n_c-1\}, \beta\in\{1,2, \dots, n/n_c\}, a\in\{0,1\}, b\in\{0,1\}$.
We write an entry in $\tP_{n_c}$ as  $\tP_{n_c}[  i, \beta ][a, b]$, which stores $V_i^{(n_c)}(\hat{ \mathbi{x}}_{i,\beta}^{(n_c)},{\mathbi{y}}_{\beta}^{(n_c)}| a,b)$, i.e.,
\begin{equation}\label{eq:tp}
    \tP_{n_c}[i,\beta][a,b] = V_i^{(n_c)}(\hat{ \mathbi{x}}_{i,\beta}^{(n_c)},{\mathbi{y}}_{\beta}^{(n_c)}| a,b).
\end{equation}
We omit $\hat{ \mathbi{x}}_{i,\beta}^{(n_c)}$ and ${\mathbi{y}}_{\beta}^{(n_c)}$ in the notation $\tP_{n_c}[i,\beta][a,b]$ because they remain unchanged in the whole decoding procedure.

We use another data structure $\tB_{n_c}$ to store the decoding results of the intermediate vector $(X_{1}^{(n_c)}, X_{2}^{(n_c)},\linebreak[4] \dots, X_{n}^{(n_c)})$.
The data structure $\tB_{n_c}$ is a two-dimensional array with indices $i\in \{1, 2,\dots, n_c\}, \beta\in\{1,2, \dots, n/n_c\}$.
We write an entry in $\tB_{n_c}$ as $\tB_{n_c}[i,\beta]$, which stores $\hat{x}_{\beta + (i-1)n/n_c}^{(n_c)}$, i.e.,
\begin{equation}\label{eq:tb_tf}
    \tB_{n_c}[i,\beta] = \hat{x}_{\beta + (i-1)n/n_c}^{(n_c)}.
\end{equation}

$\tP_{n_c} $ and $\tB_{n_c}$ are the only two data structures we need in the SC decoder.
The number of entries in $\tP_{n_c}$ is $4(n-n/n_c)$,
and the number of entries in $\tB_{n_c}$ is $n$.
Since $n_c$ takes $\log(n)$ values, the space complexity of the SC decoder is $O(n\log(n))$.
In Section~\ref{sect:efficient_dec}, we will show how to reduce the space complexity to $O(n)$.

Algorithm \ref{algo:ABS+_Decoding} outlines three main steps of the SC decoder.
As an initialization, we calculate all the entries in the array $\tP_2$ in Lines 1-2,
where the  formula in Line 2 follows from \eqref{eq:V_init} and \eqref{eq:tp}.
In Line 3, the recursive function $\mathtt{decode\_channel}$ with input parameters $(n_c=2, i=1)$ uses the probabilities in the array $\tP_2$ to obtain the decoding results of the intermediate vector $(X_{1}^{(2)}, X_{2}^{(2)}, \dots, X_{n}^{(2)})$. The parameters here refer to the first traversed adjacent-bits-channel $V^{(2),\ABSP}_1$.
As indicated in \eqref{eq:tb_tf}, the decoding results $(\hat{x}_{1}^{(2)}, \hat{x}_{2}^{(2)}, \dots, \hat{x}_{n}^{(2)})$ are stored in $\tB_2$.
The last step is to  obtain $(\hat{x}_{1}, \hat{x}_{2}, \dots, \hat{x}_{n})$ from $(\hat{x}_{1}^{(2)}, \hat{x}_{2}^{(2)}, \dots, \hat{x}_{n}^{(2)})$ in Lines~4-6,
where the calculation in Lines 5-6 follows from \eqref{eq:tb_tf} and
$$
    (X_1, X_2, \dots, X_n) = (X_{1}^{(2)}, X_{2}^{(2)}, \dots, X_{n}^{(2)}) (\mG_2^{\polar}\otimes \mI_{n/2}).
$$

\begin{algorithm}
    \DontPrintSemicolon
    \caption{\texttt{ABS+Decode}$((y_1,y_2,\dots,y_n))$}
    \label{algo:ABS+_Decoding}
    \KwIn{the received vector $(y_1,y_2,\dots,y_n)\in\cY^n$}
    \KwOut{the decoded  codeword $(\hat x_1,\hat x_2,\dots,\hat x_n)\in\{0,1\}^n$}

    \For{\em $\beta\in\{1, 2,\dots, n/2\}, a\in \{0,1\}$ and $b\in\{0,1\}$}{
        $\tP_2[  1, \beta ][  a, b ]\gets W(y_\beta|a+b)\cdot W(y_{\beta+n/2}|b)$
    }

    \texttt{decode\_channel$(2, 1)$}
    \Comment{Recursive decoding}

    \For{$\beta=1,2,\dots,n/2$}
    {
        $\hat x_\beta \gets \tB_{2}[  1, \beta] + \tB_{2}[  2, \beta]$

        $\hat x_{\beta+n/2} \gets \tB_{2}[  2, \beta]$
    }

    \Return  $(\hat x_1,\hat x_2,\dots,\hat x_n)$
\end{algorithm}

Next we explain how the function \texttt{decode\_channel} in Algorithm \ref{algo:decode_channel} works.
For each $1\le i\le n_c-1$, we define a subarray $\tP_{n_c}[i]$ of $\tP_{n_c}$ as
\begin{equation}\label{eq:pnci}
    \begin{aligned}
         & \tP_{n_c}[i] = (\tP_{n_c}[i,\beta][a,b],                                 \\
         & \hspace*{0.5in}\beta \in \{1,2,\dots, n/n_c\}, a\in\{0,1\},b\in\{0,1\}).
    \end{aligned}
\end{equation}
In the whole decoding algorithm, we always calculate all the entries in $\tP_{n_c}[i]$ before we call the function \texttt{decode\_channel} with input parameters $n_c$ and $i$.
For example, in Algorithm \ref{algo:ABS+_Decoding}, we calculate the array $\tP_{2}[1] = \tP_{2}$ in Lines 1-2 before calling the function $\texttt{decode\_channel}(2,1)$ in Line 3;
In Algorithm \ref{algo:decode_swp_channel}, we calculate $\tP_{2n_c}[2i]$ in Lines 4-5 before calling $\texttt{decode\_channel}(2n_c, 2i)$ in Line 5;
In Algorithm \ref{algo:decode_ori_channel}, we first calculate $\tP_{2n_c}[2i-1]$ in Lines 2-3 and then call $\texttt{decode\_channel}(2n_c, 2i-1)$ in Line 4.
If $1\le i\le n_c-2$, \texttt{decode\_channel}$(n_c, i)$ uses $\tP_{n_c}[i]$ to decode ${X}_{\beta + (i-1)n/n_c}^{(n_c)}$, $1\le \beta \le n/n_c$ and stores the decoding result $\hat{x}^{(n_c)}_{\beta + (i-1)n/n_c}$ in $\tB_{n_c}[i,\beta]$.
If $i = n_c-1$, \texttt{decode\_channel}$(n_c, i)$ uses $\tP_{n_c}[n_c-1]$ to decode ${X}_{\beta + (n_c-2)n/n_c}^{(n_c)}, {X}_{\beta + (n_c-1)n/n_c}^{(n_c)}$, $1\le \beta \le n/n_c$ and stores the decoding results $\hat{x}_{\beta + (n_c-2)n/n_c}^{(n_c)}$, $\hat{x}_{\beta + (n_c-1)n/n_c}^{(n_c)}$ in $\tB_{n_c}[n_c-1, \beta]$, $\tB_{n_c}[n_c, \beta]$.

\begin{algorithm}
    \DontPrintSemicolon
    \caption{\texttt{decode\_channel$(n_c, i)$}}
    \label{algo:decode_channel}
    \KwIn{$n_c\in \{2,4,8,\dots, n\}$ and index $i$, $1\le i \le n_c-1$, which together indentify the adjacent-bits-channel $V^{(n_c),\ABSP}_i$.}

    \uIf{$n_c = n$}{
        \texttt{decode\_boundary\_channel}$(i)$
    }
    \ElseIf{$2i\notin \cI^{(2n_c)}$}{
        \texttt{decode\_original\_channel}$(n_c, i)$

        \Comment{Recall that $\cI^{(2n_c)} = \cI_{S}^{(2n_c)} \cup \cI_{A}^{(2n_c)}$}
    }\ElseIf{$2i\in \cI_S^{(2n_c)}$}{
        \texttt{decode\_swapped\_channel}$(n_c, i)$
    }\ElseIf{$2i\in \cI_A^{(2n_c)}$}{
        \texttt{decode\_added\_channel}$(n_c, i)$
    }

    \Return

\end{algorithm}

\begin{algorithm}
    \DontPrintSemicolon
    \caption{\texttt{decode\_boundary\_channel$(i)$}}
    \label{algo:decode_boundary_channel}
    \KwIn{index $i$ in the last layer $(n_c = n)$}

    \eIf{$i\le n-2$}{
        \Comment{Only decode $U_i$}

        \eIf{$i\in\cA$}{
            $\tB_{n}[  i, 1 ] \gets \argmax_{a \in \{0,1\}}\sum_{b\in\{0,1\}} \tP_{n}[  i, 1 ][  a,b ]$

            \Comment{$U_i$ is an information bit}
        }{
            $\tB_{n}[  i, 1 ] \gets$ frozen value of $U_i$

            \Comment{$U_i$ is a frozen bit}
        }

    }{
        \Comment{Decode both $U_{n-1}$ and $U_n$.}

        \uIf{$n-1, n\notin \cA$}{
            $(\tB_{n}[n-1, 1], \tB_{n}[n,1]) \gets $

            \hspace*{1.2in} frozen bits $(U_{n-1}, U_{n})$

            \Comment{$U_{n-1}$ and $U_{n}$ are both frozen bits}
        }\ElseIf{$n-1\in \cA$, $n\notin \cA$}{
            $\tB_{n}[n,1] \gets $ frozen value of  $U_{n}$

            $\tB_{n}[n-1, 1] \gets \argmax_{a\in\{0,1\}} \tP_{n}[  n-1, 1 ][  a, U_{n} ]$

            \Comment{information bit $U_{n-1}$, frozen bit $U_{n}$ }

        }\ElseIf{$n-1\notin \cA$, $n\in \cA$}{
            $\tB_{n}[n-1, 1 ] \gets $ frozen value of $U_{n-1}$

            $\tB_{n}[n,1] \gets \argmax_{b \in\{0,1\}} \tP_{n}[  n-1, 1 ][  U_{n-1}, b ]$

            \Comment{frozen bit $U_{n-1}$, information bit $U_{n}$}
        }\Else{
            $(\tB_{n}[n-1, 1 ], \tB_{n}[n,1]) \gets \argmax_{(a,b)\in \{0,1\}^2} \tP_{n}[  n-1, 1 ][  a, b ]$

            \Comment{$U_{n-1}$ and $U_{n}$ are both information bits}
        }

    }
    \Return
\end{algorithm}

The implementation of $\texttt{decode\_channel}$ is divided into four cases.
The first case $n_c = n$ is the boundary case, where we can directly decode $U_i$ (and $U_{i+1}$ if $i = n-1$) from $\tP_{n}[i]$; see Algorithm \ref{algo:decode_boundary_channel}.
Note that we do not utilize the frozen value $U_{i+1}$ when decoding the message bit $U_i$ for $i < n_c-1$. This decision is based on the observation that it does not lead to any significant differences in the decoding performance.
In the other three cases, we decode ${X}_{\beta + (i-1)n/n_c}^{(n_c)}$ (and ${X}_{\beta + in/n_c}^{(n_c)}$ if $i = n_c-1$), $1\le \beta \le n/n_c$ from $\tP_{n_c}[i]$ in a recursive way.
Below we explain these three cases separately.

\begin{algorithm}
    \DontPrintSemicolon
    \caption{\texttt{calculate\_probability}\\\hspace*{2.0in}$(n_c, i, \beta, \mode, a,b)$}
    \label{algo:calcu_proba}
    \KwIn{$n_c = 2,4,,\dots, n/2$, $1\le i\le n_c-1$,  $1\le \beta\le n/(2n_c)$, $\mode \in \{\oria, \orib, \oric, \swpa, \swpb, \swpc, \adda, \addb, \addc\}$
        and $a, b\in \{0,1\}$
    }
    \KwOut{an entry in the array $\tP_{2n_c}$}

    $\beta' \gets \beta + n/(2n_c)$

    \uIf{    $\mode\in\{\oria, \swpa, \adda\}$}{
        $r_1 \gets a,\quad r_2 \gets b$

        \uIf{    $\mode = \oria$}{
            \Return $\frac14\sum_{r_3, r_4\in\{0,1\}} \tP_{n_c}[i,\beta][r_1+r_2, r_3+r_4]\tP_{n_c}[i,\beta'][r_2, r_4]$
            \Comment{$V_{2i-1}^{(2n_c)} = (V_{i}^{(n_c)})^{\oria}$}
        }\ElseIf{$\mode = \swpa$}{
            \Return $\frac14\sum_{r_3, r_4\in\{0,1\}} \tP_{n_c}[i,\beta][r_1+r_3, r_2+r_4]\tP_{n_c}[i,\beta'][r_3, r_4]$
            \Comment{$V_{2i-1}^{(2n_c)} = (V_{i}^{(n_c)})^{\swpa}$}
        }\ElseIf{$\mode = \adda$}{
            \Return $\frac14\sum_{r_3, r_4\in\{0,1\}} \tP_{n_c}[i,\beta][r_1+r_2+r_3, r_3+r_4]\tP_{n_c}[i,\beta'][r_2+r_3, r_4]$

            \Comment{$V_{2i-1}^{(2n_c)} = (V_{i}^{(n_c)})^{\adda}$}
        }
    }\ElseIf{$\mode\in\{\orib, \swpb, \addb\}$}{
        $r_1 \gets \tB_{2n_c}[2i-1,\beta], \quad r_2 \gets a,\quad r_3 \gets b$

        \uIf{    $\mode = \orib$}{
            \Return $\frac14\sum_{r_4\in\{0,1\}} \tP_{n_c}[i,\beta][r_1+r_2, r_3+r_4]\tP_{n_c}[i,\beta'][r_2, r_4]$
            \Comment{$V_{2i  }^{(2n_c)} = (V_{i}^{(n_c)})^{\orib}$}
        }\ElseIf{$\mode = \swpb$}{
            \Return $\frac14\sum_{r_4\in\{0,1\}} \tP_{n_c}[i,\beta][r_1+r_3, r_2+r_4]\tP_{n_c}[i,\beta'][r_3, r_4]$
            \Comment{$V_{2i  }^{(2n_c)} = (V_{i}^{(n_c)})^{\swpb}$}
        }\ElseIf{$\mode = \addb$}{
            \Return $\frac14\sum_{r_4\in\{0,1\}}  \tP_{n_c}[i,\beta][r_1+r_2+r_3, r_3+r_4]\tP_{n_c}[i,\beta'][r_2+r_3, r_4]$

            \Comment{$V_{2i  }^{(2n_c)} = (V_{i}^{(n_c)})^{\addb}$}
        }
    }\ElseIf{$\mode\in\{\oric, \swpc, \addc\}$}{
        $r_1 \gets \tB_{2n_c}[2i-1,\beta],\quad r_2 \gets \tB_{2n_c}[2i,\beta],\quad r_3 \gets a,\quad r_4 \gets b$

        \uIf{    $\mode = \oric$}{
            \Return $\frac14 \tP_{n_c}[i,\beta][r_1+r_2, r_3+r_4]\tP_{n_c}[i,\beta'][r_2, r_4]$

            \Comment{$V_{2i+1}^{(2n_c)} = (V_{i}^{(n_c)})^{\oric}$}
        }\ElseIf{$\mode = \swpc$}{
            \Return $\frac14 \tP_{n_c}[i,\beta][r_1+r_3, r_2+r_4]\tP_{n_c}[i,\beta'][r_3, r_4]$

            \Comment{$V_{2i+1}^{(2n_c)} = (V_{i}^{(n_c)})^{\swpc}$}
        }\ElseIf{$\mode = \addc$}{
            \Return $\frac14  \tP_{n_c}[i,\beta][r_1+r_2+r_3, r_3+r_4]\tP_{n_c}[i,\beta'][r_2+r_3, r_4]$

            \Comment{$V_{2i+1}^{(2n_c)} = (V_{i}^{(n_c)})^{\addc}$}
        }
    }
\end{algorithm}

By Lemma~\ref{lemma:recursive_relation}, $2i\in \cI_S^{(2n_c)}$ implies that $V_{2i-1}^{(2n_c)} = (V_i^{(n_c)})^{\swpa}, V_{2i}^{(2n_c)} = (V_i^{(n_c)}){^\swpb}, V_{2i+1}^{(2n_c)} = (V_i^{(n_c)}){^\swpc}$.
In this case, $\texttt{decode\_channel}$ calls the function $\texttt{decode\_swapped\_channel}$ in Algorithm~\ref{algo:decode_swp_channel}.
The first step in Algorithm~\ref{algo:decode_swp_channel} is to calculate $\tP_{2n_c}[2i-1], \tP_{2n_c}[2i]$ and $\tP_{2n_c}[2i+1]$ from $\tP_{n_c}[i]$ according to the above recursive relation; see Lines~2, 5, 8.
Note that we encapsulate the calculation of transition probabilities for adjacent-bits-channels within the function $\texttt{calculate\_probabilities}$, see Algorithm~\ref{algo:calcu_proba} for detail.
In Line~3, \texttt{decode\_channel}$(2n_c, 2i-1)$ uses $\tP_{2n_c}[2i-1]$ to decode $X^{(2n_c)}_{\beta + (2i-2)n/(2n_c)}, 1\le \beta \le n/(2n_c)$ and stores the decoding results in $\tB_{2n_c}[2i-1, \beta], 1\le \beta \le n/(2n_c) $.
Similarly, in Line~6, \texttt{decode\_channel}$(2n_c, 2i)$ uses $\tP_{2n_c}[2i]$ to decode $X^{(2n_c)}_{\beta + (2i-1)n/(2n_c)}, 1\le \beta \le n/(2n_c)$ and stores the decoding results in $\tB_{2n_c}[2i, \beta], 1\le \beta \le n/(2n_c) $.
If $i\le n_c-2$, then $\texttt{decode\_channel}(2n_c, 2i+1)$ in Line 9 only decodes one bit $X^{(2n_c)}_{\beta + 2in/(2n_c)}$ for each $\beta \in \{1, 2, \dots, n/(2n_c)\}$;
if $i = n_c-1$ (i.e., $2i+1 = 2n_c-1$),  then $\texttt{decode\_channel}(2n_c, 2i+1)$ decodes two bits $X^{(2n_c)}_{\beta + 2in/(2n_c)}, X^{(2n_c)}_{\beta + (2i+1)n/(2n_c)}$ for each $\beta \in \{1, 2, \dots, n/(2n_c)\}$.
In Line~3 and Line~6, we only decode one bit for each value $\beta$ because $2i-1<2i\le 2n_c-2$ for all $1\le i\le n_c-1$.
To summarize, after executing the first 9 lines of Algorithm~\ref{algo:decode_swp_channel}, we have the following decoding results stored in the array $\tB_{2n_c}$:
When $i\le n_c-2$, we have
\begin{equation}\label{eq:bx}
    \begin{aligned}
         & \tB_{2n_c}[2i-1, \beta] = \hat{x}_{\beta + (2i-2)n/(2n_c)}^{(2n_c)}, \\ &\tB_{2n_c}[2i  , \beta] = \hat{x}_{\beta + (2i-1)n/(2n_c)}^{(2n_c)}, \\
         & \tB_{2n_c}[2i+1, \beta] = \hat{x}_{\beta + (2i  )n/(2n_c)}^{(2n_c)}  \\ &\hspace*{1.5in}\text{~for~} 1\le \beta \le n/(2n_c).
    \end{aligned}
\end{equation}
When $i = n_c-1$, we have
\begin{equation}\label{eq:bxnc}
    \begin{aligned}
         & \tB_{2n_c}[2n_c-3, \beta] = \hat{x}_{\beta + (2n_c-4)n/(2n_c)}^{(2n_c)}, \\ & \tB_{2n_c}[2n_c-2  , \beta] = \hat{x}_{\beta + (2n_c-3)n/(2n_c)}^{(2n_c)},                                         \\
         & \tB_{2n_c}[2n_c-1, \beta] = \hat{x}_{\beta + (2n_c-2)n/(2n_c)}^{(2n_c)}, \\ & \tB_{2n_c}[2n_c,     \beta] = \hat{x}_{\beta + (2n_c-1)n/(2n_c)}^{(2n_c)} \quad\text{for~}1\le \beta \le n/(2n_c).
    \end{aligned}
\end{equation}
In the former case, we use the quantities in \eqref{eq:bx} to calculate $X^{(n_c)}_{\beta + (i-1)n/(n_c)}$, $1\le \beta \le n/n_c$ and store the results in $\tB_{n_c}[i,\beta]$, $1\le \beta \le n/n_c$;
see Lines~10-15.
In the latter case, we use the quantities in \eqref{eq:bxnc} to calculate $X^{(n_c)}_{\beta + (n_c-2)n/(n_c)}, X^{(n_c)}_{\beta + (n_c-1)n/(n_c)}$, $1\le \beta \le n/n_c$ and store the results in $\tB_{n_c}[n_c-1,\beta], \tB_{n_c}[n_c, \beta]$, $1\le \beta \le n/n_c$; see Lines~16-23.
In Fig.~\ref{fig:swp}, we further explain the calculations in Lines~13-15 and Lines~19-23.
\begin{figure}
    \centering
    \begin{subfigure}{0.85\linewidth}
        \centering
        \includegraphics[scale=0.95]{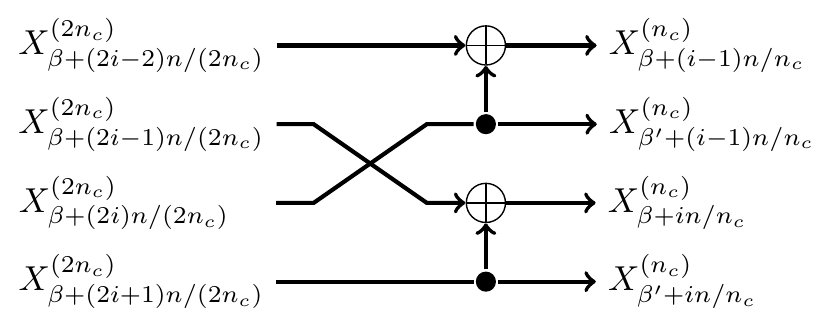}
        \caption{$2i \in \cI_{S}^{(2n_c)}$, $\beta' = \beta + n/(2n_c)$}
        \label{fig:swp}
    \end{subfigure}

    \begin{subfigure}{.85\linewidth}
        \centering
        \includegraphics[scale = 0.95]{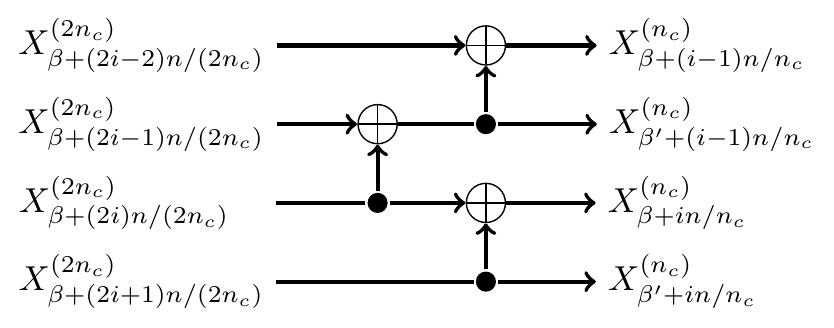}
        \caption{$2i\in \cI_{A}^{(2n_c)}$, $\beta' = \beta + n/(2n_c)$}
        \label{fig:add}
    \end{subfigure}
    \caption{
        $2i\in \cI_S^{(2n_c)}$ means that we apply the swapping  transform to $(X_{\beta + (2i-1)n/(2n_c)}^{(2n_c)}, X_{\beta + 2in/(2n_c)}^{(2n_c)})$ for $1\le \beta \le n/(2n_c)$.
        Lines~13-15 and Lines~19-23 in Algorithm~\ref{algo:decode_swp_channel} follow from Fig.~\ref{fig:swp} and \eqref{eq:tb_tf}.
        $2i\in \cI_A^{(2n_c)}$ means that we apply the Ar{\i}kan transform to $(X_{\beta + (2i-1)n/(2n_c)}^{(2n_c)}, X_{\beta + 2in/(2n_c)}^{(2n_c)})$ for $1\le \beta \le n/(2n_c)$.
        Lines~13-15 and Lines~19-23 in Algorithm~\ref{algo:decode_add_channel} follow from Fig.~\ref{fig:add} and \eqref{eq:tb_tf}.
    }
    \label{fig:swp&add}
\end{figure}

\begin{algorithm}
    \DontPrintSemicolon
    \caption{\texttt{decode\_swapped\_channel$(n_c, i)$}}
    \label{algo:decode_swp_channel}
    \KwIn{$n_c\in \{2,4,8,\dots, n\}$ and index $i$ satisfying $2i\in \cI_S^{(2n_c)}$.}


    \For{\em $\beta \in \{1,2,\dots,n/(2n_c)\}, a\in\{0,1\}$ and $b\in\{0,1\}$}{
        $\tP_{2n_c}[2i-1,\beta][a,b] \gets $

        \texttt{calculate\_probability}$(n_c, i, \beta, \swpa, a,b)$
    }

    \texttt{decode\_channel}$(2n_c, 2i-1)$


    \For{\em $\beta \in \{1,2,\dots,n/(2n_c)\}, a\in\{0,1\}$ and $b\in\{0,1\}$}{
        $\tP_{2n_c}[2i,\beta][a,b] \gets $

        \texttt{calculate\_probability}$(n_c, i, \beta, \swpb, a,b)$
    }

    \texttt{decode\_channel}$(2n_c, 2i)$


    \For{\em $\beta \in \{1,2,\dots,n/(2n_c)\}, a\in\{0,1\}$ and $b\in\{0,1\}$}{
        $\tP_{2n_c}[2i+1,\beta][a,b] \gets $

        \texttt{calculate\_probability}$(n_c, i, \beta, \swpc, a,b)$
    }

    \texttt{decode\_channel}$(2n_c, 2i+1)$

    \eIf{$i\le n_c-2$}{
        \Comment{Only decode one bit ${X}_{\beta + (i-1)n/n_c}^{(n_c)}$ for each $\beta$}

        \For{$\beta \in \{1,2,\dots, n/(2n_c)\}$}{
            $\beta' \gets \beta + n/(2n_c)$ \\[0.4em]

            $\tB_{n_c}[i, \beta ] \gets \tB_{2n_c}[2i-1, \beta] + \tB_{2n_c}[2i+1, \beta]$\\[0.4em]

            $\tB_{n_c}[i, \beta'] \gets \tB_{2n_c}[2i+1, \beta]$

            \Comment{See Fig.~\ref{fig:swp} for an explanation}
        }
    }{
        \Comment{Decode two bits ${X}_{\beta + (n_c-2)n/n_c}^{(n_c)}, {X}_{\beta + (n_c-1)n/n_c}^{(n_c)}$ for each $\beta$}

        \For{$\beta \in \{1,2,\dots, n/(2n_c)\}$}{
            $\beta' \gets \beta + n/(2n_c)$\\[0.4em]

            $\tB_{n_c}[n_c-1, \beta ] \gets \tB_{2n_c}[2n_c-3,\beta] + \tB_{2n_c}[2n_c-3, \beta]$\\[0.4em]

            $\tB_{n_c}[n_c-1, \beta'] \gets \tB_{2n_c}[2n_c-1,\beta]$ \\[0.4em]

            $\tB_{n_c}[n_c,   \beta ] \gets \tB_{2n_c}[2n_c-2,\beta] + \tB_{2n_c}[2n_c, \beta]$\\[0.4em]

            $\tB_{n_c}[n_c,   \beta'] \gets \tB_{2n_c}[2n_c,  \beta]$

            \Comment{See Fig.~\ref{fig:swp} for an explanation}
        }
    }

    \Return
\end{algorithm}

\begin{algorithm}
    \DontPrintSemicolon
    \caption{\texttt{decode\_added\_channel$(n_c, i)$}}
    \label{algo:decode_add_channel}
    \KwIn{$n_c\in \{2,4,8,\dots, n\}$ and index $i$ satisfying $2i\in \cI_A^{(2n_c)}$}


    \For{\em $\beta \in \{1,2,\dots,n/(2n_c)\}, a\in\{0,1\}$ and $b\in\{0,1\}$}{
        $\tP_{2n_c}[2i-1,\beta][a,b] \gets $

        \texttt{calculate\_probability}$(n_c, i, \beta, \adda, a,b)$
    }

    \texttt{decode\_channel}$(2n_c, 2i-1)$


    \For{\em $\beta \in \{1,2,\dots,n/(2n_c)\}, a\in\{0,1\}$ and $b\in\{0,1\}$}{
        $\tP_{2n_c}[2i,\beta][a,b] \gets $

        \texttt{calculate\_probability}$(n_c, i, \beta, \addb, a,b)$
    }

    \texttt{decode\_channel}$(2n_c, 2i)$


    \For{\em $\beta \in \{1,2,\dots,n/(2n_c)\}, a\in\{0,1\}$ and $b\in\{0,1\}$}{
        $\tP_{2n_c}[2i+1,\beta][a,b] \gets $

        \texttt{calculate\_probability}$(n_c, i, \beta, \addc, a,b)$
    }

    \texttt{decode\_channel}$(2n_c, 2i+1)$

    \eIf{$i\le n_c-2$}{
        \Comment{Only decode one bit $X_{\beta + (i-1)n/n_c}^{(n_c)}$ for each $\beta$}

        \For{$\beta\in\{ 1, 2, \dots, n/(2n_c)\}$}{
            $\beta' \gets \beta + n/(2n_c)$ \\[0.4em]

            $\tB_{n_c}[i, \beta ] \gets \tB_{2n_c}[2i-1, \beta] + \tB_{2n_c}[2i, \beta] + \tB_{2n_c}[2i+1, \beta]$ \\[0.4em]

            $\tB_{n_c}[i, \beta'] \gets                           \tB_{2n_c}[2i, \beta] + \tB_{2n_c}[2i+1, \beta]$

            \Comment{See Fig.~\ref{fig:add} for an explanation}
        }

    }{
        \Comment{Decode two bits ${X}_{\beta + (n_c-2)n/n_c}^{(n_c)}, {X}_{\beta + (n_c-1)n/n_c}^{(n_c)}$ for each $\beta$}

        \For{$\beta\in\{ 1, 2, \dots, n/(2n_c)\}$}{
            $\beta'\gets \beta + n/(2n_c)$ \\[0.4em]

            $\tB_{n_c}[n_c-1, \beta ] \gets \tB_{2n_c}[2n_c-3, \beta] + \tB_{2n_c}[2n_c-2, \beta] + \tB_{2n_c}[2n_c-1, \beta]$ \\[0.4em]

            $\tB_{n_c}[n_c-1, \beta'] \gets                             \tB_{2n_c}[2n_c-2, \beta] + \tB_{2n_c}[2n_c-1, \beta]$\\[0.4em]

            $\tB_{n_c}[n_c,   \beta ] \gets \tB_{2n_c}[2n_c-1, \beta] + \tB_{2n_c}[2n_c,   \beta]$ \\[0.4em]

            $\tB_{n_c}[n_c,   \beta'] \gets \tB_{2n_c}[2n_c,   \beta]$

            \Comment{See Fig.~\ref{fig:add} for an explanation}
        }
    }
    \Return

\end{algorithm}

The structure of Algorithm~\ref{algo:decode_add_channel} is exactly the same as that of Algorithm~\ref{algo:decode_swp_channel}.
The only difference is that we call \texttt{decode\_added\_channel} in Algorithm~\ref{algo:decode_channel} when $2i\in \cI_{A}^{(2n_c)}$.
In this case, we have $V_{2i-1}^{(2n_c)} = (V_i^{(n_c)})^{\adda}, V_{2i}^{(2n_c)} = (V_i^{(n_c)}){^\addb}, V_{2i+1}^{(2n_c)} = (V_i^{(n_c)}){^\addc}$,
and all the calculations in Algorithm~\ref{algo:decode_add_channel} follow from this recursive relation.

The structure of Algorithm~\ref{algo:decode_ori_channel} differs from that of Algorithm~\ref{algo:decode_swp_channel}  and Algorithm~\ref{algo:decode_add_channel} in two places.
First, Algorithm~\ref{algo:decode_ori_channel} only calculates $\tP_{2n_c}[2i-1]$ and calls \texttt{decode\_channel}$(2n_c, 2i-1)$ when $2(i-1)\notin \cI^{(2n_c)}$; see Lines~1-4.
In contrast, Algorithm~\ref{algo:decode_swp_channel}  and Algorithm~\ref{algo:decode_add_channel} always calculate $\tP_{2n_c}[2i-1]$ and call \texttt{decode\_channel}$(2n_c, 2i-1)$ to decode $X_{\beta + (2i-2)n/(2n_c)}^{(2n_c)}, 	1 \le \beta \le n/(2n_c) $ for all values of $i$.
This is because we have already decoded $X_{\beta + (2i-2)n/(2n_c)}^{(2n_c)}, 	1 \le \beta \le n/(2n_c)$ when $2(i-1) \in \cI^{(2n_c)}$, and this condition can only hold for the input parameters $n_c$ and $i$ in Algorithm~\ref{algo:decode_ori_channel}.
In both Algorithm~\ref{algo:decode_swp_channel} and Algorithm~\ref{algo:decode_add_channel}, we have $2i\in\cI^{(2n_c)}$, and the fully separated requirement \eqref{eq:fully_separeted} implies that $2(i-1)\notin \cI^{(2n_c)}$.
Second,  Algorithm~\ref{algo:decode_ori_channel} only calculates $\tP_{2n_c}[2i+1]$ and calls \texttt{decode\_channel}$(2n_c, 2i+1)$ when $i = n_c - 1$; see Lines~19-21.
In contrast, Algorithm~\ref{algo:decode_swp_channel}  and Algorithm~\ref{algo:decode_add_channel}  calculate $\tP_{2n_c}[2i+1]$ and call \texttt{decode\_channel}$(2n_c, 2i+1)$ to decode $X_{\beta + 2in/(2n_c)}^{(2n_c)}, 	1 \le \beta \le n/(2n_c) $ for all $1\le i \le n_c-1$.
This is because $X_{\beta  + 2in/(2n_c) }^{(2n_c)}, 1\le \beta \le n/(2n_c)$ is needed in the calculation of $X_{\beta + (i-1)n/n_c}^{(n_c)}, 1\le \beta \le n/n_c$ if and only if $2i\in \cI^{(2n_c)}$; see Fig.~\ref{fig:swp&add} and Fig.~\ref{fig:ori}.

\begin{figure}
    \centering
    {\large $\beta' = \beta + n/(2n_c)$\\[1em]}

    \begin{subfigure}{.85\linewidth}
        \centering
        \includegraphics[scale=0.90]{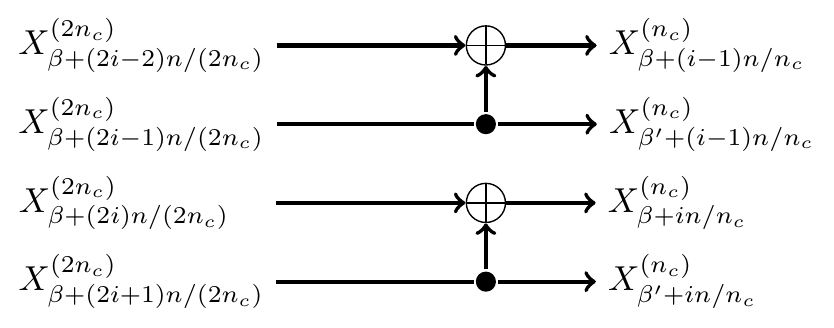}
        \caption{$2i \notin \cI^{(2n_c)}$, $2(i-1) \notin \cI^{(2n_c)}$}
        \label{ori_ori}
    \end{subfigure}

    \begin{subfigure}{.85\linewidth}
        \centering
        \includegraphics[scale = 0.90]{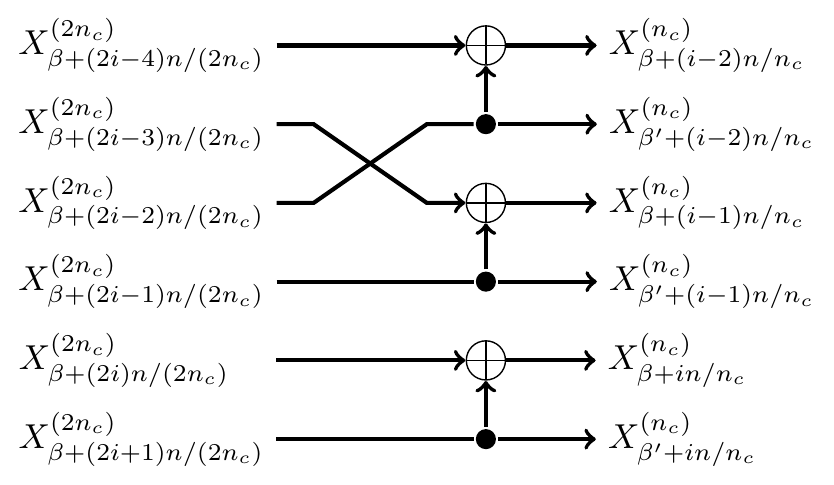}
        \caption{$2i \notin \cI^{(2n_c)}$, $2(i-1)\in \cI_{S}^{(2n_c)}$}
        \label{ori_swp}
    \end{subfigure}

    \begin{subfigure}{.85\linewidth}
        \centering
        \includegraphics[scale = 0.90]{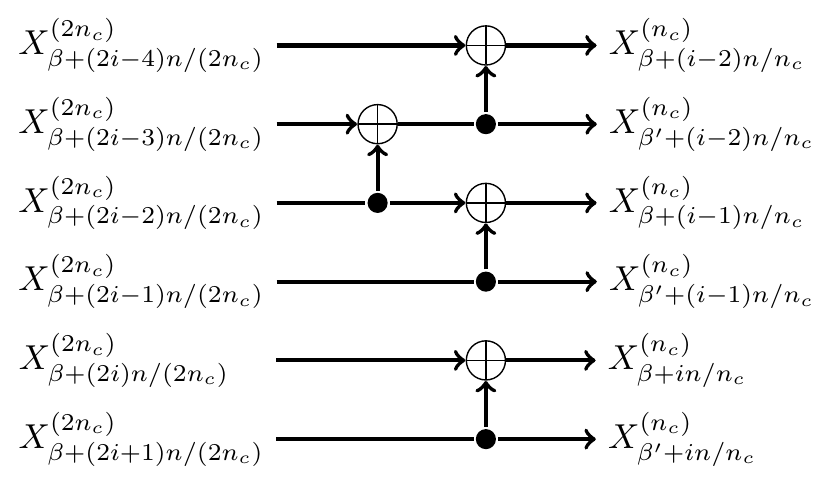}
        \caption{$2i \notin \cI^{(2n_c)}$, $2(i-1)\in \cI_{A}^{(2n_c)}$}
        \label{ori_add}
    \end{subfigure}
    \caption{
        Explanation of Algorithm~\ref{algo:decode_ori_channel}.
        The aim of the function $\texttt{decode\_channel}(n_c, i)$ is to compute the values of $X^{(n_c)}_{\beta+(i-1)n/n_c}$ and $X^{(n_c)}_{\beta'+(i-1)n/n_c}$.
        In case (a), we calculate the value of $X_{\beta + (2i-2)n/(2n_c)}^{(2n_c)}$ first.
        In cases (b) and (c), the values of $X_{\beta + (2i-3)n/(2n_c)}^{(2n_c)}$ and $X_{\beta + (2i-2)n/(2n_c)}^{(2n_c)}$ respectively, has been calculated in the function $\texttt{decode\_channel}(n_c, i-1)$.
    }
    \label{fig:ori}
\end{figure}

\begin{algorithm}
    \DontPrintSemicolon
    \caption{\texttt{decode\_original\_channel$(n_c, i)$}}
    \label{algo:decode_ori_channel}
    \KwIn{$n_c\in \{2,4,8,\dots, n\}$ and index $i$  satisfying $1\le i \le n_c-1$ and $2i\notin \cI^{(2n_c)}$}

    \uIf{\em $2(i-1)\notin \cI^{(2n_c)}$}{

        \For{\em $\beta \in \{1,2,\dots,n/(2n_c)\}, a\in\{0,1\}$ and $b\in\{0,1\}$}{
            $\tP_{2n_c}[2i-1,\beta][a,b] \gets$

            \texttt{calculate\_probability}$(n_c, i, \beta, \oria, a,b)$
        }

        \texttt{decode\_channel}$(2n_c, 2i-1)$
    }\ElseIf{$2(i-1)\in \cI_{S}^{(2n_c)}$}{
        \For{$\beta \in \{1,2,\dots, n/(2n_c)\}$}
        {
            $\tB_{2n_c}[2i-1, \beta] \gets \tB_{2n_c}[2i-2,\beta]$
        }
    }


    \For{\em $\beta \in \{1,2,\dots,n/(2n_c)\}, a\in\{0,1\}$ and $b\in\{0,1\}$}{
        $\tP_{2n_c}[2i,\beta][a,b] \gets $

        \texttt{calculate\_probability}$(n_c, i, \beta, \orib, a,b)$
    }

    \texttt{decode\_channel}$(2n_c, 2i)$

    \eIf{$i \le n_c-2$}{
        \Comment{Only decode one bit ${X}_{\beta + (i-1)n/n_c}^{(n_c)}$ for each $\beta$}

        \For{$ \beta\in\{1, 2, \dots, n/(2n_c)\}$}{
            $\beta' \gets \beta + n/(2n_c)$ \\[0.4em]

            $\tB_{n_c}[i, \beta ] \gets \tB_{2n_c}[2i-1, \beta] + \tB_{2n_c}[2i, \beta]$ \\[0.4em]

            $\tB_{n_c}[i, \beta'] \gets                           \tB_{2n_c}[2i, \beta]$

        }
    }{
        \Comment{Decode two bits ${X}_{\beta + (n_c-2)n/n_c}^{(n_c)}, {X}_{\beta + (n_c-1)n/n_c}^{(n_c)}$ for each $\beta$}

        \For{\em $\beta \in \{1,2,\dots,n/(2n_c)\}, a\in\{0,1\}$ and $b\in\{0,1\}$}{
            $\tP_{2n_c}[2i+1,\beta][a,b] \gets $

            \texttt{calculate\_probability}$(n_c, i, \beta, \oric, a,b)$
        }

        \texttt{decode\_channel}$(2n_c, 2i+1)$

        \For{$ \beta\in\{1, 2, \dots, n/(2n_c)\}$}{
            $\beta' \gets \beta + n/(2n_c)$ \\[0.4em]

            $\tB_{n_c}[n_c-1, \beta ]\gets \tB_{2n_c}[2n_c-3, \beta] + \tB_{2n_c}[2n_c-2, \beta]$ \\[0.4em]

            $\tB_{n_c}[n_c-1, \beta']\gets                             \tB_{2n_c}[2n_c-2, \beta]$ \\[0.4em]

            $\tB_{n_c}[n_c, \beta ]  \gets \tB_{2n_c}[2n_c-1, \beta] + \tB_{2n_c}[2n_c, \beta]$  \\[0.4em]

            $\tB_{n_c}[n_c, \beta']  \gets                             \tB_{2n_c}[2n_c, \beta]$
        }
    }

    \Return

\end{algorithm}

The calculations in Lines~14-16 and Lines~23-27 of Algorithm~\ref{algo:decode_ori_channel} are explained in Fig.~\ref{fig:ori}.
More specifically, Fig.~\ref{ori_ori} and Fig.~\ref{ori_add} tell us that when $2(i-1)\notin\cI_S^{(2n_c)}$, we have the following relation
\begin{equation}\label{eq:cal_ori}
    \begin{aligned}
         & X_{\beta + (i-1)n/n_c}^{(n_c)} = X_{\beta + (2i-2)n/(2n_c)}^{(2n_c)} + X_{\beta + (2i-1)n/(2n_c)}^{(2n_c)}, \\& X_{\beta' + (i-1)n/n_c}^{(n_c)} = X_{\beta + (2i-1)n/(2n_c)}^{(2n_c)}, \\
         & X_{\beta +     in/n_c}^{(n_c)} = X_{\beta +     2in/(2n_c)}^{(2n_c)} + X_{\beta + (2i+1)n/(2n_c)}^{(2n_c)}, \\& X_{\beta' +     in/n_c}^{(n_c)} = X_{\beta + (2i+1)n/(2n_c)}^{(2n_c)},
    \end{aligned}
\end{equation}
where $\beta' = \beta + n/(2n_c)$.
The relation for the case $2(i-1)\in \cI_{S}^{(2n_c)}$ is obtain from replacing $X_{\beta + (2i-2)n/(2n_c)}^{(2n_c)}$ with $X_{\beta + (2i-3)n/(2n_c)}^{(2n_c)}$ in the first equation above.
In this case, we move the decoding result $\hat{x}_{\beta + (2i-3)n/(2n_c)}^{(2n_c)}$ stored in $\tB_{2n_c}[2i-2, \beta]$ to $\tB_{2n_c}[2i-1, \beta]$; see Lines~5-7 in Algorithm~\ref{algo:decode_ori_channel}.
Then the calculations in Lines~14-16 and Lines~23-27 follow from \eqref{eq:cal_ori} and \eqref{eq:tb_tf}.

In Fig.~\ref{fig:example_dec}, we use the $(16,8)$ ABS+ polar code defined in Fig.~\ref{fig:example_enc} as a concrete example to illustrate the recursive structure of the function $\texttt{decode\_channel}$ in Algorithm~\ref{algo:decode_channel}.

\begin{figure*}
    \centering
    \includegraphics[scale = 0.65]{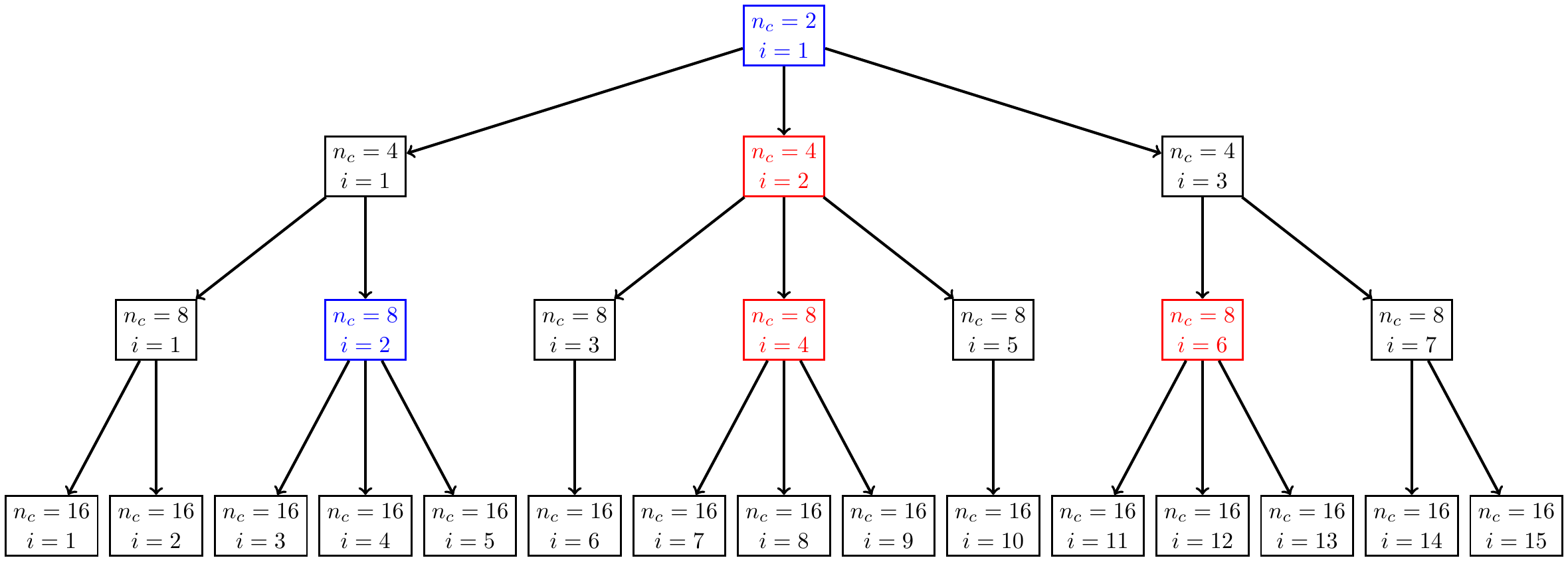}
    \caption{Recursive decoding of the $(16, 8)$ ABS+ polar code defined in Fig. \ref{fig:example_enc}.
        We put $(n_c, i)$ in a blue block (e.g., $n_c = 2, i = 1$) if $2i\in\cI_S^{(2n_c)}$. In this case, $\texttt{decode\_channel}(n_c, i)$ in Algorithm~\ref{algo:decode_channel} calls $\texttt{decode\_swapped\_channel}(n_c, i)$.
        We put $(n_c, i)$ in a  red block (e.g., $n_c = 4, i = 2$) if $2i\in\cI_A^{(2n_c)}$. In this case, $\texttt{decode\_channel}(n_c, i)$ calls $\texttt{decode\_added\_channel}(n_c, i)$.
        An arrow from the block $(n_c, i)$ to the block $(n_c', i')$ means that $\texttt{decode\_channel}(n_c', i')$ is called in the execution of  $\texttt{decode\_channel}(n_c, i)$.
        For example, we call $\texttt{decode\_channel}$ with input parameters $(4, 1)$, $(4, 2)$ and $(4, 3)$ in the execution of  $\texttt{decode\_channel}(2,1)$.
    }
    \label{fig:example_dec}
\end{figure*}

\begin{table*}
    \centering
    \begin{tabular}{c|cccccc}
        \hline
        $(n,k)$                 & $(256,77)$   & $(256,128)$  & $(256,179)$  & $(512,154)$  & $(512,256)$   & $(512,358)$   \\
        \hline
        {\bf (1) ST,~~~ $L=32$} & $0.963$ms    & 1.41ms       & 1.73ms       & 1.94ms       & 2.80ms        & 3.54ms        \\

        {\bf (2) ABS,~ $L=20$}  & 0.816ms      & 1.24ms       & 1.47ms       & 1.86ms       & 2.66ms        & 3.10ms        \\

        {\bf (3) ABS,~ $L=32$}  & 1.29ms       & 1.99ms       & 2.37ms       & 2.93ms       & 4.36ms        & 5.13ms        \\

        {\bf (4) ABS+, $L=20$}  & 0.807ms      & 1.25ms       & 1.48ms       & 1.75ms       & 2.56ms        & 3.15ms        \\

        {\bf (5) ABS+, $L=32$}  & 1.30ms       & 1.98ms       & 2.50ms       & 2.85ms       & 4.29ms        & 5.28ms        \\

        \hline
        $(n,k)$                 & $(1024,307)$ & $(1024,512)$ & $(1024,717)$ & $(2048,614)$ & $(2048,1024)$ & $(2048,1434)$ \\
        \hline
        {\bf (1) ST,~~~ $L=32$} & 4.21ms       & 5.75ms       & 7.15ms       & 9.05ms       & 11.7ms        & 14.6ms        \\

        {\bf (2) ABS,~ $L=20$}  & 4.32ms       & 5.90ms       & 6.67ms       & 10.6ms       & 12.6ms        & 14.0ms        \\

        {\bf (3) ABS,~ $L=32$}  & 6.63ms       & 9.41ms       & 10.8ms       & 16.7ms       & 20.1ms        & 23.2ms        \\

        {\bf (4) ABS+, $L=20$}  & 4.35ms       & 5.58ms       & 6.76ms       & 10.2ms       & 13.1ms        & 14.2ms        \\

        {\bf (5) ABS+, $L=32$}  & 6.86ms       & 8.89ms       & 10.9ms       & 16.0ms       & 20.3ms        & 22.9ms        \\
        \hline
    \end{tabular}
    \caption{Comparison of the decoding time over the binary-input AWGN channel with $E_b/N_0=2\dB$.
        The row starting with $(n,k)$ lists the code length and code dimension we have tested.
        ``ST" refers to standard polar codes.
        The parameter $L$ is the list size.
        The time unit ``ms" is $10^{-3}$s.}
    \label{tb:time}
\end{table*}

In the whole decoding procedure, we call the function \texttt{decode\_channel}$(n_c, i)$ exactly once for each $n_c\in \{2, 4, 8,\dots, n\}$ and each $1\le i\le n_c-1$; see Fig.~\ref{fig:example_dec} for an illustration.
It is easy to see that the time complexity\footnote{We do not include the running time of the recursive calls  \texttt{decode\_channel}$(2n_c, 2i-1)$, \texttt{decode\_channel}$(2n_c, 2i)$ and \texttt{decode\_channel}$(2n_c, 2i+1)$ in the time complexity of \texttt{decode\_channel}$(n_c, i)$.} of \texttt{decode\_channel}$(n_c, i)$ is $O(n/n_c)$.
Since $i$ takes $n_c-1$ values and $n_c$ takes $\log(n)$ values, the time complexity of the SC decoder is $O(n/n_c) \cdot n_c \cdot \log(n) = O(n\log(n))$.

\begin{proposition}
    The time complexity of the SC decoder for ABS+ polar codes is $O(n\log (n))$.
\end{proposition}

\subsection{Space-efficient version of the SC decoder}\label{sect:efficient_dec}
As mentioned earlier, we can reduce the space complexity of the SC decoder from $O(n\log(n))$ to $O(n)$.
As we can see from Algorithms~\ref{algo:decode_channel}--\ref{algo:decode_add_channel}, we only use the entries stored in $\tP_{n_c}[i]$ when we call the function $\texttt{decode\_channel}$ with input parameters $(n_c, i)$.
The entries in $\tP_{n_c}[i]$ are never used again in the whole decoding  algorithm after the function $\texttt{decode\_channel}(n_c, i)$ returns.
Moreover, for $1\le i \le n_c-2$, the function  $\texttt{decode\_channel}(n_c, i+1)$ is called after  $\texttt{decode\_channel}(n_c, i)$ returns.
Therefore, for each $n_c\in\{2, 4, 8, \dots, n\}$,
we can reduce the 4-dimensional array
\begin{align*}
    (\tP_{n_c}[i,\beta][a,b],\quad 1\le i\le n_c-1,\quad 1\le \beta \le n/n_c, \\ a\in\{0,1\},\quad b\in\{0,1\})
\end{align*}
to a 3-dimensional array
$$
    (\tP_{n_c}[\beta][a,b],  1\le \beta \le n/n_c, a\in\{0,1\}, b\in\{0,1\})
$$
by dropping the index $i$.
Each entry $\tP_{n_c}[\beta][a,b]$ stores  $V_i^{(n_c)}(\hat{ \mathbi{x}}_{i,\beta}^{(n_c)},{\mathbi{y}}_{\beta}^{(n_c)}| a,b)$ when we call the function $\texttt{decode\_channel}$ with input parameters $(n_c, i)$.

Upon completion of the function $\texttt{decode\_channel}(n_c,\linebreak i)$, we note that
(1) If $i\in \cI_{S}^{(n_c)}$, the decoding results of $X_{\beta + (i-1)n/(n_c)}^{(n_c)}$ will be used to compute the transition probabilities for adjacent-bits-channel $V^{(n_c), \ABSP}_{i+2}$ (see Fig.~\ref{fig:ori} (b));
(2) If $i-1\in \cI_{A}^{(n_c)}$, the decoding results of $X_{\beta + (i-1)n/(n_c)}^{(n_c)}$ will be used to compute the transition probabilities for adjacent-bits-channel $V^{(n_c), \ABSP}_{i+1}$ (see Fig.~\ref{fig:ori} (c)).
For these two exceptional cases, we need an additional helper space $\tH$ to store these intermediate results.
More precisely, we reduce the 2-dimensional array $(\tB_{n_c}[i,\beta], 1\le i\le n_c, 1\le \beta \le n/n_c)$ to a 1-dimensional array $(\tB_{n_c}[\beta], 1\le \beta \le n/n_c)$.
In addition, we introduce a new 1-dimensional array $(\tH_{n_c}[\beta], 1\le \beta\le n/n_c)$.
Each entry $\tB_{n_c}[\beta]$ stores $\hat{x}_{\beta + (i-1)n/n_c}^{(n_c)}$ when the function $\texttt{decode\_channel}(n_c, i)$ returns, and each entry $\tH_{n_c}[\beta]$ stores $\hat{x}_{\beta + (n_c-1)n/n_c}^{(n_c)}$ when the function $\texttt{decode\_channel}(n_c, n_c-1)$ returns.

To summarize, we have three data structures
\begin{equation*}
    \begin{aligned}
         & (\tP_{n_c}[\beta][a,b],\quad  1\le \beta \le n/n_c, a\in\{0,1\}, b\in\{0,1\}),          \\
         & (\tB_{n_c}[\beta], 1\le \beta \le n/n_c),\quad (\tH_{n_c}[\beta], 1\le \beta\le n/n_c), \\
         & n_c\in \{2, 4, 8, \dots, n\}
    \end{aligned}
\end{equation*}
in the space-efficient version of the SC decoder. For each $n_c\in\{2, 4, 8, \dots, n\}$, the arrays $\tP_{n_c}$, $\tB_{n_c}$ and $\tH_{n_c}$ have $6n/n_c$ entries in total.
Therefore the space complexity of the space-efficient SC decoder is $6(n/2 + n/4 + n/8 + \cdots + n/n) = O(n)$.

\begin{proposition}
    The space complexity of the space-efficient SC decoder for ABS+ polar codes is $O(n)$.
\end{proposition}

Next we show how to modify Algorithms~\ref{algo:ABS+_Decoding}-\ref{algo:decode_add_channel} to the space-efficient version.
Algorithms~\ref{algo:ABS+_Decoding}-\ref{algo:calcu_proba} only require a few modifications, which are listed below.
\begin{enumerate}[(i)]
    \item In Algorithm~\ref{algo:ABS+_Decoding}, Algorithm~\ref{algo:decode_boundary_channel}  and Algorithm~\ref{algo:calcu_proba},  we replace $\tP_{n_c}[i,\beta][a,b]$ with $\tP_{n_c}[\beta][a,b]$ and replace $\tB_{n_c}[i,\beta]$ with $\tB_{n_c}[\beta]$ for all $n_c \in\{2, 4, 8, \dots, n\}$, $1\le i\le n_c-1$, $1\le \beta \le n/n_c$ and $a, b \in \{0,1\}$.
    \item In Algorithm~\ref{algo:ABS+_Decoding} and Algorithm~\ref{algo:decode_boundary_channel}, we replace $\tB_{n_c}[n_c, \beta]$ with $\tH_{n_c}[\beta]$ for all $n_c \in\{2, 4, 8, \dots, n\}$ and $1\le \beta \le n/n_c$.
    \item After Line~8 in Algorithm~\ref{algo:decode_channel}, we add the following three lines

          \textbf{if} $i\in \cI_S^{(n_c)}$ or $i-1\in \cI_A^{(n_c)}$ \textbf{then}

          \qquad \textbf{for} $\beta \in \{1, 2, \dots, n/n_c\}$ \textbf{do}

          \qquad \qquad $\tH_{n_c}[\beta] \gets \tB_{n_c}[\beta]$
    \item In Algorithm~\ref{algo:calcu_proba}, we replace Line~12 with

          $\quad r_1 \gets \tB_{n_c}[\beta], \quad r_2 \gets a,\quad r_3 \gets b $

          and replace Line~21 with

          $\quad r_1 \gets \tB_{n_c}[\beta],\quad r_2 \gets \tB_{n_c}[\beta'],\quad r_3 \gets a,\quad r_4 \gets b$
\end{enumerate}
Algorithms~\ref{algo:decode_swp_channel}-\ref{algo:decode_ori_channel} require more changes than Algorithms~\ref{algo:ABS+_Decoding}-\ref{algo:calcu_proba}.
In Algorithm~\ref{algo:decode_swp_channel_space_efficient}, we present the space-efficient version of Algorithm~\ref{algo:decode_swp_channel}.
In Algorithm~\ref{algo:decode_add_channel_space_efficient}, we present the space-efficient version of Algorithm~\ref{algo:decode_add_channel}.
Finally, we present the space-efficient version of Algorithm~\ref{algo:decode_ori_channel} in Algorithm~\ref{algo:decode_ori_channel_space_efficient}.
Algorithms~\ref{algo:decode_ori_channel_space_efficient}-\ref{algo:decode_add_channel_space_efficient} are given in Appendix \ref{sect:space}.

\begin{figure*}
    \centering
    \begin{subfigure}{0.41\linewidth}
        \centering
        \includegraphics[width=\linewidth]{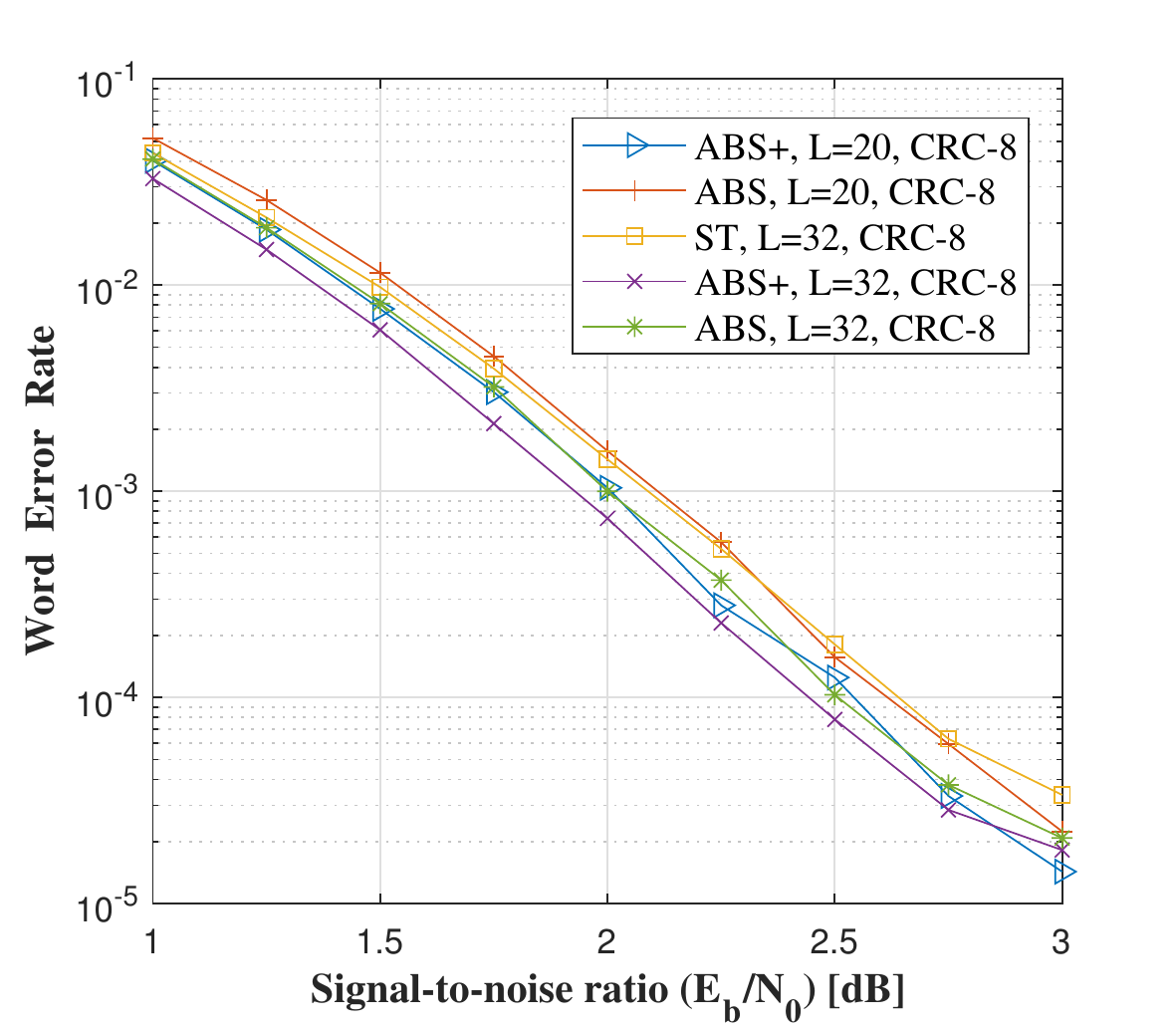}
        \caption{length 256, dimension 77}
    \end{subfigure}
    ~\hspace*{0.2in}
    \begin{subfigure}{0.41\linewidth}
        \centering
        \includegraphics[width=\linewidth]{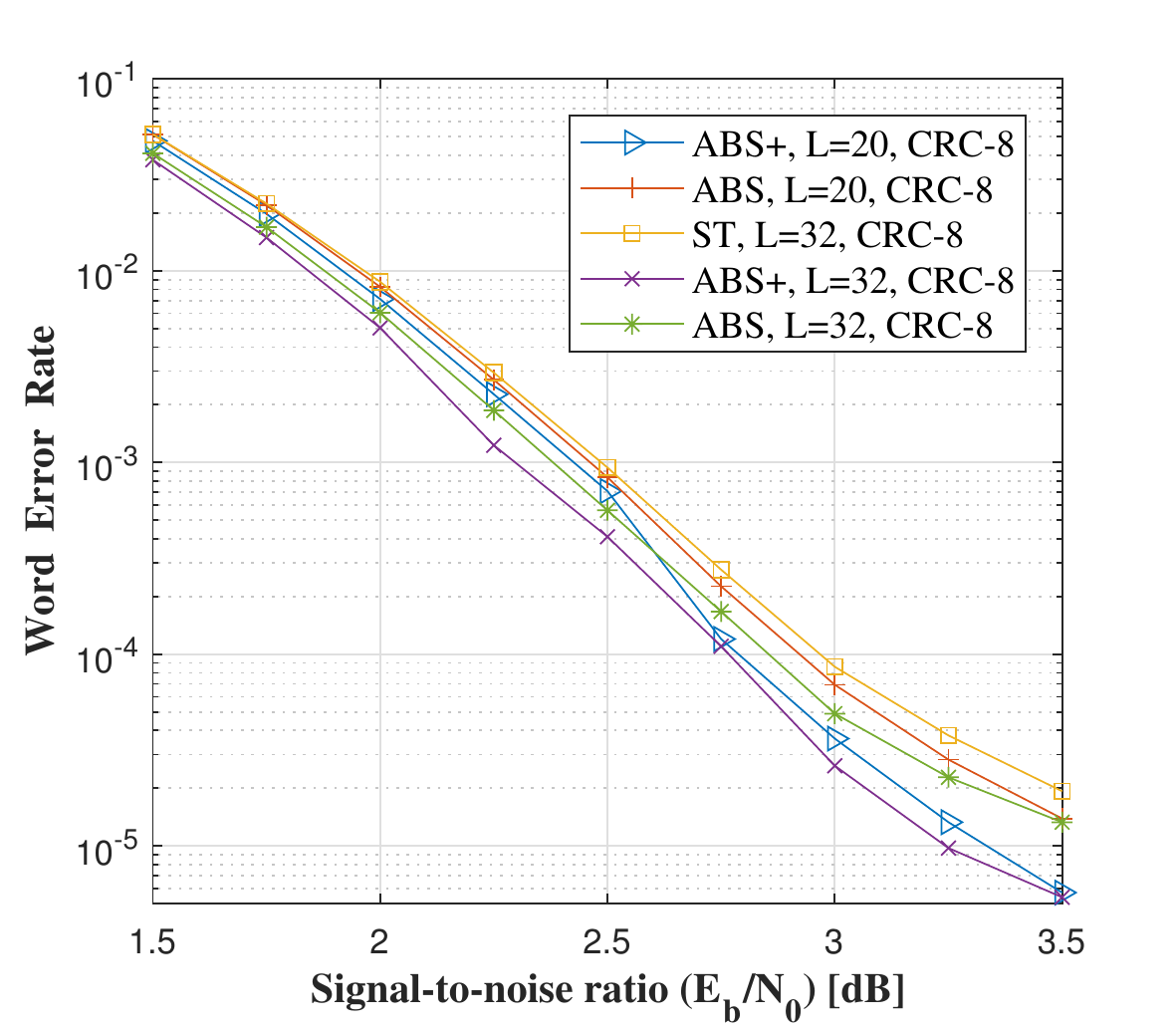}
        \caption{length 256, dimension 128}
    \end{subfigure}

    \vspace*{0.1in}

    \begin{subfigure}{0.41\linewidth}
        \centering
        \includegraphics[width=\linewidth]{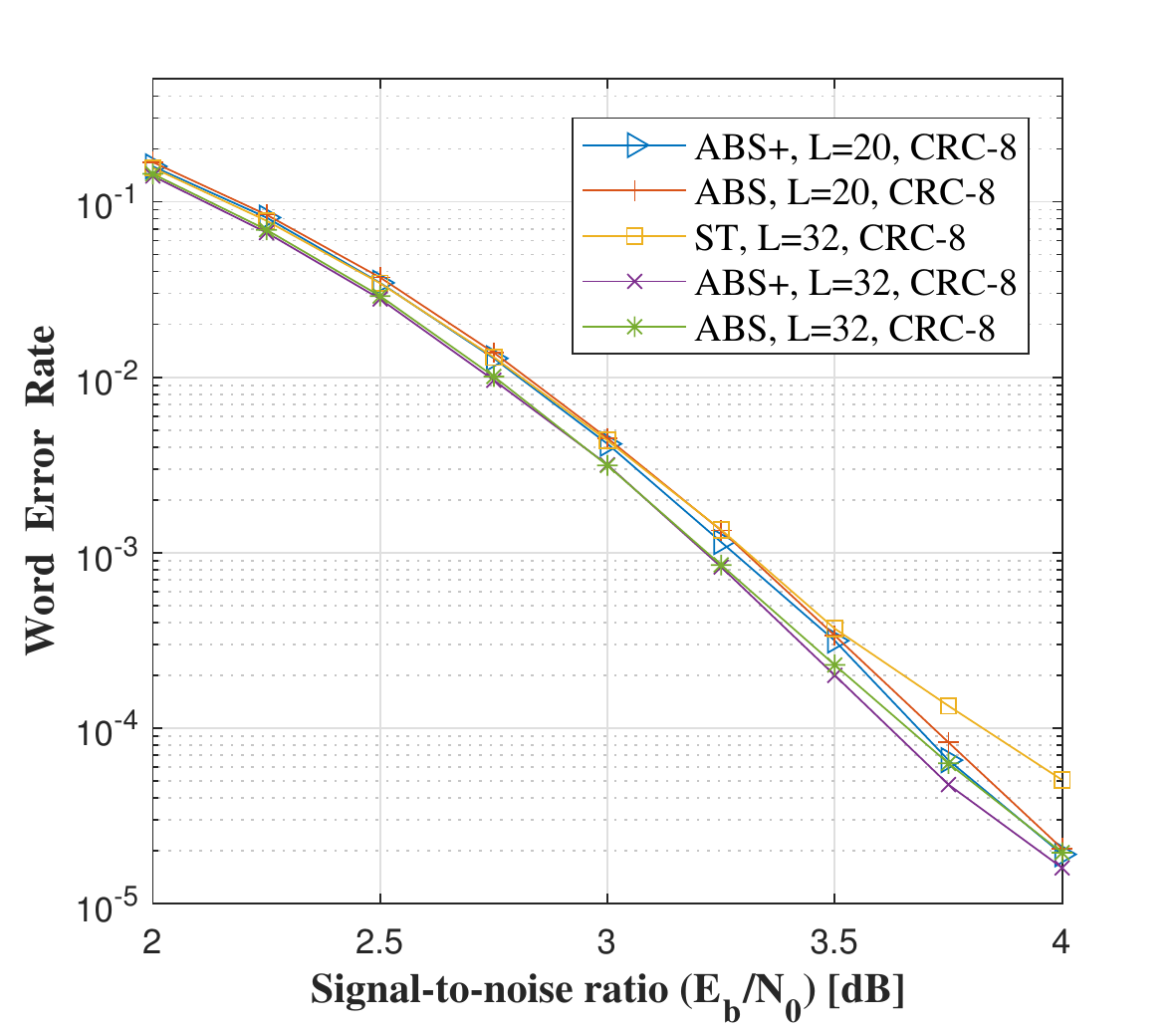}
        \caption{length 256, dimension 179}
    \end{subfigure}
    ~\hspace*{0.2in}
    \begin{subfigure}{0.41\linewidth}
        \centering
        \includegraphics[width=\linewidth]{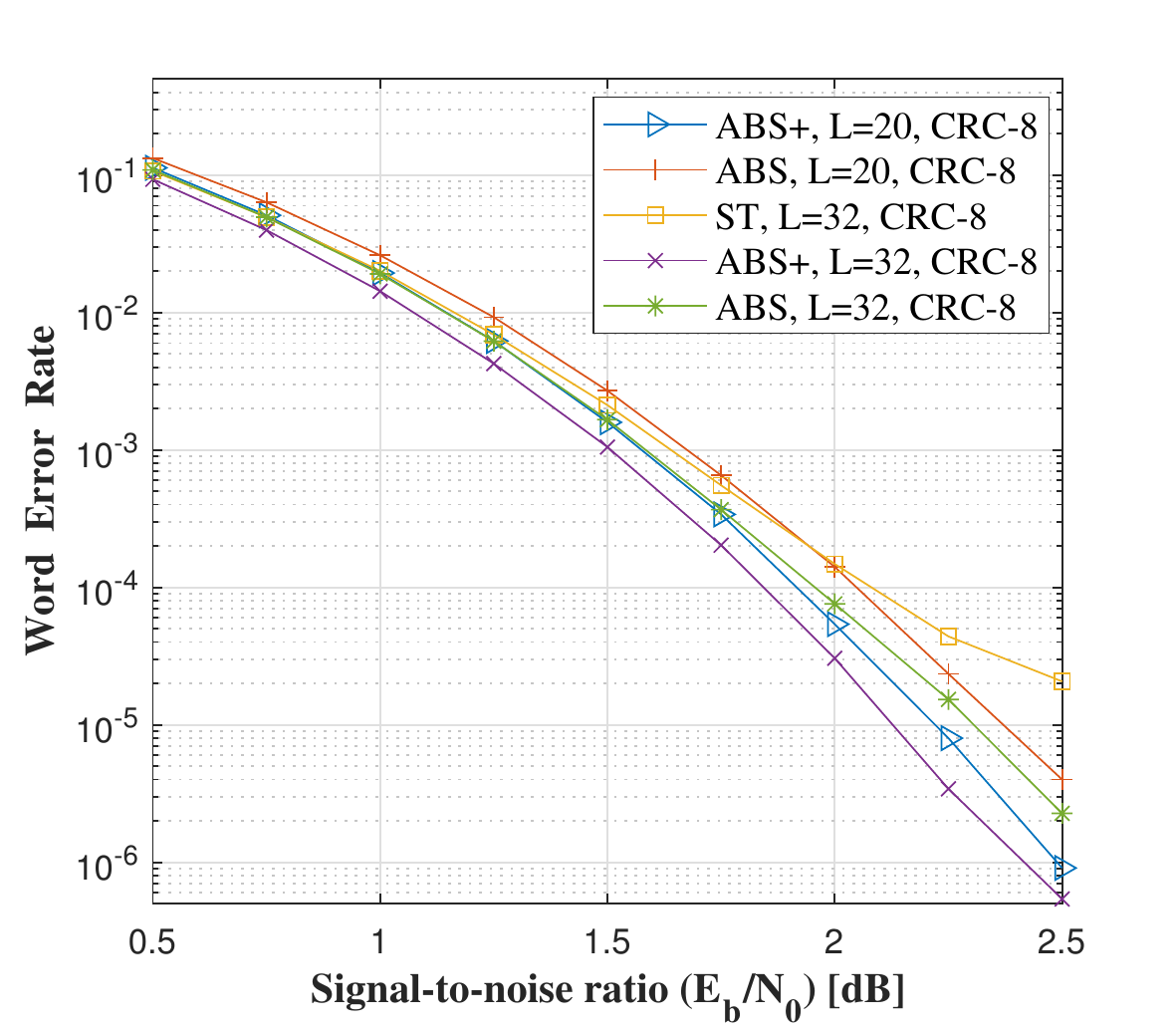}
        \caption{length 512, dimension 154}
    \end{subfigure}

    \vspace*{0.1in}

    \begin{subfigure}{0.41\linewidth}
        \centering
        \includegraphics[width=\linewidth]{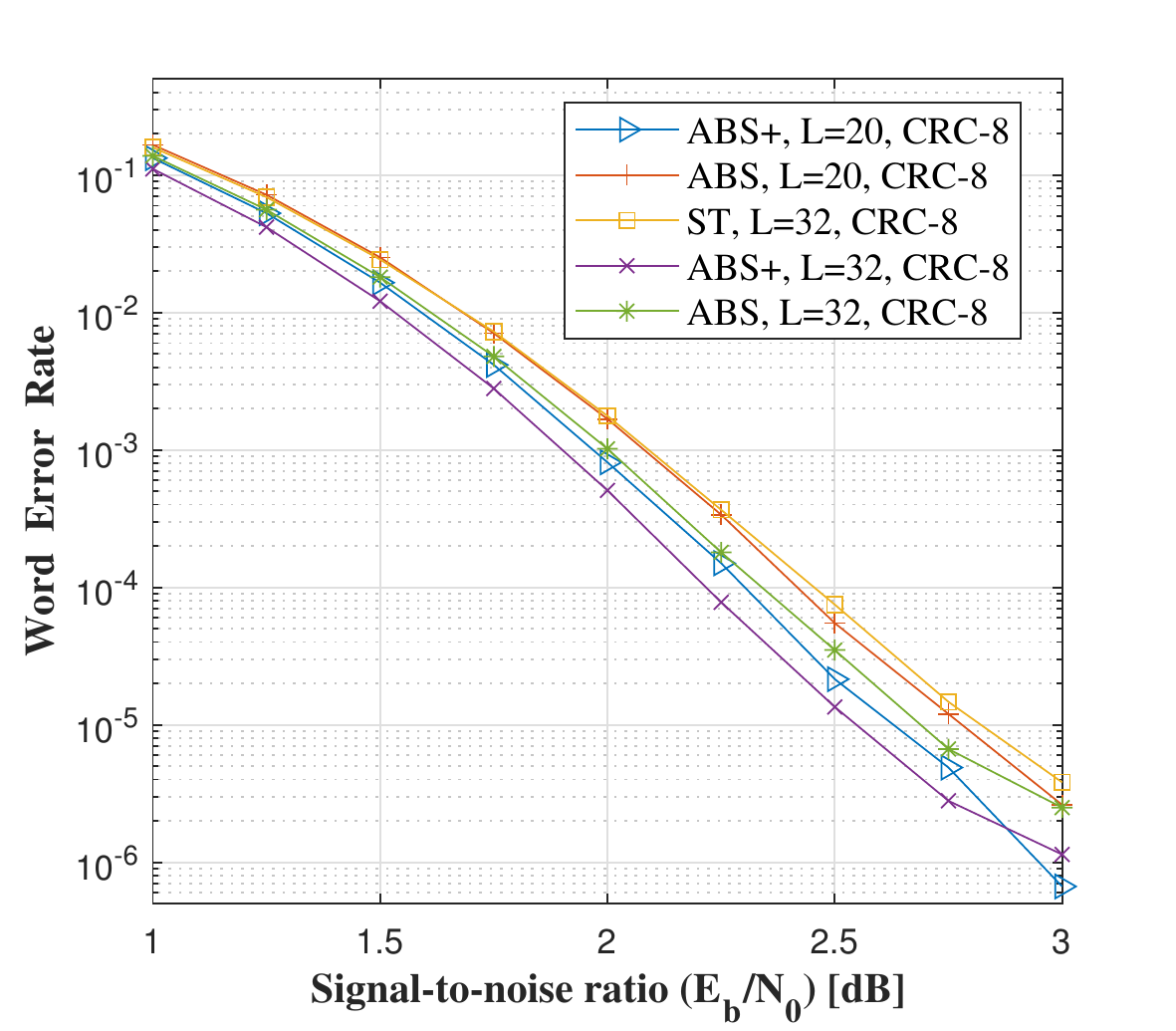}
        \caption{length 512, dimension 256}
    \end{subfigure}
    ~\hspace*{0.2in}
    \begin{subfigure}{0.41\linewidth}
        \centering
        \includegraphics[width=\linewidth]{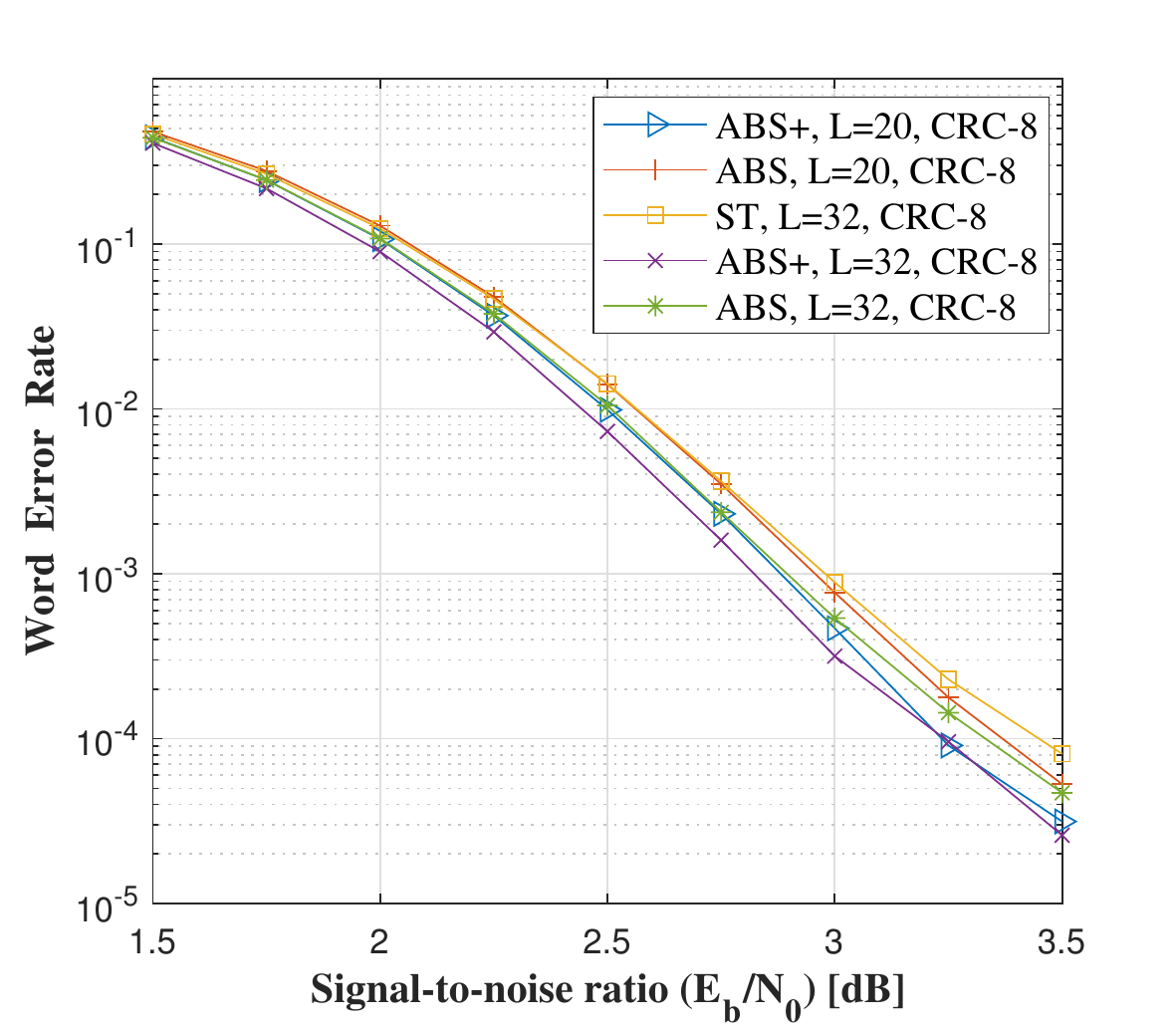}
        \caption{length 512, dimension 358}
    \end{subfigure}

    \caption{Performance of ABS+ polar codes, ABS polar codes, and standard polar codes over the binary-input AWGN channel.
        The legend ``ST" refers to standard polar codes.
        The CRC length is chosen from the set $\{4,8,12,16,20\}$ to minimize the decoding error probability.
        The parameter $L$ is the list size.
        For standard polar codes, we always choose $L=32$.
        For ABS+ and ABS polar codes, we test two different list sizes $L=20$ and $L=32$.}
    \label{fig:cp1}
\end{figure*}

\begin{figure*}
    \centering
    \begin{subfigure}{0.41\linewidth}
        \centering
        \includegraphics[width=\linewidth]{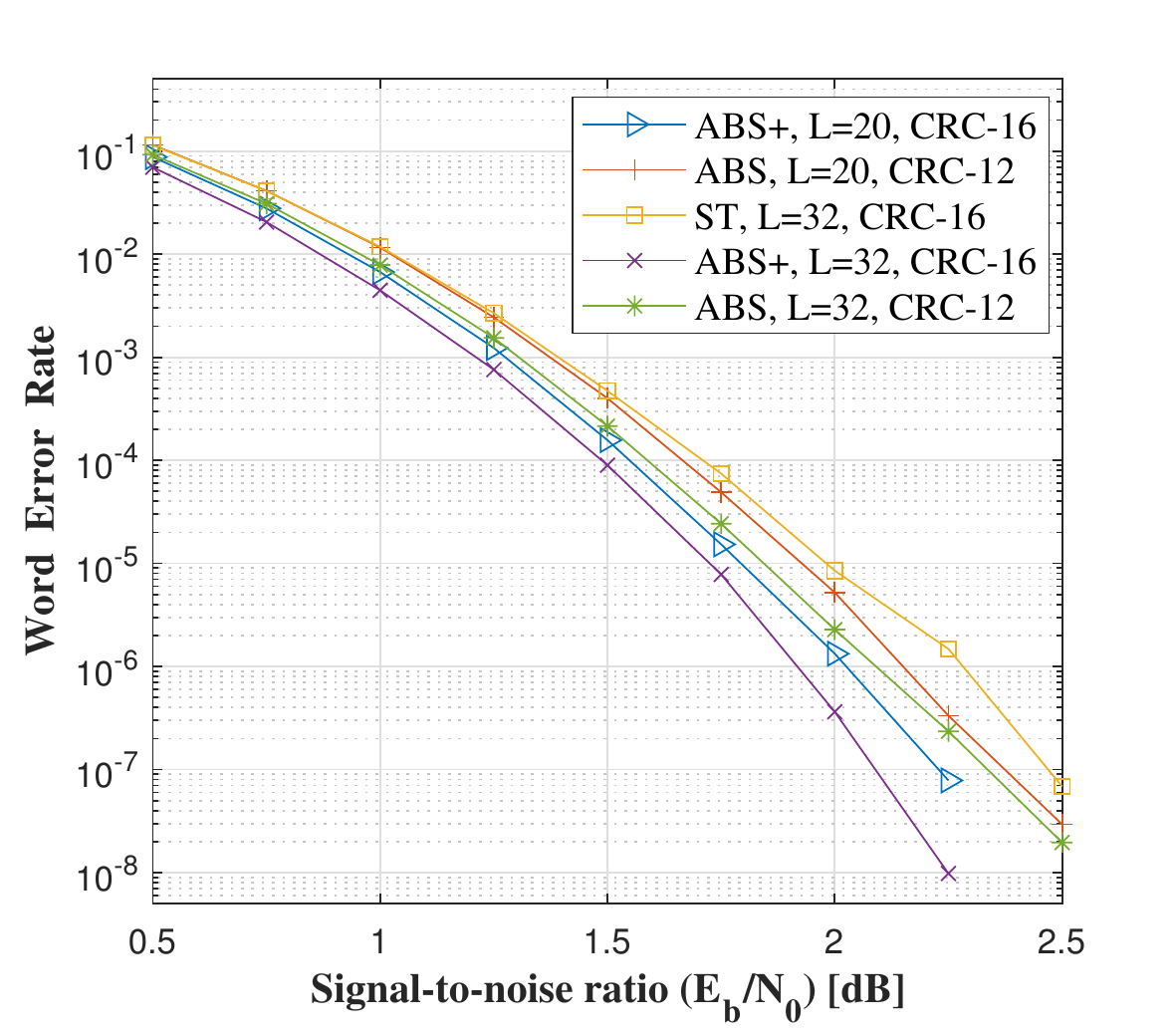}
        \caption{length 1024, dimension 307}
    \end{subfigure}
    ~\hspace*{0.2in}
    \begin{subfigure}{0.41\linewidth}
        \centering
        \includegraphics[width=\linewidth]{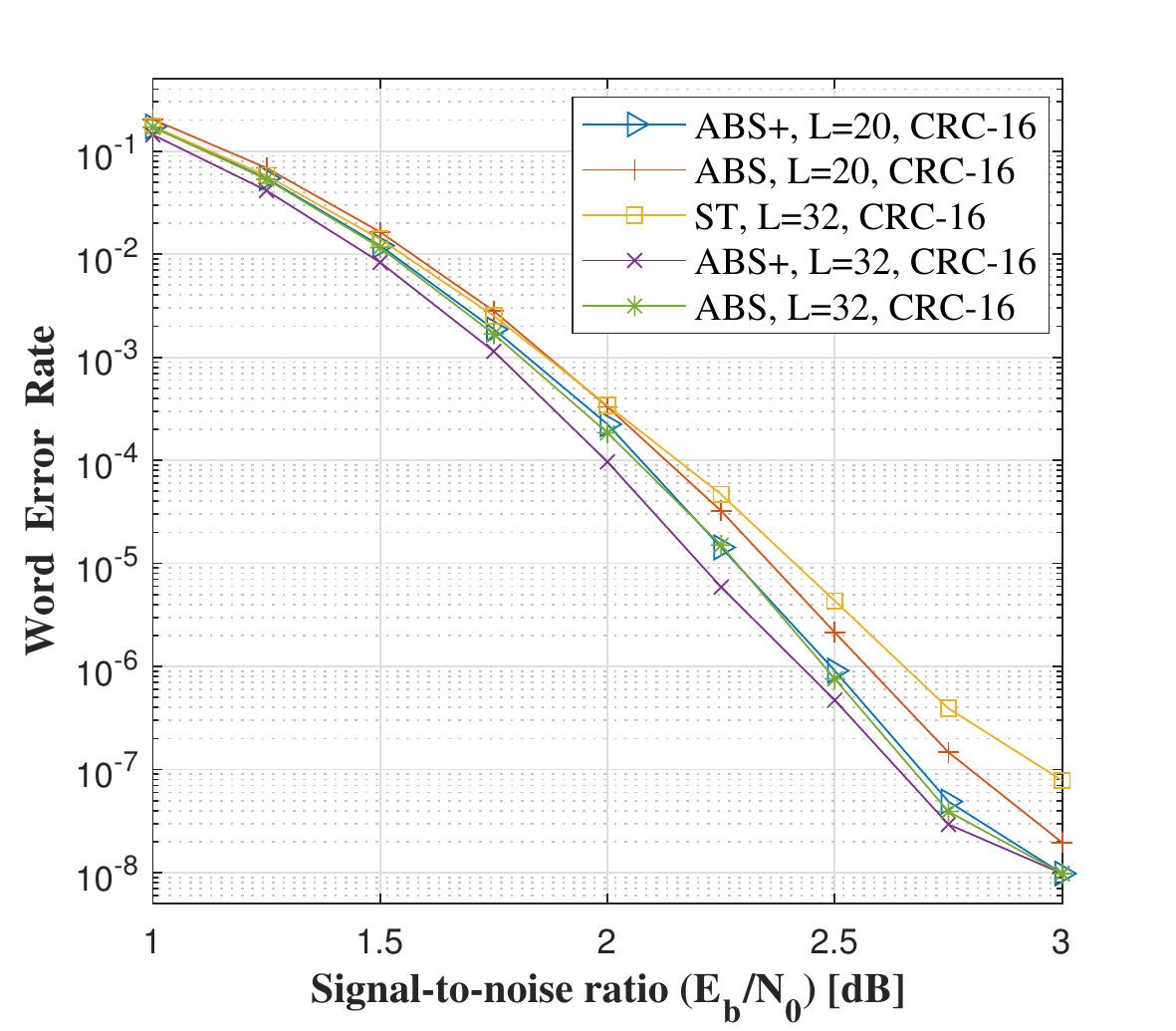}
        \caption{length 1024, dimension 512}
    \end{subfigure}

    \vspace*{0.1in}

    \begin{subfigure}{0.41\linewidth}
        \centering
        \includegraphics[width=\linewidth]{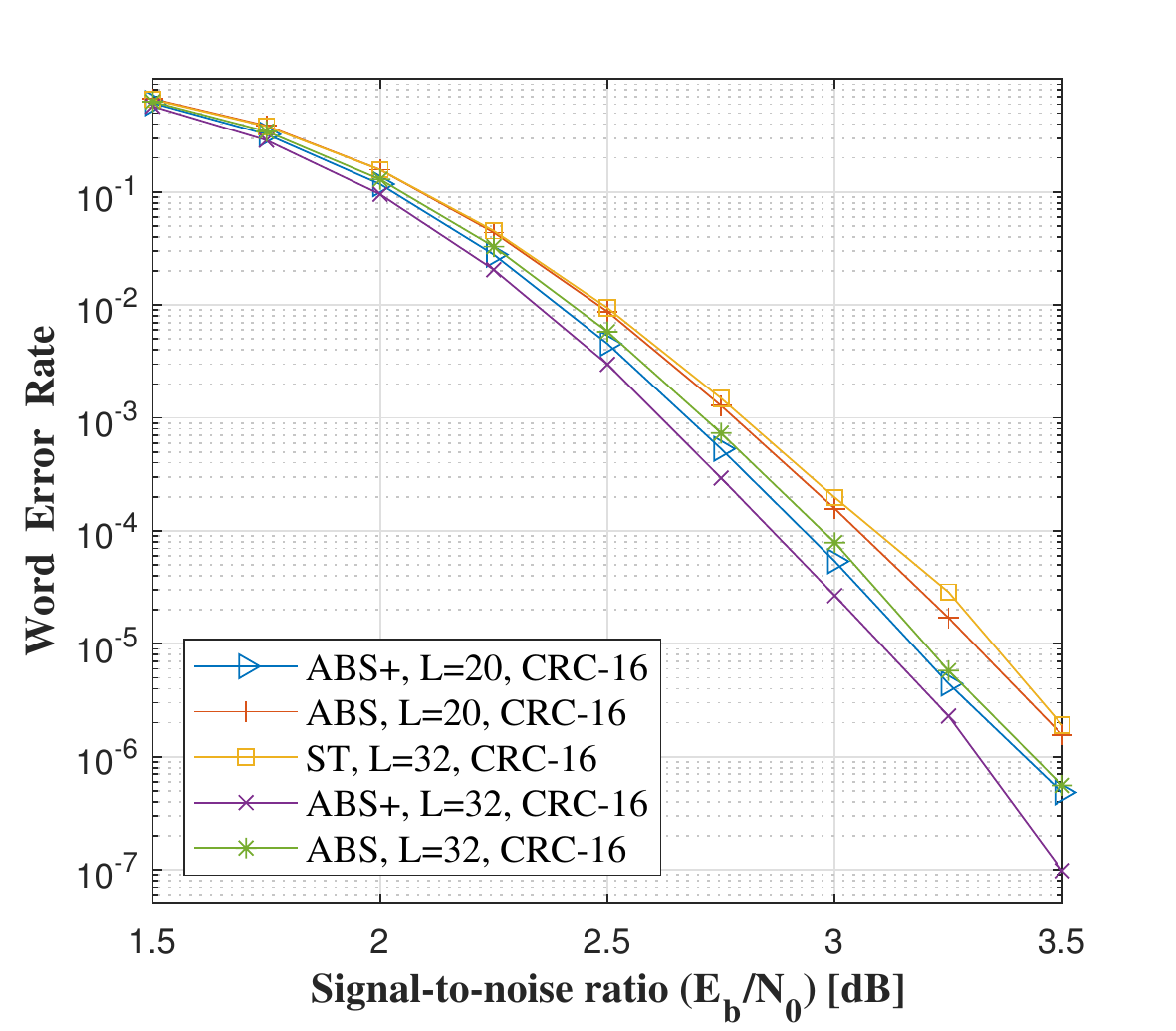}
        \caption{length 1024, dimension 717}
    \end{subfigure}
    ~\hspace*{0.2in}
    \begin{subfigure}{0.41\linewidth}
        \centering
        \includegraphics[width=\linewidth]{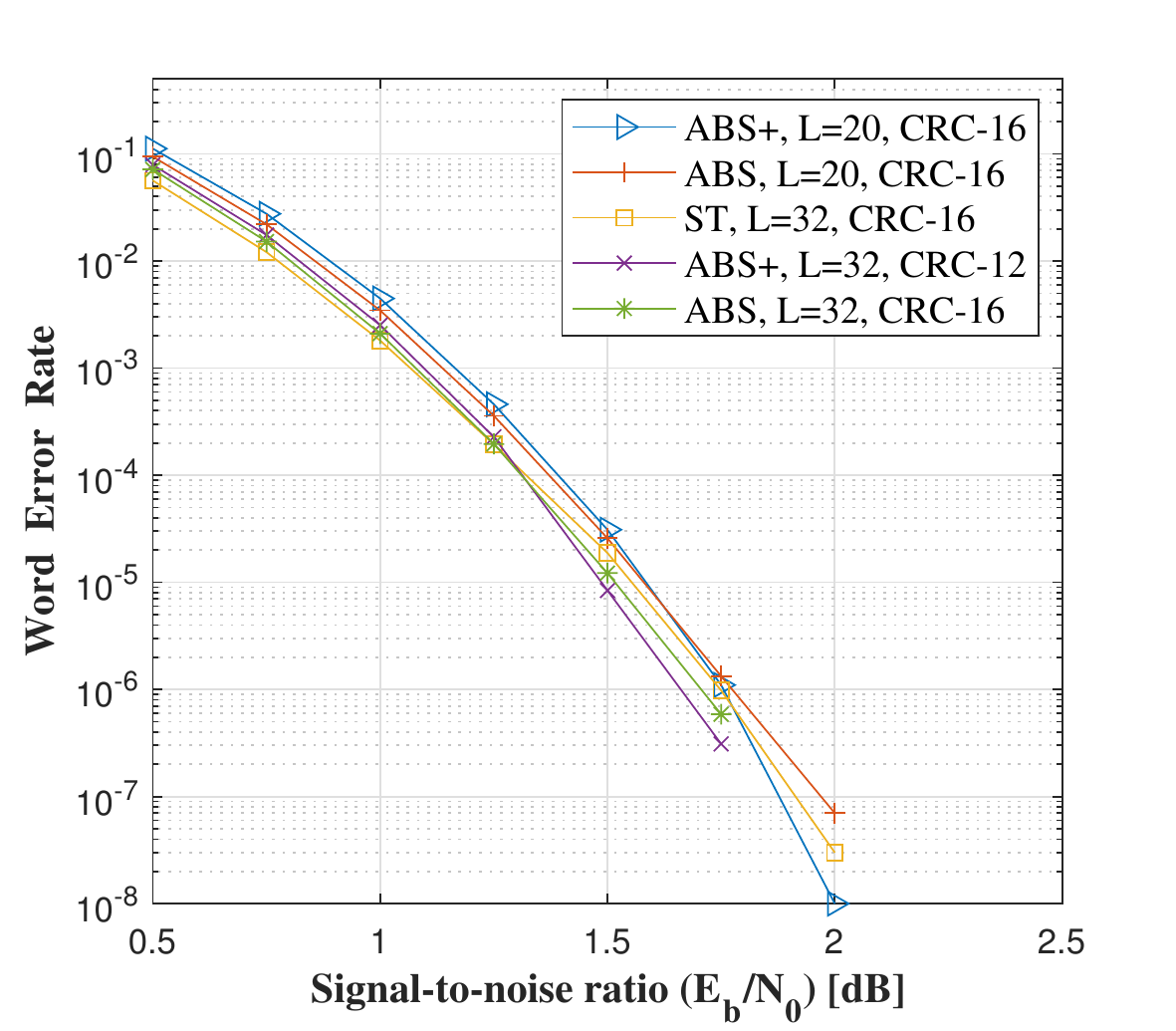}
        \caption{length 2048, dimension 614}
    \end{subfigure}

    \vspace*{0.1in}

    \begin{subfigure}{0.41\linewidth}
        \centering
        \includegraphics[width=\linewidth]{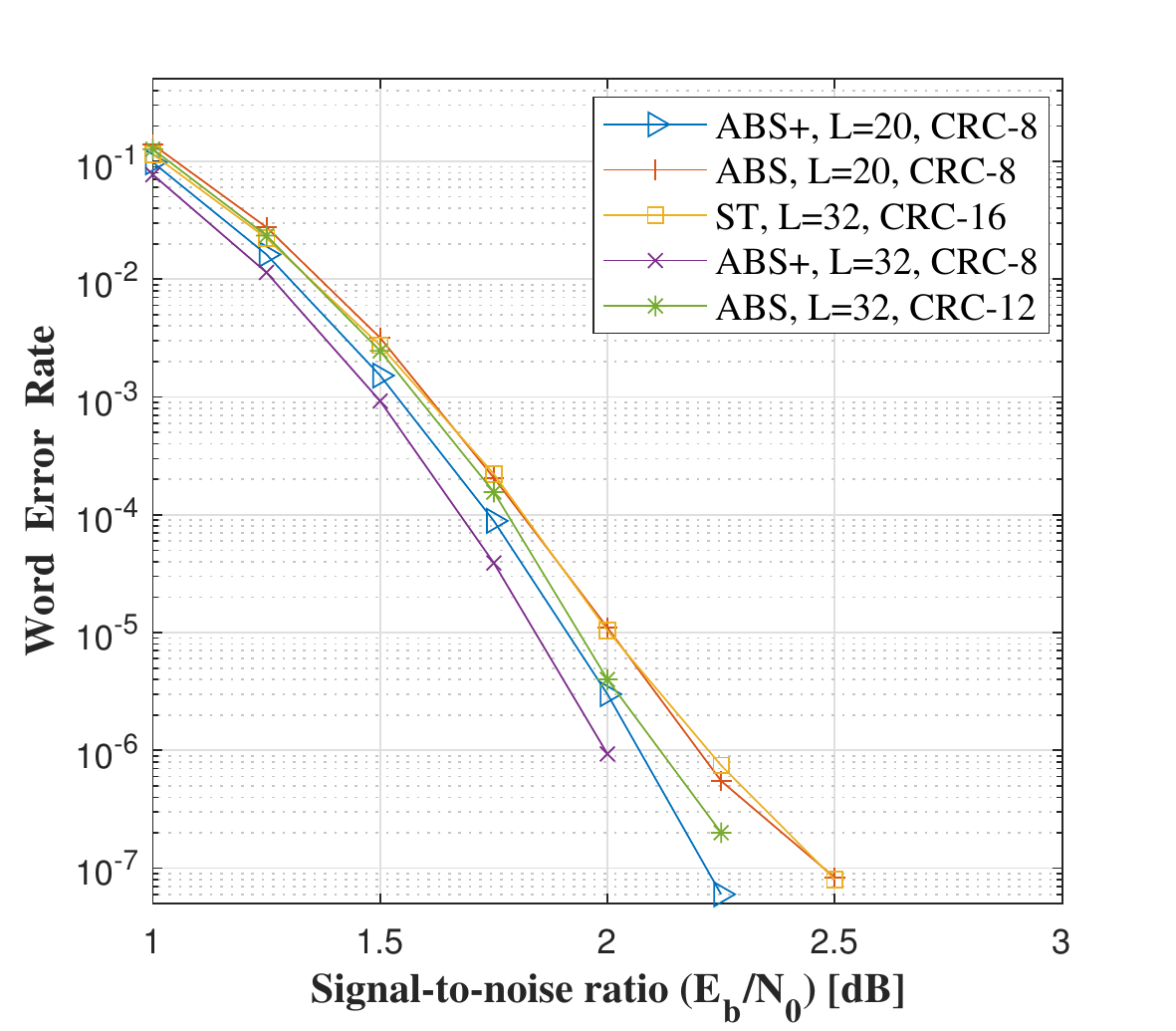}
        \caption{length 2048, dimension 1024}
    \end{subfigure}
    ~\hspace*{0.2in}
    \begin{subfigure}{0.41\linewidth}
        \centering
        \includegraphics[width=\linewidth]{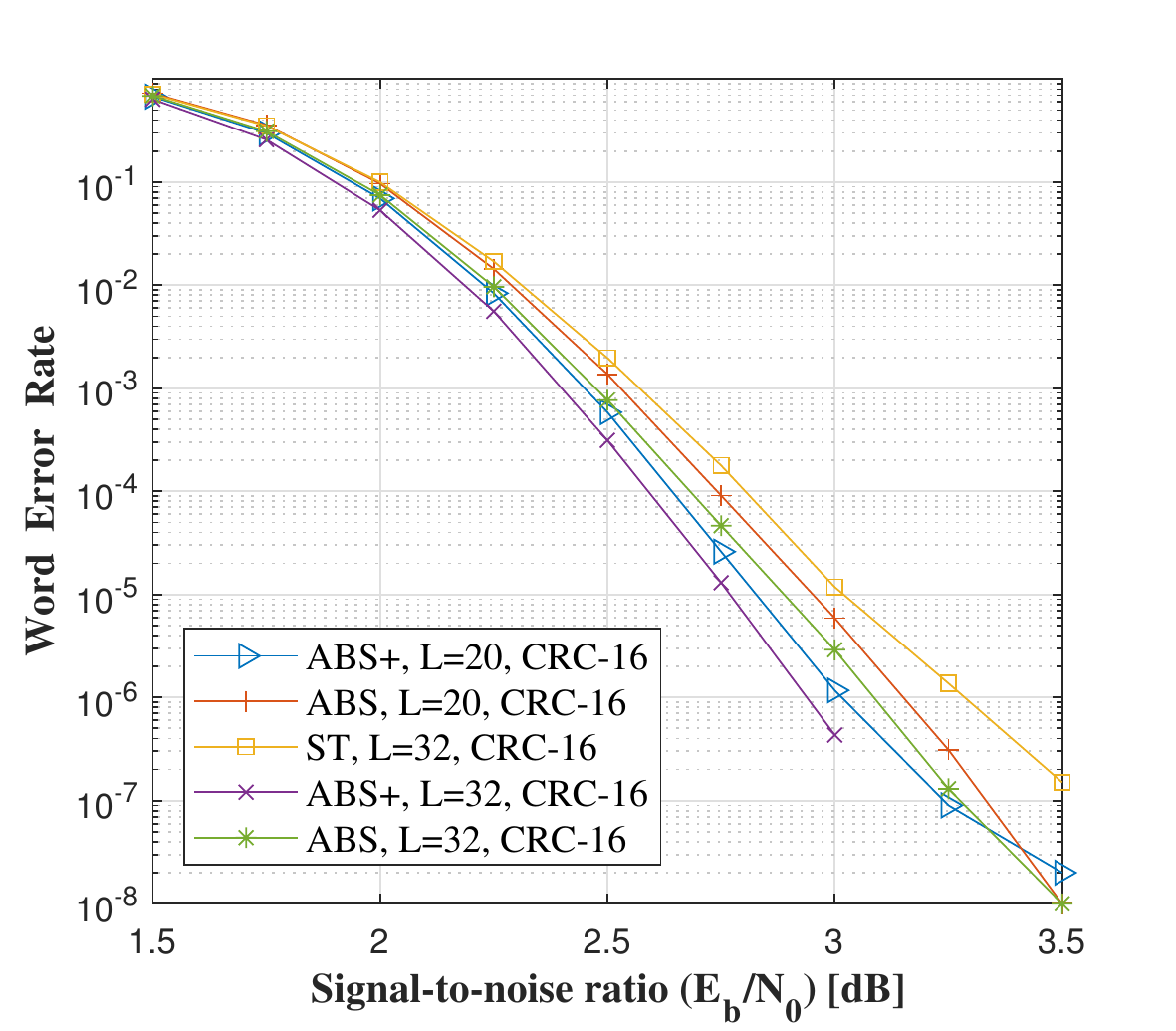}
        \caption{length 2048, dimension 1434}
    \end{subfigure}

    \caption{Performance of ABS+ polar codes, ABS polar codes, and standard polar codes over the binary-input AWGN channel.
        The legend ``ST" refers to standard polar codes.
        The CRC length is chosen from the set $\{4,8,12,16,20\}$ to minimize the decoding error probability.
        The parameter $L$ is the list size.
        For standard polar codes, we always choose $L=32$.
        For ABS+ and ABS polar codes, we test two different list sizes $L=20$ and $L=32$.}
    \label{fig:cp2}
\end{figure*}

\section{Simulation results} \label{sect:simu}

We conduct extensive simulations over binary-input AWGN channels to compare the performance of ABS+ polar codes, ABS polar codes, and standard polar codes.
We run simulations for $4$ different choices of code length $256,512,\linebreak1024,2048$.
For each choice of code length, we test $3$ different code rates $0.3, 0.5$, and $0.7$.
The comparison of decoding error probability is given in Fig.~\ref{fig:cp1} and Fig.~\ref{fig:cp2}.
Specifically, Fig.~\ref{fig:cp1} contains the plots for code length $256$ and $512$; Fig.~\ref{fig:cp2} contains the plots for code length $1024$ and $2048$.
The comparison of decoding time is given in Table~\ref{tb:time}.

In Fig.~\ref{fig:cp1}--\ref{fig:cp2} and Table~\ref{tb:time}, for each choice of code length and code dimension, we compare the performance of the following $5$ decoders. {\bf (1) ST, $L=32$}: SCL decoder for standard polar codes with list size $32$ and optimal CRC length; {\bf (2) ABS, $L=20$}: SCL decoder for ABS polar codes with list size $20$ and optimal CRC length; {\bf (3) ABS, $L=32$}: SCL decoder for ABS polar codes with list size $32$ and optimal CRC length; {\bf (4) ABS+, $L=20$}: SCL decoder for ABS+ polar codes with list size $20$ and optimal CRC length; {\bf (5) ABS+, $L=32$}: SCL decoder for ABS+ polar codes with list size $32$ and optimal CRC length.
The optimal CRC length is chosen from the set $\{4,8,12,16,20\}$ to minimize the decoding error probability.

From Table~\ref{tb:time} we can see that decoders (1),(2),(4) have more or less the same running time; decoders (3),(5) have more or less the same running time. Moreover, the decoding time of decoders (3),(5) is longer than that of decoders (1),(2),(4) by roughly $60\%$. Note that if we set the list size in the SCL decoder to be the same, then the decoding time of ABS+ polar codes and ABS polar codes is very close to each other.

As for the decoding error probability, we mainly compare the performance of decoders (1), (2), and (4) because they have similar decoding time. As we can see from Fig.~\ref{fig:cp1}--\ref{fig:cp2}, ABS+ polar codes with list size $20$ improves upon ABS polar codes with list size $20$ by $0.1\dB$--$0.25\dB$; ABS+ polar codes with list size $20$ improves upon standard polar codes with list size $32$ by $0.15\dB$--$0.35\dB$.
For instance, in Fig.~\ref{fig:cp2} (e) for the case of length $n = 2048$ and dimension $k = 1024$, we can observe that when the word error rate (WER) is $10^{-7}$ and the list size is 20, the signal-to-noise ratio requirement for ABS+ polar code is approximately 0.25 dB lower than that of the ABS polar code.
Moreover, in Fig.~\ref{fig:cp1}(b) for the case of length $n = 512$ and dimension $k = 154$, we can see that when the WER is $2\times 10^{-5}$, the signal-to-noise ratio requirement for ABS+ polar code with list size $20$ is approximately 0.35 dB lower than that of the standard polar code with list size $32$.
Finally, if we set the list size to be $32$ for both ABS+ and standard polar codes, then ABS+ polar codes demonstrate $0.2\dB$--$0.45\dB$ improvement over standard polar codes.

As a final remark, the implementations of all the algorithms in this paper are available at the website
\texttt{https://\linebreak github.com/PlumJelly/ABS-Polar}

\newpage
\appendices
\section{Space-efficient decoding algorithms}\label{sect:space}

\begin{algorithm}[h]
    \DontPrintSemicolon
    \caption{\texttt{decode\_swapped\_channel$(n_c, i)$}\\ \hspace*{1.8in}space-efficient version}
    \label{algo:decode_swp_channel_space_efficient}
    \KwIn{$n_c\in \{2,4,8,\dots, n\}$ and index $i$ satisfying $2i\in \cI_S^{(2n_c)}$.}


    \For{\em $\beta \in \{1,2,\dots,n/(2n_c)\}, a\in\{0,1\}$ and $b\in\{0,1\}$}{
        $\tP_{2n_c}[\beta][a,b] \gets $
        \texttt{calculate\_probability}$(n_c, i, \beta, \swpa, a,b)$
    }

    \texttt{decode\_channel}$(2n_c, 2i-1)$


    \For{$\beta \in\{1,2,\dots, n/(2n_c)\}$}{
        $\tB_{n_c}[\beta]\gets \tB_{2n_c}[\beta]$
    }

    \For{\em $\beta \in \{1,2,\dots,n/(2n_c)\}, a\in\{0,1\}$ and $b\in\{0,1\}$}{
        $\tP_{2n_c}[\beta][a,b] \gets $

        \texttt{calculate\_probability}$(n_c, i, \beta, \swpb, a,b)$
    }

    \texttt{decode\_channel}$(2n_c, 2i)$


    \For{$\beta \in\{1,2,\dots, n/(2n_c)\}$}{
        $\tB_{n_c}[\beta + n/(2n_c)]\gets \tB_{2n_c}[\beta]$
    }

    \For{\em $\beta \in \{1,2,\dots,n/(2n_c)\}, a\in\{0,1\}$ and $b\in\{0,1\}$}{
        $\tP_{2n_c}[\beta][a,b] \gets $

        \texttt{calculate\_probability}$(n_c, i, \beta, \swpc, a,b)$
    }

    \texttt{decode\_channel}$(2n_c, 2i+1)$

    \eIf{$i\le n_c-2$}{
        \Comment{Only decode one bit ${X}_{\beta + (i-1)n/n_c}^{(n_c)}$ for each $\beta$}

        \For{$\beta \in \{1,2,\dots, n/(2n_c)\}$}{
            $\beta' \gets \beta + n/(2n_c)$ \\[0.4em]

            $\tB_{n_c}[\beta ] \gets \tB_{n_c}[\beta] + \tB_{2n_c}[\beta]$\\[0.4em]

            $\tB_{n_c}[\beta'] \gets                    \tB_{2n_c}[\beta]$

            \Comment{See Fig.~\ref{fig:swp} for an explanation}
        }
    }{
        \Comment{Decode two bits ${X}_{\beta + (n_c-2)n/n_c}^{(n_c)}, {X}_{\beta + (n_c-1)n/n_c}^{(n_c)}$ for each $\beta$}

        \For{$\beta \in \{1,2,\dots, n/(2n_c)\}$}{
            $\beta' \gets \beta + n/(2n_c)$\\[0.4em]

            $\tB_{n_c}[\beta ] \gets \tB_{n_c}[\beta ] + \tB_{2n_c}[\beta]$\\[0.4em]

            $\tH_{n_c}[\beta ] \gets \tB_{n_c}[\beta'] + \tH_{2n_c}[\beta]$\\[0.4em]

            $\tB_{n_c}[\beta'] \gets                    \tB_{2n_c}[\beta]$\\[0.4em]

            $\tH_{n_c}[\beta'] \gets \tH_{2n_c}[\beta]$

            \Comment{See Fig.~\ref{fig:swp} for an explanation}
        }
    }

    \Return
\end{algorithm}

\begin{algorithm}[h]
    \DontPrintSemicolon
    \caption{\texttt{decode\_added\_channel$(n_c, i)$}\\ \hspace*{1.65in}space-efficient version}
    \label{algo:decode_add_channel_space_efficient}
    \KwIn{$n_c\in \{2,4,8,\dots, n\}$ and index $i$ satisfying $2i\in \cI_A^{(2n_c)}$}


    \For{\em $\beta \in \{1,2,\dots,n/(2n_c)\}, a\in\{0,1\}$ and $b\in\{0,1\}$}{
        $\tP_{2n_c}[\beta][a,b] \gets $

        \texttt{calculate\_probability}$(n_c, i, \beta, \adda, a,b)$
    }

    \texttt{decode\_channel}$(2n_c, 2i-1)$


    \For{$\beta \in\{1,2,\dots, n/(2n_c)\}$}{
        $\tB_{n_c}[\beta]\gets \tB_{2n_c}[\beta]$
    }

    \For{\em $\beta \in \{1,2,\dots,n/(2n_c)\}, a\in\{0,1\}$ and $b\in\{0,1\}$}{
        $\tP_{2n_c}[\beta][a,b] \gets $

        \texttt{calculate\_probability}$(n_c, i, \beta, \addb, a,b)$
    }

    \texttt{decode\_channel}$(2n_c, 2i)$


    \For{$\beta \in\{1,2,\dots, n/(2n_c)\}$}{
        $\tB_{n_c}[\beta + n/(2n_c)]\gets \tB_{2n_c}[\beta]$
    }

    \For{\em $\beta \in \{1,2,\dots,n/(2n_c)\}, a\in\{0,1\}$ and $b\in\{0,1\}$}{
        $\tP_{2n_c}[\beta][a,b] \gets $

        \texttt{calculate\_probability}$(n_c, i, \beta, \addc, a,b)$
    }

    \texttt{decode\_channel}$(2n_c, 2i+1)$

    \eIf{$i\le n_c-2$}{
        \Comment{Only decode one bit $X_{\beta + (i-1)n/n_c}^{(n_c)}$ for each $\beta$}

        \For{$\beta\in\{ 1, 2, \dots, n/(2n_c)\}$}{
            $\beta' \gets \beta + n/(2n_c)$ \\[0.4em]

            $\tB_{n_c}[\beta ] \gets \tB_{n_c}[\beta] + \tB_{n_c}[\beta'] + \tB_{2n_c}[\beta]$ \\[0.4em]

            $\tB_{n_c}[\beta'] \gets                    \tB_{n_c}[\beta'] + \tB_{2n_c}[\beta]$

            \Comment{See Fig.~\ref{fig:add} for an explanation}
        }

    }{
        \Comment{Decode two bits ${X}_{\beta + (n_c-2)n/n_c}^{(n_c)}, {X}_{\beta + (n_c-1)n/n_c}^{(n_c)}$ for each $\beta$}

        \For{$\beta\in\{ 1, 2, \dots, n/(2n_c)\}$}{
            $\beta'\gets \beta + n/(2n_c)$ \\[0.4em]

            $\tB_{n_c}[\beta ] \gets \tB_{n_c}[\beta] + \tB_{n_c}[\beta'] + \tB_{2n_c}[\beta]$ \\[0.4em]

            $\tB_{n_c}[\beta'] \gets                    \tB_{n_c}[\beta'] + \tB_{2n_c}[\beta]$\\[0.4em]

            $\tH_{n_c}[\beta ] \gets \tB_{2n_c}[\beta] + \tH_{2n_c}[\beta]$ \\[0.4em]

            $\tH_{n_c}[\beta'] \gets \tH_{2n_c}[\beta]$

            \Comment{See Fig.~\ref{fig:add} for an explanation}
        }
    }
    \Return

\end{algorithm}

\begin{algorithm}[h]
    \DontPrintSemicolon
    \caption{\texttt{decode\_original\_channel$(n_c, i)$}\\\hspace*{1.825in}space-efficient version}
    \label{algo:decode_ori_channel_space_efficient}
    \KwIn{$n_c\in \{2,4,8,\dots, n\}$ and index $i$  satisfying $1\le i \le n_c-1$ and $2i\notin \cI^{(2n_c)}$}

    \eIf{\em $i = 1$ or $2(i-1)\notin \cI^{(2n_c)}$}{

        \For{\em $\beta \in \{1,2,\dots,n/(2n_c)\}, a\in\{0,1\}$ and $b\in\{0,1\}$}{
            $\tP_{2n_c}[\beta][a,b] \gets$ \texttt{calculate\_probability}$(n_c, i, \beta, \oria, a,b)$
        }

        \texttt{decode\_channel}$(2n_c, 2i-1)$

        \For{$\beta \in\{1,2,\dots, n/(2n_c)\}$}{
            $\tB_{n_c}[\beta]\gets \tB_{2n_c}[\beta]$
        }
    }{
        \For{$\beta \in \{1,2,\dots, n/(2n_c)\}$}
        {
            $\tB_{n_c}[\beta] \gets \tH_{2n_c}[\beta]$
        }
    }


    \For{\em $\beta \in \{1,2,\dots,n/(2n_c)\}, a\in\{0,1\}$ and $b\in\{0,1\}$}{
        $\tP_{2n_c}[\beta][a,b] \gets $

        \hspace*{.2in}\texttt{calculate\_probability}$(n_c, i, \beta, \orib, a,b)$
    }

    \texttt{decode\_channel}$(2n_c, 2i)$

    \For{$\beta \in \{1,2,\dots, n/(2n_c)\}$}
    {
        $\tB_{n_c}[\beta + n/(2n_c)] \gets \tB_{2n_c}[\beta]$
    }

    \eIf{$i \le n_c-2$}{
        \Comment{Only decode one bit ${X}_{\beta + (i-1)n/n_c}^{(n_c)}$ for each $\beta$}

        \For{$ \beta\in\{1, 2, \dots, n/(2n_c)\}$}{
            $\tB_{n_c}[\beta + n/(2n_c)] \gets \tB_{n_c}[\beta] + \tB_{n_c}[\beta + n/(2n_c)]$
        }
    }{
        \Comment{Decode two bits ${X}_{\beta + (n_c-2)n/n_c}^{(n_c)}, {X}_{\beta + (n_c-1)n/n_c}^{(n_c)}$ for each $\beta$}

        \For{\em $\beta \in \{1,2,\dots,n/(2n_c)\}, a\in\{0,1\}$ and $b\in\{0,1\}$}{
            $\tP_{2n_c}[\beta][a,b] \gets $

            \texttt{calculate\_probability}$(n_c, i, \beta, \oric, a,b)$
        }

        \texttt{decode\_channel}$(2n_c, 2i+1)$

        \For{$ \beta\in\{1, 2, \dots, n/(2n_c)\}$}{
            $\beta' \gets \beta + n/(2n_c)$ \\[0.4em]

            $\tB_{n_c}[\beta] \gets \tB_{n_c}[\beta] + \tB_{n_c}[\beta']$ \\[0.4em]

            $\tH_{n_c}[\beta ]  \gets \tB_{2n_c}[\beta] + \tH_{2n_c}[\beta]$  \\[0.4em]

            $\tH_{n_c}[\beta']  \gets                     \tH_{2n_c}[\beta]$
        }
    }

    \Return

\end{algorithm}

\clearpage
\balance

\end{document}